\documentclass[a4paper]{article}

\usepackage[english]{babel}
\usepackage[utf8x]{inputenc}
\usepackage[T1]{fontenc}

\usepackage[a4paper,top=3cm,bottom=2cm,left=3cm,right=3cm,marginparwidth=1.75cm]{geometry}

\usepackage{dsfont}
\usepackage{amsthm}
\usepackage{amsmath}
\usepackage{amssymb}
\usepackage{graphicx}    
\usepackage[colorinlistoftodos]{todonotes}
\usepackage[colorlinks=true, allcolors=blue]{hyperref}
\usepackage{calrsfs} 
\usepackage{cancel} 
\usepackage[normalem]{ulem}
\usepackage{bbold}						
\DeclareMathAlphabet{\mathcal}{OMS}{cmsy}{m}{n} 
\usepackage{stmaryrd}					
\usepackage{bm}
\usepackage{fancyvrb}
\usepackage[all]{xy}					 
\usepackage{authblk}			

\newcommand{\ba}{\textbf{a}}
\newcommand{\bA}{\textbf{A}}
\newcommand{\bb}{\textbf{b}}
\newcommand{\bB}{\textbf{B}}
\newcommand{\bc}{\textbf{c}}

\newcommand{\bH}{\textbf{H}}
\newcommand{\bx}{\textbf{x}}

\newcommand{\br}{\textbf{r}}

\newcommand{\bJ}{\textbf{J}}

\newcommand{\bW}{\textbf{W}}
\newcommand{\bD}{\textbf{D}}
\newcommand{\bC}{\textbf{C}}
\newcommand{\be}{\textbf{e}}
\newcommand{\bu}{\textbf{u}}
\newcommand{\bv}{\textbf{v}}
\newcommand{\bj}{\textbf{j}}

\newcommand{\bp}{\textbf{p}}
\newcommand{\bP}{\textbf{P}}
\newcommand{\bM}{\textbf{M}}
\newcommand{\bh}{\textbf{h}}
\newcommand{\bs}{\textbf{s}}
\newcommand{\bL}{\textbf{L}}
\newcommand{\bk}{\textbf{k}}
\newcommand{\bq}{\textbf{q}}
\newcommand{\bS}{\textbf{S}}
\newcommand{\bU}{\textbf{U}}
\newcommand{\bV}{\textbf{V}}
\newcommand{\bR}{\textbf{R}}

\newcommand{\boldeta}{\boldsymbol{\eta}}
\newcommand{\bLambda}{\boldsymbol{\Lambda}}
\newcommand{\bmu}{\boldsymbol{\mu}}
\newcommand{\bphi}{\boldsymbol{\phi}}
\newcommand{\bpsi}{\boldsymbol{\psi}}
\newcommand{\bomega}{\boldsymbol{\omega}}

\newcommand{\balpha}{\boldsymbol{\alpha}}
\newcommand{\bbeta}{\boldsymbol{\beta}}
\newcommand{\blambda}{\boldsymbol{\lambda}}

\newcommand{\bgamma}{\boldsymbol{\gamma}}
\newcommand{\bzeta}{\boldsymbol{\zeta}}

\newcommand{\mbd}{\mathbb{d}}
\newcommand{\rmd}{\mathrm{d}}
\newcommand{\mbF}{\mathbb{F}}
\newcommand{\mbJ}{\mathbb{J}}
\newcommand{\mbR}{\mathbb{R}}

\newcommand{\mbN}{\mathbb{N}}
\newcommand{\mbL}{\mathbb{L}}

\newcommand{\mbH}{\mathbb{H}}

\newcommand{\mbT}{\mathbb{T}}

\newcommand{\mbA}{\mathbb{A}}

\newcommand{\mcF}{\mathcal{F}}
\newcommand{\mcM}{\mathcal{M}}
\newcommand{\mcN}{\mathcal{N}}

\newcommand{\mcD}{\mathcal{D}}
\newcommand{\mcP}{\mathcal{P}}

\newcommand{\mcB}{\mathcal{B}}

\newcommand{\mcA}{\mathcal{A}}

\newcommand{\mcS}{\mathcal{S}}
\newcommand{\mcC}{\mathcal{C}}
\newcommand{\mcT}{\mathcal{T}}
\newcommand{\mcL}{\mathcal{L}}

\newcommand{\bbphi}{\textbf{b}_{\bphi}}

\newcommand{\bvss}{\bv_{\rm ss}}    
\newcommand{\bJss}{\bJ_{\rm ss}}
\newcommand{\Jss}{J_{\rm ss}}

\newcommand{\Pss}{P_{\rm ss}}
\newcommand{\tPss}{\tilde{P}_{\rm ss}}

\newcommand{\bdelta}{\boldsymbol{\delta}}
\newcommand{\rhoss}{\rho_{\rm ss}}    

\newcommand{\IM}{\mathrm{Im}}

\newcommand{\ie}{\textit{i.e. }}
\newcommand{\eg}{\textit{e.g. }}
\newcommand{\via}{\textit{via }}

\newcommand{\DDiv}{\mathbb{d}\mathrm{iv}}

\newcommand{\tbphi}{\delta\tilde{\bphi}}
\newcommand{\tba}{\widetilde{\ba}}
\newcommand{\tbh}{\widetilde{\bh}}
\newcommand{\tblambda}{\widetilde{\blambda}}
\newcommand{\tbb}{\widetilde{\bb}}
\newcommand{\tbs}{\widetilde{\bs}}
\newcommand{\teta}{\widetilde{\eta}}
\newcommand{\tbD}{\widetilde{\bD}}

\newcommand{\whomega}{\widehat{\bomega}}
\newcommand{\whSigma}{\widehat{\Sigma}}
\newcommand{\hrho}{\hat{\rho}}
\newcommand{\hpsi}{\hat{\psi}}
\newcommand{\halpha}{\hat{\alpha}}
\newcommand{\hbeta}{\hat{\beta}}
\newcommand{\hbj}{\hat{\bj}}

\newcommand{\bbb}{\mathfrak{b}}
\newcommand{\ant}{\bk}

\newcommand{\whgamma}{\widehat{\bgamma}}

\newcommand{\mkh}{\mathfrak{h}}
\newcommand{\mks}{\mathfrak{s}}
\newcommand{\mkm}{\mathfrak{m}}
\newcommand{\mke}{\mathfrak{e}}
\newcommand{\mkt}{\mathfrak{t}}

\newcommand{\pmcA}{\widehat{\mcA}}
\newcommand{\pmcS}{\widehat{\mcS}}
\newcommand{\mmcA}{\bar{\mcA}}
\newcommand{\mmcS}{\bar{\mcS}}

\newcommand{\llangle}{\left\langle}
\newcommand{\rrangle}{\right\rangle}

\newcommand{\Det}{\mathrm{det}}
\newcommand{\wt}[1]{\widetilde{#1}}

\newcommand{\barev}{\ba^{\rm rev}}
\newcommand{\bairrev}{\ba^{\rm irrev}}

\definecolor{monVert}{RGB}{0,200,0}

\definecolor{monOrange}{RGB}{250,150,0}

\definecolor{lavande}{RGB}{170,150,255}

\newcommand{\Magenta}[1]{{\leavevmode\color{magenta}{#1}}}

\definecolor{myPink}{RGB}{255,153,204}
\newcommand{\MyPink}[1]{{\leavevmode\color{myPink}{#1}}}
\definecolor{monVioletrose}{RGB}{100,0,255}
\definecolor{grayy}{RGB}{128,128,128}

\definecolor{colorL}{RGB}{255,50,127}
\definecolor{colorUp}{RGB}{0,128,204}
\definecolor{colorR}{RGB}{178,102,255}

\newtheorem{Lemme}{Lemma}

\title{Geometric theory of (extended) time-reversal symmetries in stochastic processes -- Part II: field theory}

\author[1,2]{J. O'Byrne}
\author[1]{M.E. Cates}

\affil[1]{DAMTP, Centre for Mathematical Sciences, University of Cambridge, Wilberforce Road, Cambridge CB3 0WA, United Kingdom}
\affil[2]{Laboratoire Jean Perrin, Sorbonne Universit\'e, 4 Place Jussieu, Paris 75005, France}

\date{}

\begin{document}
\maketitle

\begin{abstract}
In this article, we study the time-reversal properties of a generic Markovian stochastic field dynamics with Gaussian noise. 
We introduce a convenient functional geometric formalism that allows us to straightforwardly generalize known results from finite dimensional systems to the case of continuous fields.
We give, at field level, full reversibility conditions for three notions of time-reversal defined in the first article of this two-part series, namely T-, MT-, and EMT-reversibility. When the noise correlator is invertible, these reversibility conditions do not make reference to any generically unknown function like the stationary probability, and can thus be verified systematically.
Focusing on the simplest of these notions, where only the time variable is flipped upon time reversal, we show that time-reversal symmetry  breaking is quantified by a single geometric object: the vorticity two-form, which is a two-form over the functional space $\mbF$ to which the field belongs. Reversibility then amounts to the cancellation of this vorticity two-form. This condition applies at distributional level and can thus be difficult to use in practice. For fields that are defined on a spatial domain of dimension $d=1$, we overcome this problem by building a basis of the space of two-forms $\Omega^2(\mbF)$. Reversibility is then equivalent to the vanishing of the vorticity's coordinates in this basis, a criterion that is readily applicable to concrete examples. Furthermore, we show that this basis provides a natural direct-sum decomposition of $\Omega^2(\mbF)$, each subspace of which is associated with a distinctive kind of phenomenology.
This decomposition enables a classification of celebrated out-of-equilibrium phenomena, ranging from non-reciprocal (chaser/chased) interactions to the flocking of active agents, dynamical reaction-diffusion patterns, and interface-growth dynamics. We then partially extend these results to dimensions $d>1$.
Furthermore, we study several notions of entropy production and show, in particular, the entropy production rate to be a linear functional of the vorticity two-form, which implies that the factors in our decomposition of $\Omega^2(\mbF)$ can be interpreted as independent sources of entropy production. Finally, we discuss how extending our results to more general situations could provide a natural framework for the generic study of the notoriously diverse and surprising behavior of active systems at their boundaries.
The geometric framework offered in this paper is illustrated throughout by reference to particular models that break time-reversal symmetry, such as Active Model B. 
\end{abstract}


\newpage

\tableofcontents

\newpage
\section{Introduction}

Spatially extended physical systems are often best described by field theory, an approach that has led to a wide range of applications, from the description of the superconducting phase transition to weather forecasting and the design of large-scale civil-engineering infrastructure. When field theory serves as a large-scale, long-time representation of underlying degrees of freedom,
fluctuations naturally arise alongside the deterministic trend and taking them into account is the purpose of statistical field theory. Such an emergent field theory can be derived either through explicit coarse-graining -- often a tedious and inexact process -- or phenomenologically postulated based on the system's symmetries. Notably, when a system's statistics exhibit time-reversal symmetry (TRS) -- such as when it is at or near thermodynamic equilibrium -- the construction and analysis of the field theory are greatly simplified in both approaches.
While early applications of statistical field theory generally obeyed this symmetry, recent decades have witnessed growing interest in irreversible systems, often driven by biological motivations~\cite{kardar1986dynamic,toner1995long,marchetti2013hydrodynamics}. Bridging the theoretical gap between equilibrium and out-of-equilibrium systems has sparked a vast body of literature focused on understanding the breakdown of TRS itself~\cite{graham1971generalized,haussmann1986time,jiang2004mathematical,seifert2005entropy,seifert2012stochastic,yang2021bivectorial,
schnakenberg1976network,dal2023geometry,o2020lamellar,dinelli2023non,o2023nonequilibrium}, although most of the attention has been dedicated so far to the simpler case of finite-dimensional systems.

In this context, two key questions arise: (1) How to determine whether a given field theory is reversible? (2) How does TRS breakdown manifest phenomenologically? These questions represent critical milestones that remain incompletely understood. 
This article is the second in a two-part series. 
The first part~\cite{o2024geometric} was devoted to systematically exploring various types of time-reversal symmetries in stochastic systems with a finite number of degrees of freedom. This second part extends and adapts these results to field theory, with a focus on addressing the aforementioned questions (1) and (2), and in doing so, it notably provides the foundation for a classification of out-of-equilibrium field theories.


Question (1) is notoriously related to determining whether a (possibly vector-valued) map $\bzeta=(\zeta_i(\br,[\bphi]))_i$ -- that is both a function of a spatial variable $\br$ and a functional of the considered field $\bphi$ -- is the functional derivative of a free energy $\mcF[\bphi]$, \ie if $\bzeta = \delta\mcF/\delta \bphi$. However, the first issue is that, unlike in systems with finitely-many degrees of freedom, we are not aware of any systematic approach in the literature that allows identifying, in a generic stochastic PDE, the functional $\bzeta$ that needs to derive from a free-energy for the dynamics to be time-reversible.
We fill this gap here by introducing convenient notations that allow to straightforwardly extend the knowledge from finite dimension~\cite{o2024geometric}.
Furthermore, the popular approach to tackle the aforementioned integrability issue has consisted -- until quite recently -- in guessing a fairly general form of the free energy $\mcF[\bphi]$, then taking its functional derivative, and finally deducing the corresponding constraints over $\bzeta$. While this approach works well when $\bzeta$ and $\mcF$ are sufficiently simple~\cite{wittkowski2014scalar,nardini2017entropy}, it can quickly become impractical. 
In recent years, it was understood that a more systematic criterion to answer this question was given by a functional version of the Schwarz condition~\cite{grafke2017spatiotemporal,o2020lamellar,dinelli2023non}: $\bzeta$ is integrable iff $\delta\zeta_i(\br,[\bphi])/\delta\phi^j(\br')=\delta\zeta_j(\br',[\bphi])/\delta\phi^i(\br)$ in a distributional sense, \ie iff, for all $\delta\bphi_1(\br), \delta\bphi_2(\br)$,
\begin{equation}
\sum_{i,j}\int\left[ \frac{\delta\zeta_i(\br,[\bphi])}{\delta\phi^j(\br')}-\frac{\delta\zeta_j(\br',[\bphi])}{\delta\phi^i(\br)}\right]\delta\phi_1^j(\br')\delta\phi_2^i(\br) d\br d\br' =0 \ .
\label{intro_Schwarz}
\end{equation} 
Despite its systematic nature, since this criterion is distributional, it can be highly non-trivial to deduce the constraints that follow from~\eqref{intro_Schwarz} on \eg the coefficients of $\bzeta$ in cases when the latter is a gradient expansion in $\bphi$. 
To solve this problem, we here reformulate condition~\eqref{intro_Schwarz} in a more geometric form, generalizing the approach developed for two specific examples in a previous paper by one of us~\cite{o2023nonequilibrium}. The resulting integrability condition consists in the cancellation of a functional two-form, the \textit{vorticity two-form}, which is a vector (with additional properties) that lives in a certain functional space $\Omega^2(\mbF)$. Exhibiting a basis of the corresponding space then makes the distributional condition~\eqref{intro_Schwarz} equivalent to the cancellation of the components of this vector on that basis, a more algebraic formulation of the integrability condition that turns out to be readily applicable. In this article, we exhibit such a basis for the case of a local functional\footnote{See below eq.~\eqref{AMB_chemical_potential} for the definition of local functionals.}
 $\bzeta$ in one spatial dimension, and partially generalize our result to higher dimensions.

Besides, if the existence of a free energy in reversible systems grants access to generic methods for analyzing their phenomenology,
that of out-of-equilibrium systems are generally examined case by case, hence addressing only anecdotally the generic aspect of question (2).
Remarkably, the approach we adopt to solve problem (1) provides a partial yet interesting answer to question (2) as well, since we show that a phenomenology can be attributed to each element of the aforementioned basis. 
Further, in one spatial dimension, this basis can be split into three subfamilies,
each generating a subspace of $\Omega^2(\mbF)$ that we respectively call the antisymmetric, self-symmetric, and inter-symmetric subspace, and denote by $\mcA$, $\mcS_{\rm self}$, and $\mcS_{\rm  inter}$. The combination of these three subspace hence provides a direct-sum decomposition of the space of two-forms: 
\begin{equation}
\Omega^2(\mbF)= \mcA\oplus\mcS_{\rm self} \oplus\mcS_{\rm inter} \ .
\label{vorticity_space_decomp_intro}
\end{equation}
The phenomenology within each of these three subspaces turns out to be quite robust and generic: the antisymmetric subspace typically corresponds to systems with non-reciprocal interactions, while the inter-symmetric one is related to flocking, and the self-symmetric subspace gathers field dynamics like that of Active Model B (AMB) or the Kardar-Parisi-Zang (KPZ) equation, which display anisotropic propagation of fluctuations.
The decomposition~\eqref{vorticity_space_decomp_intro} -- and its extension~\eqref{general_decomposition_vorticity_space} to more general situations -- thus paves the way for a classification of out-of-equilibrium  stochastic field theories.

Among the set of observables quantifying irreversibility at field level, entropy production plays a special role~\cite{o2023nonequilibrium,nardini2017entropy}. The question of its relation with the vorticity two-form hence arises naturally. In this article we prove that, as in finite dimension~\cite{yang2021bivectorial,o2024geometric}, the loop-wise entropy production\footnote{This is the entropy produced along a loop in function space, \ie during one period of a cyclic evolution of the field configuration.} and the entropy production rate are both linear functional of the vorticity two-form, hence emphasizing the central role played by the latter in the analysis of TRS breakdown. Interestingly, this also implies that the elements of the decomposition~\eqref{vorticity_space_decomp_intro} -- and its generalisation~\eqref{general_decomposition_vorticity_space} -- can be interpreted as \textit{independent sources of entropy production}.
These results hold for the usual notion of entropy production -- that we here call the \textit{bare entropy production} -- that compares the statistics of field trajectories to their time-reversal. Nevertheless, we show that in certain situations, like when the field follows a conservation equation, there exists an extended notion of entropy production, which directly takes into account the statistics of the current associated to $\bphi$, and that is greater or equal to the bare one, the difference between the two being invisible from observations solely of the field statistics and hence from the vorticity two-form.

Time reversal, for a given stochastic dynamics, is a context dependent notion~\cite{o2024geometric}. In the simplest cases it solely amounts to reversing the time variable. In other dynamics it also requires us to flip some degrees of freedom, as is the case for \eg momentum variables in underdamped models. Finally, it can also require us to directly reverse a part of the ``force field'' that is applied to the system, like a magnetic field for instance.
In the first article~\cite{o2024geometric} of this two-part series, we respectively called these three types of time reversal T-, MT-, and EMT-reversals, the letters T, M, and E standing for ``time'', ``mirror'', and ``extended'' (see~\cite{o2024geometric} for details). 
Just as in the finite dimensional setting, MT- and EMT- reversibility do not reduce to a single integrability condition, as opposed to the case of T-reversibility. 
In this article, we give for the first time general reversibility conditions for a stochastic field dynamics to be EMT-reversible (this notion of time reversal subsuming the two others).
%
%


The article is structured as follows.
We start by describing in section~\ref{sec:framework} the type of field dynamics we will focus on, together with the notations that will be used to extend known finite-dimensional results -- in particular those of our companion paper~\cite{o2024geometric}.
We then progressively introduce in section~\ref{sec:quantifyingTRSbreaking} the theoretical machinery -- in particular a functional exterior calculus -- that we will use later to study TRS breakdown in the considered field theories. Throughout the development of our theoretical framework, we illustrate its various components on a fixed example: a nonequilibrium model of phase separation.
In section~\ref{sec:basis}, which is the culmination of this article, we uncover a basis of vorticity two-forms in one spatial dimension. As announced above, the latter first allows us to turn the previously distributional reversibility condition into a readily-applicable, algebraic one. We then show that the resulting basis can be divided into three subfamilies, each one having a particular kind of phenomenology. This allows a classification of out-of-equilibrium field theories in one spatial dimension, going from non-reciprocally interacting systems, to irreversible reaction-diffusion systems, flocking, and interface growth. We then partially extend these results for spatial dimensions higher than one. We also discuss how generalizing the decomposition of the space $\Omega^2(\mbF)$ could account for some of the notoriously surprising behaviors of active systems at their boundaries.
%
In section~\ref{sec:entropy}, we connect our work to various notions of entropy production considered in the previous literature~\cite{maes2000definition,jiang2004mathematical,seifert2005entropy,yang2021bivectorial}
and some new ones.
Finally, while the main focus of sections~\ref{sec:quantifyingTRSbreaking}-\ref{sec:entropy} is on T-reversal,
we give in section~\ref{sec:MTreversal} reversibility conditions for the more general notions of MT- and EMT-reversibility and for dynamics in which the space to which the field belongs is not simply-connected.


\section{General context and notations}
\label{sec:framework}
This section describes the class of stochastic partial differential equations (SPDEs) whose irreversibility we shall study in the following. These SPDEs will consist of fairly general PDEs perturbed by a Gaussian random field with zero mean, delta-correlated in time, and arbitrarily correlated in space.

SPDEs are notoriously difficult to define in general~\cite{holden1996stochastic,prevot2007concise,hairer2014theory}, and we will not attempt to provide a comprehensive definition here. However, we will construct our framework so as to avoid the pitfalls already encountered in finite dimensions -- as identified in Part I~\cite{o2024geometric} -- in the geometric study of the irreversibility of stochastic differential equations (SDEs).

We begin in section~\ref{subsec:finiteDimFramework} by briefly summarizing the main issues that arise in the physical description of SDEs. In section~\ref{subsec:SPDEframework}, we introduce notations that will facilitate extending the finite-dimensional approach developed in Part I to the infinite-dimensional setting of field theory, by analogy.

\subsection{Physical description of SDEs}
\label{subsec:finiteDimFramework}

In this section, we recapitulate how to describe a SDE in a way that is physical, covariant, and independent of the stochastic prescription (see Part I~\cite{o2024geometric} for more details).

Let us consider a stochastic process $\bx_t$ on a finite-dimensional manifold $\mcM$, satisfying the following SDE for a given stochastic integral prescription $\varepsilon_0 \in [0,1]$ (\eg $\varepsilon_0 = 0$ and $1/2$ correspond to the Ito and Stratonovich conventions, respectively): 
\begin{equation}
 \dot{\bx}_t = \bA_{(\varepsilon_0)} + \bb_{\alpha} \eta_t^\alpha \ . 
\label{EDS0} 
\end{equation}
This dynamics may, for instance, describe the evolution of the positions of $N$ particles in $\mbR^3$ -- in which case $\bx_t = (\br_1(t), \dots, \br_N(t))$ and $\mcM = \mbR^{3N}$ -- interacting with each other and in contact with a thermal bath.
In equation~\eqref{EDS0}, the $\eta_t^\alpha$ are real Gaussian white noises with zero mean and correlations $\langle \eta_t^\alpha \eta_{t'}^\beta \rangle = 2\delta^{\alpha\beta}\delta(t - t')$; the $\bb_\alpha(\bx_t)$ are vector fields on $\mcM$, and $\bA_{(\varepsilon_0)}(\bx_t)$ -- which is not generally a vector field\footnote{The transformation law of $\bA_{(\varepsilon_0)}$ under coordinate changes is determined by the chain rule associated with the prescription $\varepsilon_0$. It is therefore a genuine vector field only when $\varepsilon_0 = 1/2$ (Stratonovich).} -- is the so-called ``$\varepsilon_0$-drift''.

The dynamics~\eqref{EDS0} can equivalently be rewritten in any stochastic prescription $\varepsilon$, in which case the new drift takes the form $A^i_{(\varepsilon)} = A^i_{(\varepsilon_0)} + 2(\varepsilon_0 - \varepsilon)b^j_\alpha \partial_j b^i_\alpha$. The true physical drift of the process -- which can be interpreted as a force field (up to mobility) -- should not depend on the chosen prescription.
Building on the case of a Langevin equation at uniform temperature but with inhomogeneous mobility, earlier works~\cite{lau2007state, cates2022stochastic} identified a suitable candidate (in an orthonormal frame) as $\bA \equiv \bA_{(\varepsilon)} - \bs_{(\varepsilon)}$, where
\begin{equation}
s^i_{(\varepsilon)}\equiv \partial_j b^j_\alpha b^i_\alpha - 2\varepsilon b^j_\alpha\partial_j b^i_\alpha \ ,
\label{finiteDimSpurious}
\end{equation}
is the so-called ``$\varepsilon$-spurious drift''. 

Although constructed to be independent of $\varepsilon$, one can show that $\bA$ is not a contravariant vector field, and therefore cannot be interpreted as a physical force field.
In Part I~\cite{o2024geometric}, we showed that in order to correct this lack of covariance, one must choose a (positive) reference measure on $\mcM$, denoted $\blambda = \lambda(\bx)dx^1 \dots dx^n$, and define a ``$\blambda$-drift'' $\ba_{\blambda} \equiv \bA - \bh_{\blambda}$, where
\begin{equation}
h^i_{\blambda}\equiv b^i_\alpha b^j_\alpha \partial_j \lambda \ ,
\label{finiteDimPriorDrift}
\end{equation}
is called the ``$\blambda$-gauge drift''. This allows one to rewrite the SDE~\eqref{EDS0} in the following form: 
\begin{equation}
\dot{\bx}_t = \ba_{\blambda} + \bs_{(\varepsilon)} + \bh_{\blambda} + \bb_{\alpha} \eta_t^\alpha \ , 
\label{EDS}
\end{equation}
where $\ba_{\blambda}$ is now indeed a (contravariant) vector field on $\mcM$, independent of the stochastic prescription $\varepsilon$, and can thus be interpreted as a force field.

This is done at the cost of choosing a reference measure $\blambda$.
The latter plays the role of a prior or a gauge, depending on whether one aims to construct an SDE from a given ``force field'' $\ba$ and noise $\bb_\alpha \eta^\alpha$, or to analyze a given SDE. In the former case, our work suggests that in addition to specifying the drift $\ba$ and the noise $\bb_\alpha \eta^\alpha$, it is necessary to choose a ``prior'' measure $\blambda$ corresponding to the physicist’s \textit{a priori} assumption about the system’s state of maximal disorder (or minimal information), and then impose $\ba_{\blambda} \equiv \ba$. In particular, in the absence of a ``force field'', \ie when $\ba = 0$, the stationary distribution is then given by $\blambda / \int_{\mcM} d\blambda$.
In the latter case, $\blambda$ is simply a gauge choice that affects the $\blambda$-drift $\ba_{\blambda}$ but not physical observables such as entropy production.


Throughout the remainder of this article, we fix the reference measure $\blambda$ once and for all, and will therefore omit the $\blambda$ subscript on $\ba_{\blambda}$ for clarity.

\subsection{Extension to SPDEs}
\label{subsec:SPDEframework}

In order to generalize the finite dimensional results of~\cite{o2024geometric}, we first need to identify the field-theoretic counterpart of each finite-dimensional geometric object involved in studying the time-irreversibility of~\eqref{EDS}.

\paragraph{General setting.}
Let us consider a time-dependent field $\bphi_t : \mbR^{d_1}\rightarrow\mbR^{d_2}$ that evolves on a time interval $t\in\mbT\equiv[0,\mcT]$ according to a SPDE~\eqref{EDPS01}, which is Stratonovich-discretized in space\footnote{Stating that each spatial differential operator is defined as a limit of a midpoint-discretized operator is actually not sufficient to fully characterize the limit SPDE~\cite{cates2022stochastic}, but we will disregard such questions in this article. In particular, we will assume that the usual chain rule applies to spatial derivatives, hence the term ``Stratonovich-discretized in space''.}, and $\varepsilon\in[0,1]$-discretized in time. Having the finite-dimensional formulation~\eqref{EDS} in mind, we write this SPDE as
\begin{equation}
\partial_t\bphi(\br,t) = \ba(\br,[\bphi]) + \bh_{\blambda}(\br,[\bphi]) + \bs_{(\varepsilon)}(\br,[\bphi]) + \sum_{i=1}^{d_3} \int \rmd\br'\bb_i(\br,\br',[\phi])\eta^i(\br',t) \ ,
\label{EDPS01}
\end{equation}
where $\boldeta(\br,t)\equiv(\eta^1(\br,t),\dots\eta^{d_3}(\br,t))^\top$ is a random Gaussian field on $\mbR^{d_1}$, taking values in $\mbR^{d_3}$, of zero mean and correlations given by $\langle\eta^i(\br,t)\eta^j(\br',t')\rangle = 2\delta^{ij}\delta(\br-\br')\delta(t-t')$, and $\bb(\br,\br',[\bphi])\equiv(\bb_1(\br,\br',[\bphi]),\dots, \bb_{d_3}(\br,\br',[\bphi]))$ is a kernel operator sending $\mbR^{d_3}$-valued fields to $\mbR^{d_2}$-valued fields over $\mbR^{d_1}$.
In order to study dynamics~\eqref{EDPS01} in both a covariant\footnote{Here the covariance is with respect to a change of \textit{chart over the function space} $\mbF$ to which $\bphi$ belongs. For instance, going from Cartesian to spherical coordinates in $\mbR^{d_1}$ amounts to a change of chart on $\mbF$, since it gives a different coordinate expression of a field $\bphi(r^1,\dots,r^{d_1})\to\tilde{\bphi}(r, \theta_1,\dots,\theta_{d_1-1})$.} and discretization-free way, just as in the finite-dimensional situation~\eqref{EDS}, we have split the deterministic drift of eq.~\eqref{EDPS01} into three terms: a time-discretization-dependent spurious drift $\bs_{(\varepsilon)}(\br,[\bphi])$, a $\blambda$-gauge drift $\bh_{\blambda}(\br,[\bphi])$, and a $\blambda$-drift  $\ba(\br,[\bphi])$ (the latter two being independent of $\varepsilon$), where $\blambda$ is a reference-measure over the functional space -- denoted by $\mbF$ -- to which $\bphi$ belongs. 
We will denote by $\bbphi$,  $\bb[\bphi]$, or simply $\bb$ when there is no ambiguity, the operator associated to the kernel $\bb(\br,\br',[\bphi])$, \ie $\bbphi\boldeta (\br)\equiv \int \bb(\br,\br',[\bphi])\cdot\boldeta(\br')\rmd \br' =\sum_i\int \bb_i(\br,\br',[\bphi])\eta^i(\br')\rmd \br' $, where $\rmd \br\equiv \rmd r^1\dots \rmd r^{d_1}$ denotes the Lebesgue measure on $\mbR^{d_1}$.

To generalize more easily the finite-dimensional results from~\cite{o2024geometric} to this field-theoretic context, we are going to formally ``geometrize'' the dynamics~\eqref{EDPS01}. This first means that we regard the functional space $\mbF$ as an infinite dimensional manifold on which the dynamics takes place. We further denote by $T_{\bphi}\mbF$ its tangent space at any $\bphi$, which can be thought of as the space of small fluctuations around $\bphi$, and whose structure will be further detailed shortly. In turn, $\ba$ is a vector field over $\mbF$, \ie for all $\bphi\in\mbF$, $\ba[\bphi]\equiv \ba(\cdot,[\bphi])$ is a tangent vector belonging to $T_{\bphi}\mbF$. In this perspective, the spatial variable $\br$ plays the role of a continuous coordinate index so that, for instance, the vector $\ba[\bphi]\in T_{\bphi}\mbF$ has coordinates $a^{i\br}_{\bphi}=a^{i\br}[\bphi]\equiv a^i(\br,[\bphi])$, where $i\in\llbracket 1,d_2 \rrbracket$ and $\br\in\mbR^{d_1}$.
Furthermore, we will use the following generalized Einstein convention: in addition to summing over repeated discrete indices, we will also implicitly integrate over repeated continuous variables, when written as indices.
Dynamics~\eqref{EDPS01} then reads
\begin{equation}
\partial_t\phi^{i\br} = a^{i\br}[\bphi] + h_{\blambda}^{i\br}[\bphi] + s^{i\br}_{(\varepsilon)}[\bphi] + b^{i\br}_{j\br'}[\bphi]\eta^{j\br'} \ ,
\label{EDPS02}
\end{equation}
or, if we stack together the discrete indices:
\begin{equation}
\partial_t\bphi^\br = \ba^\br[\bphi] + \bh_{\blambda}^{\br}[\bphi] + \bs^{\br}_{(\varepsilon)}[\bphi] + \bb^\br_{\br'}[\bphi]\cdot \boldeta^{\br'} \ .
\end{equation}
These conventions allow us to deduce the coordinate expression of the spurious and $\blambda$-gauge drift directly from their finite dimensional counterparts~\eqref{finiteDimSpurious} \&~\eqref{finiteDimPriorDrift}: they respectively read
\begin{equation}
s^{i_1\br_1}_{(\varepsilon)} \equiv \frac{\delta}{\delta \phi^{i_2\br_2}} b^{i_2\br_2}_{i_3\br_3} b^{i_1\br_1}_{i_3\br_3} - 2\varepsilon	 b^{i_2\br_2}_{i_3\br_3} \frac{\delta}{\delta \phi^{i_2\br_2}} b^{i_1\br_1}_{i_3\br_3} \ ,
\label{functionalSpuriousDrift}
\end{equation} 
and
\begin{equation}
h^{i_1\br_1} = b^{i_1\br_1}_{i_3\br_3}b^{i_2\br_2}_{i_3\br_3}\frac{\delta}{\delta \phi^{i_2\br_2}}\lambda \ .
\label{functionalPriorDrift}
\end{equation}
The only kind of (functional) coordinate charts on the manifold $\mbF$ that we consider here correspond to choosing an arbitrary coordinate system on $\mbR^{d_1}$ and any basis in $\mbR^{d_2}$. In particular, this implies that if we choose the Lebesgue measure on $\mbF$ as the reference measure $\blambda$, then whatever the functional coordinate chart, $\blambda$ is independent of $\bphi$ and hence the $\blambda$-gauge drift $\bh_{\blambda}$ vanishes\footnote{What we say about the ``functional Lebesgue measure'' $\blambda\equiv\mcD\bphi$ should be understood in a discretized setting. For instance, discretizing $(\phi^i(\br))_{i=1\dots d_2}$ at given sites $\br_{\alpha}\in\mbR^{d_1}$ gives $(\phi^{i\alpha})_{i,\alpha}$, and the ``functional Lebesgue measure'' approximately reads $\blambda \simeq \widetilde{\blambda}\equiv \prod_{i,\alpha} d\phi^{i\alpha}$ in the canonical basis of $\mbR^{d_2}$. If we now change basis in $\mbR^{d_2}$, a prefactor appears in the expression of $\widetilde{\blambda}$. But the latter is independent of $\bphi$, hence the discrete version of the $\blambda$-gauge drift vanishes whatever the spatial discretization and we can consider its continuous limit $\bh_{\blambda}$ to be equally vanishing.}.

Similarly, if the operator $\bb$ does not depend on $\bphi$, the spurious drift $\bs_{(\varepsilon)}$ is identically zero.
When this is not the case, some diverging terms may appear that should then be taken care of by a proper spatial-discretization  scheme~\cite{cates2022stochastic}.
For instance if $\bphi$ is a scalar field, \ie $d_2=1$, and 
\begin{equation}
b^{\br}_{j\br'} =\sqrt{M(\bphi(\br'))}\partial_{r'_j}\delta(\br-\br') \ ,
\label{ExampleOfB}
\end{equation}
with $M$ a mobility, then
\begin{eqnarray}
\bb^\br_{\br'}[\phi]\cdot\boldeta^{\br'}  =  b^{\br}_{j\br'}[\phi]\eta^{j\br'} = \int \rmd \br' [\sqrt{M(\bphi(\br'))}\partial_{r'_j}\delta(\br-\br')] \eta^j(\br') = -\nabla\cdot \sqrt{M(\bphi(\br))}\boldeta(\br) \ ,
\end{eqnarray}
while the spurious drift formally reads (see appendix~\ref{AppIllDefSpurious})
\begin{eqnarray}
s^{\br_1}_{(\varepsilon)} = (1-\varepsilon) \sum_i \int \rmd \br_2 \left[ \partial_{r_2^i}\delta(\br_1-\br_2)\right] \left.\left[\partial_{r_3^i}\frac{\delta M(\phi(\br_2))}{\delta \phi^{\br_3}}\right]\right|_{\br_3=\br_2} \ .
\label{spurious_temp}
\end{eqnarray}
The issue is that the last term in squared brackets in eq.~\eqref{spurious_temp} is ill defined.
But if we go back to its discretized version, assuming $d_1=1$ to simplify notation, it takes the form: 
\begin{eqnarray}
\left.\left[\partial_{r_3^i}\frac{\delta M(\phi(\br_2))}{\delta \phi^{\br_3}}\right]\right|_{\br_3=\br_2} =  M'(\phi(\br_2)) \left.\left[\partial_{r_3^i}\frac{\delta\phi^{\br_2}}{\delta\phi^{\br_3}}\right]\right|_{\br_3=\br_2}\to  M'(\phi(k_2)) \frac{1}{2}\left[\frac{\partial \phi(k_2)}{\partial \phi(k_2+1)}-\frac{\partial \phi(k_2) }{\partial \phi(k_2-1)} \right] \ ,
\end{eqnarray}
hence it clearly vanishes.

%

Using the conventions introduced above, the Fokker-Planck equation associated to~\eqref{EDPS02} reads
\begin{equation}
\partial_tP_t[\bphi] = - \frac{\delta}{\delta\phi^{i\br}} J_t^{i\br} \ ,
\label{FokkerPlanck01}
\end{equation}
where $P_t[\bphi]$ is the probability density (with respect to the Lebesgue measure on $\mbF$) of the solution to eq.~\eqref{EDPS02} at time $t$, and $\bJ_t$ is the corresponding functional probability current, which is a vector field over $\mbF$ and whose coordinate expression is
\begin{equation}
J_t^{i\br} \equiv a^{i\br}[\bphi] P_t - D^{i\br,j\br'}\frac{\delta}{\delta \phi^{j\br'}}P_t \ .
\label{ProbCurrent01}
\end{equation}
In the expression~\eqref{ProbCurrent01} of the probability current, $\bD$ is the diffusion operator, defined by
\begin{equation}
D^{i\br,j\br'} \equiv b^{i\br}_{k\br''}b^{j\br'}_{k\br''} \ ,
\end{equation}
or, in operator notations, $\bD = \bb \bb^\dagger$, were $\bb^\dagger$ stands for the $L^2$-adjoint of $\bb$ for the Lebesgue measure on $\mbR^{d_1}$ and the canonical scalar products of $\mbR^{d_2}$ and $\mbR^{d_3}$.
Throughout this article, we assume that the stationary probability measure $\Pss[\bphi]$ exists and is unique.

 A pivotal object in our study of the time-reversal properties of dynamics~\eqref{EDPS02} will be $\bD^{-1}\ba$. Note that our results can be (at least partially) extended when $\bD$ is not invertible, as we did in the finite dimensional setting in~\cite{o2024geometric}. But for simplicity, up to section~\ref{sec:MTreversal}, we will assume below that $\bD^{-1}\ba$ is always well defined. In particular, this requires that $\ba_{\bphi}$ belongs to the image of $\bb_{\bphi}$, $\IM(\bb_{\bphi})$. Consequently all the drifts applied to $\bphi$ in eq.~\eqref{EDPS02} are included\footnote{Note that the whole right-hand side of a stochastic dynamics~\eqref{EDS} is generically a tangent vector field only in the Stratonovitch prescription, for which we see that $\bs_{(1/2)}+\bh_{\blambda}$ indeed belongs to $\IM(\bb_{\bx})$. We here extrapolate this phenomenon, known for SDEs~\cite{elworthy1998stochastic,hsu2002stochastic}, to SPDEs.} in $\IM(\bb_{\bphi})$ and hence the smallest functional manifold $\mbF$ in which the solutions of eq.~\eqref{EDPS02} are confined is imposed by $\IM(\bbphi)$, \ie $T_{\bphi}\mbF=\IM(\bb_{\bphi})$.
  Regarding the global topology of $\mbF$, we likewise assume, up to section~\ref{sec:MTreversal}, it is \textit{simply-connected}.
We stress that, although the field $\bphi$ is defined on $\mbR^{d_1}$ for simplicity, our approach can be adapted to any boundary-less smooth manifold without major difficulty.

\paragraph{A running example: the Active Model B.}
For the sake of concreteness, as we develop our somewhat abstract theoretical framework in sections~\ref{sec:quantifyingTRSbreaking}-\ref{sec:entropy}, we will repeatedly illustrate most of its various components on a fixed example,
the so-called Active Model B (AMB), which is arguably the simplest field theory to show out-of-equilibrium phase separation~\cite{wittkowski2014scalar,nardini2017entropy}. The AMB dynamics takes the form:
\begin{equation}
\partial_t\rho(\br) = \nabla\cdot \left[ M\nabla\mu(\br,[\rho]) + \sqrt{M}\boldeta \right] .
\label{AMB}
\end{equation}
In eq.~\eqref{AMB}, $M$ is a constant scalar mobility, $\boldeta(\br,t)$ a Gaussian random field identical to that of eq.~\eqref{EDPS02}, and $\mu(\br,[\rho])$ is a chemical potential defined, in the spirit of Landau-Ginzburg theory, as a second order expansion in the gradient of $\rho$:
\begin{equation}
\mu(\br,[\rho])\equiv \alpha\rho(\br) + \beta\rho(\br)^3 +\lambda(\rho(\br))|\nabla\rho(\br)|^2 - \kappa(\rho(\br)) \Delta\rho(\br),
\label{AMB_chemical_potential}
\end{equation}
where $\alpha$ and $\beta$ are constants, while, generalizing slightly~\cite{wittkowski2014scalar}, we choose $\lambda$ and $\kappa$ to be strictly local functionals of $\rho$. 
In this article, a map $\mu(\br,[\rho])$, which is both a function of a spatial variable $\br$ and a functional of a field $\rho$, will be said to be \textit{local} if it depends on the value of $\rho$ and of its derivatives up to a finite order $q$ at $\br$, $\mu(\br,[\rho])=\mu\left(\rho(\br),\frac{\partial\rho(\br)}{\partial r^i},\dots, \frac{\partial^q \rho(\br)}{\partial r^{i_1}\dots\partial r^{i_q}}\right)$, and \textit{strictly local} if it is local with $q=0$.
With the notations used to describe dynamics~\eqref{EDPS02} above, of which the AMB dynamics is a particular instance, the reference measure is taken to be the Lebesgue measure so that $\bh_{\blambda}=0$. The $\blambda$-drift reads $a^{\br}[\rho]=\nabla\cdot M\nabla\mu(\br,[\rho])$, while the operator $\bb$ takes the form $b^{\br'}_{i\br}=-\sqrt{M}\frac{\partial}{\partial r^i}\delta(\br-\br')$, which in particular implies that the spurious drift vanishes for any time-discretization scheme, $\bs_{(\varepsilon)}=0$, and that the diffusion operator reads $D^{\br\br'}=-M\Delta \delta(\br-\br')$.
Note that, in this case, as assumed for the more general dynamics~\eqref{EDPS02}, $\bD^{-1}\ba=-\Delta^{-1}\Delta\mu$ is well defined and coincides with minus the chemical potential $\mu$. 
Finally, we stress that this example also satisfies the general hypothesis on the simply-connectedness of $\mbF$. Indeed, the constraint over the space $\mbF$ in this case is essentially the conservation of a given total ``mass'' $m=\int_{\mbR^{d_1}}\rho$. Hence, if we take $\rho_1,\rho_2\in\mbF$, any convex linear superposition shares the same global mass $m$, and thus also belongs to $\mbF$. It follows that $\mbF$ is convex, and consequently simply-connected.

\newpage

\section{Quantifying time-irreversibility through the vorticity two-form}
\label{sec:quantifyingTRSbreaking}

In this section, we study the presence or absence of time-reversal symmetry (TRS) in the dynamics~\eqref{EDPS01} or equivalently~\eqref{EDPS02}.
In~\cite{o2023nonequilibrium}, one of us introduced for this purpose a functional exterior derivative for scalar field, which generalizes the finite-dimensional exterior derivative -- the latter being itself a generalization of the curl operator of vector analysis to dimensions higher than three.
In section~\ref{subsec:RevCondFuncExtDer}, we start by extending this construction to the more general context of the vector-valued field $\bphi$ obeying eq.~\eqref{EDPS02}.
This then allows us to define a functional vorticity two-form (alternatively called a functional cycle-affinity two-form), denoted by $\bomega$, that generalizes to field theory the corresponding objects already introduced for Markov chains~\cite{kolmogorov1936theorie,schnakenberg1976network} and finite dimensional stochastic processes~\cite{yang2021bivectorial}.
In section~\ref{subsec:vorticityOperator}, we introduce two differential operators, the cycle-affinity operator $\whomega$ and the vorticity operator $\bW$, that are both dual to the vorticity two-form $\bomega$.
While the benefits of introducing the former operator will only appear in section~\ref{subsec:skewsym_operators}, we show below that the latter can be used to gain insight into the out-of-equilibrium phenomenology of the dynamics~\eqref{EDPS02} under study.

\subsection{Reversibility condition and functional exterior derivative}
\label{subsec:RevCondFuncExtDer}
Up until section~\ref{sec:MTreversal}, we assume that all the components of the field $\bphi$ \textit{are even under time-reversal}.
Dynamics~\eqref{EDPS02} is then said to be time-reversal symmetric when, in steady state, it has the same probability to travel any trajectory $(\bphi_t)_{t\in\mbT}$ forward and backward in time\footnote{We also assume that the time-reversal of dynamics~\eqref{EDPS02} does not involve any direct modification of the drift and diffusion operator, as it would happen \eg in the presence of an external magnetic field. In the language used in paper I~\cite{o2024geometric}, this means that we only focus on T-reversal until section~\ref{sec:MTreversal} where we consider EMT-reversal. Besides, at least when the field $\bphi$ is conserved, one can also define a resolved notion of reversibility where, together with the statistics of the field, those of the corresponding current are explicitly considered (see section~\ref{subsec:hiddenIrrev}).}. This reads $\mcP[(\bphi_t)_{t\in\mbT}] = \mcP[(\bphi_{\mcT-t})_{t\in\mbT}]$, where $\mcP$ is the (stationary) path probability associated to dynamics~\eqref{EDPS02}.
Just as in finite dimension~\cite{risken1996fokker,van1992stochastic,gardiner1985handbook,yang2021bivectorial,jiang2004mathematical}, this time-reversibility property is equivalent to the cancellation of the stationary probability current $\bJss=\ba\Pss - \bD\delta \Pss/\delta\bphi$, a property that can be straightforwardly reformulated as
\begin{equation}
\bD^{-1}\ba = \frac{\delta \ln \Pss}{\delta\bphi} \ .
\label{TRScond01}
\end{equation}
In other words, time-reversibility of dynamics~\eqref{EDPS02} amounts to a (functional) integrability constraint\footnote{There is also an normalizability constraint over the potential from which $\bD^{-1}\ba$ derives. But in practice, just as in finite dimension~\cite{o2024geometric}, we only need the kernel of the Fokker-Planck operator, $P\mapsto -\frac{\delta}{\delta \phi^{i\br}}[a^{i\br}P - D^{i\br j\br'}\frac{\delta}{\delta \phi^{j\br'}}P]$, to be one-dimensional, not that any of these stationary measures integrates to one. Thus, without loss of generality, we will disregard this normalizability constraint throughout this article.} over $\bD^{-1}\ba$. 
The reversibility condition~\eqref{TRScond01} has been argued~\cite{dinelli2023non} to be equivalent to $\bD^{-1}\ba$ satisfying the functional integrability condition:
\begin{equation}
\frac{\delta[D^{-1}a_{\bphi}]_{i\br}}{\delta\phi^{j\br'}}-\frac{\delta[D^{-1}a_{\bphi}]_{j\br'}}{\delta\phi^{i\br}} = 0 \ ,
\label{TRScond02}
\end{equation}
in a distributional sense, for all $\bphi\in\mbF$, with $[D^{-1}a_{\bphi}]_{i\br}$ standing for $[D^{-1}a]_{i\br}[\bphi]=[D^{-1}_{\bphi}]_{i\br j\br'}a_{\bphi}^{j\br'}$. While the authors of~\cite{dinelli2023non} proved that eq.~\eqref{TRScond01} implies eq.~\eqref{TRScond02}, the simply-connectedness of $\mbF$ was assumed to be enough for the reciprocal implication to be valid, as in finite dimension. Among other things, the formalism introduced below allows us to prove this is indeed the case.

We here emphasize that our analysis of functional integrability is conducted with significantly greater mathematical rigor than our treatment of stochastic PDEs in this article. This level of rigor -- necessary to draw strong physical conclusions later on -- justifies the formal nature of some of the sections that follow.


\paragraph{Functional one-forms.}

In order to fully prove and interpret the equivalence between eqs.~\eqref{TRScond01} and~\eqref{TRScond02}, and also to exploit the information content of the object on the left-hand side of eq.~\eqref{TRScond02} when it is not identically zero, 
 we are going to recast that object in a more geometric form.
 
We start by noting that, from our geometrical point of view, since $\bD$ is a functional contravariant tensor field of order two over $\mbF$, its inverse $\bD^{-1}$ -- which is defined by $[D_\phi]^{i\br j\br'}[D^{-1}_\phi]_{j\br' k\br''} =\delta^{i\br}_{k\br''}$, with $\delta^{i\br}_{k\br''}$ the identity map on $T_{\bphi}\mbF$ -- is naturally a covariant functional tensor field of order two. Hence $\bD^{-1}\ba$ is a covariant vector field over $\mbF$, \ie a functional one-form.
%
It acts on perturbations $\delta\bphi\in T_{\bphi}\mbF$ as
\begin{equation}
\bD^{-1}\ba_{\bphi} (\delta\bphi) \equiv [D^{-1}a_{\bphi}]_{i\br} \delta\phi^{i\br} = [D^{-1}_{\bphi}]_{i\br j\br'}a_{\bphi}^{j\br'}\delta\phi^{i\br} \ .
\end{equation}
Upon introducing $\delta^{i\br}$, the Dirac delta at point $\br$ applied to the $i^{th}$ component of $\delta\bphi$, \ie 
\begin{equation}
\delta^{i\br}(\delta\bphi)\equiv \delta\phi^i(\br) \ ,
\end{equation}
the one-form $\bD^{-1}\ba$ can be written as
\begin{equation}
\bD^{-1}\ba= [\bD^{-1}\ba]_{i\br}\delta^{i\br} \ .
\label{functOneFormsUsingDelta}
\end{equation}
We denote by $\Omega^1(\mbF)$ the space of functional one-forms $\bzeta_{\bphi}=\zeta_{i\br}[\bphi]\delta^{i\br}$ over $\mbF$.

\paragraph{Functional two-forms and exterior derivative.}
A functional two-form $\bgamma$ over $\mbF$ is a field of maps $\bgamma_{\bphi}:(\delta\bphi_1,\delta\bphi_2)\in T_{\bphi}\mbF\times T_{\bphi}\mbF \to \bgamma_{\bphi}(\delta\bphi_1,\delta\bphi_2)\in \mbR$ that are bilinear and antisymmetric for all $\bphi$. We denote by $\Omega^2(\mbF)$ the space of functional two-forms over $\mbF$.
As promised above, we now generalize to vector-valued fields the functional exterior derivative introduced in~\cite{o2023nonequilibrium} for scalar field: the exterior derivative of a one-form $\bzeta_{\bphi}=\zeta_{i\br}[\bphi] \delta^{i\br}$ is a two-form, denoted by $\mbd \bzeta$, which acts on any pair of perturbations $\delta\bphi_1(\br),\delta\bphi_2(\br) \in T_{\bphi}\mbF$ as:
\begin{equation}
\mbd \bzeta_{\bphi}(\delta\bphi_1,\delta\bphi_2) \equiv \sum_{i,j=1}^{d_2}\int \left\{ \frac{\delta \zeta_{i\br}[\bphi]}{\delta\phi^{j\br'}}-\frac{\delta\zeta_{j\br'}[\bphi]}{\delta\phi^{i\br}} \right\} \delta\phi_1^{j\br'}\delta\phi_2^{i\br} \rmd\br\rmd\br' \ .
\label{functExtder01}
\end{equation}
The space $\mbd \Omega^1(\mbF)\subset\Omega^2(\mbF)$ of exterior derivatives of one-forms is the vector space of \textit{exact} functional two-forms.
We now defined the wedge product $\mcD_1\wedge\mcD_2$ of two distributions $\mcD_1,\mcD_2:T_{\bphi}\mbF\to\mbR$ as the antisymmetric bilinear map from $T_{\bphi}\mbF\times T_{\bphi}\mbF$ to $\mbR$ which reads 
\begin{equation}
\mcD_1\wedge\mcD_2 (\delta\bphi_1,\bphi_2)\equiv \mcD_1(\delta\phi_1)\mcD_2(\delta\bphi_2)-\mcD_1(\delta\phi_2)\mcD_2(\delta\bphi_1) \ .
\end{equation}
This product is antisymmetric: $\mcD_1\wedge \mcD_2 = -\mcD_2\wedge\mcD_1$.
The wedge product allows us to re-write the exterior derivative~\eqref{functExtder01} without its arguments $\delta\bphi_1,\delta\bphi_2$ as
\begin{equation}
\mbd \bzeta =  \frac{1}{2}\left\{ \frac{\delta\zeta_{i\br}}{\delta\phi^{j\br'}}-\frac{\delta\zeta_{j\br'}}{\delta\phi^{i\br}} \right\}\delta^{j\br'}\wedge\delta^{i\br} \ ,
\label{functExtder02}
\end{equation}
which emphasizes the fact that it is a geometrical object on its own.
Note that the wedge product of Dirac deltas also allows to write any two-form $\bgamma$ as $\gamma_{\bphi}=\gamma_{i\br j\br'}[\bphi]\delta^{i\br}\wedge \delta^{j\br'}$, where the $\gamma_{i\br j\br'}[\bphi]$ play the role of components of $\bgamma$ along the elementary two-forms $\delta^{i\br}\wedge \delta^{j\br'}$, the latter being antisymmetric bilinear maps on each tangent space $T_{\bphi}\mbF$.

Just as in finite dimension, this functional exterior derivative can be extended to differential forms of any order. In this article, in addition to the exterior derivative~\eqref{functExtder01} of one-forms, we will only need that of zero-forms, \ie of functionals $\mcF[\bphi]$. The latter is simply defined as the differential: 
\begin{equation}
\mbd \mcF \equiv \frac{\delta\mcF}{\delta \phi^{i\br}} \delta^{i\br} \ .
\end{equation}
Importantly, this functional exterior derivative retains a crucial property of its finite-dimensional counterpart: it squares to zero, \ie 
\begin{equation}
\mbd(\mbd\mcF)=0 \ ,
\label{function_ext_squared}
\end{equation}
thanks to the functional Schwarz theorem~\cite{dinelli2023non} of permutation of functional derivatives: $\frac{\delta \mcF}{\delta \phi^{i\br}\delta\phi^{j\br'}}=\frac{\delta \mcF}{\delta \phi^{j\br'}\delta\phi^{i\br}}$.
In other words, exact functional one-forms are closed, as in finite dimension.
This a geometric reformulation of the fact that eq.~\eqref{TRScond01} implies eq.~\eqref{TRScond02}, as proven in~\cite{dinelli2023non}, since the fact that eq.~\eqref{TRScond02} is distributional precisely means that $\mbd\bD^{-1}\ba_{\bphi}(\delta\bphi_1,\delta\bphi_2)$ must vanish for all $\delta\bphi_1,\delta\bphi_2$.
Before showing that the reciprocal also holds, let us make a remark that will turn out to be useful in practice. 
Using the antisymmetry of the wedge product, the exterior derivative~\eqref{functExtder02} can be re-written as
\begin{equation}
\mbd \bzeta =  \frac{\delta\zeta_{i\br}}{\delta\phi^{j\br'}}\delta^{j\br'}\wedge\delta^{i\br} \ .
\end{equation}
In turn, this relates the functional derivatives of two- and one- forms through the formula,
\begin{equation}
\mbd \bzeta = \mbd \left[ \zeta_{i\br}\delta^{i\br} \right] = \left[\mbd \zeta_{i\br} \right] \wedge \delta^{i\br} = \frac{\delta\zeta_{i\br}}{\delta\phi^{j\br'}}\delta^{j\br'}\wedge\delta^{i\br} \ 
\label{extDer_convenient}
\end{equation}
where $\mbd \zeta_{i\br} $ stands for the exterior derivative of $\zeta_{i\br}$ seen as a 0-form with $i$ and $\br$ fixed.
Formula~\eqref{extDer_convenient}, which is the infinite-dimensional counterpart of a well-known relation in finite dimension~\cite{lee2003introduction}, 
is very convenient for practical computations. 
We will use this relation later on in this article to compute functional exterior derivatives (see appendices~\ref{app:NR_flock} \&~\ref{app:active_Ising} for detailed calculations).

\paragraph{Functional Stokes theorem and Poincaré lemma.}
We are now in position to show that, when $\mbF$ is simply-connected, the converse of eq.~\eqref{function_ext_squared} also holds, \ie that closed functional one-forms are exact. This represents a functional version of the Poincar\'e lemma.

We show in appendix~\ref{App:subsec:Stokes} that our functional exterior derivative satisfies a functional version of the Stokes theorem: if $S\subset \mbF$ is a smooth oriented surface in $\mbF$ with smooth boundary $\partial S$, then for any one-form $\bzeta$ (see appendix~\ref{App:subsec:Stokes} for details), we have
\begin{equation}
\int_{\partial S} \bzeta = \int_S \mbd \bzeta \ .
\label{functional_Stokes}
\end{equation}
In section~\ref{subsubsec:loopWiseEP_AMB}, we will use this functional Stokes theorem to compute the entropy production around a loop\footnote{A loop in $\mbF$ corresponds to an evolution $(\bphi_{\tau})_{\tau\in [0,T]}$ of the field that ends up in the same state it started from, $\bphi_0=\bphi_T$.} in $\mbF$ for AMB and relate it to the phenomenology of this out-of-equilibrium model.
Most importantly, eq.~\eqref{functional_Stokes} allows us to prove that, when $\mbF$ is simply connected, the converse of relation~\eqref{function_ext_squared} also holds, \ie that a closed functional one-form is also exact. Indeed, let us consider a one-form $\bzeta$ that is closed, \ie $\mbd \bzeta=0$. Then, we define the functional
\begin{equation}
\mcF[\bphi] \equiv -\int_\mcC\bzeta \equiv -\int_{0}^{1} \zeta_{i\br}[\bphi_\tau] \partial_\tau\phi_\tau^{i\br} \rmd \tau \ ,
\label{GenericfreeEnergy}
\end{equation}
where $\bphi_0\in\mbF$ is arbitrary, and $\mcC\equiv (\bphi_\tau)_{\tau\in[0,1]}$ is an arbitrary smooth path in $\mbF$ connecting $\bphi_0$ to $\bphi_1\equiv \bphi$. 
First, the functional $\mcF[\bphi]$ is well defined, since picking another path $\mcC'\subset\mbF$ joining $\bphi_0$ to $\bphi$ in the definition~\eqref{GenericfreeEnergy} of $\mcF$ creates a difference with the right-hand side of~\eqref{GenericfreeEnergy} that is equal to the integral of $\bzeta$ along the loop $\mcC\cup\mcC'$. Since $\mbF$ is simply connected, there exist a surface $S\subset\mbF$ such that $\partial S =\mcC\cup\mcC'$. Using Stokes theorem~\eqref{functional_Stokes} and the fact that $\bzeta$ is closed, we conclude that this difference vanishes.
Furthermore, using the continuity of $\bzeta[\bphi]$ we have $\mcF[\bphi+\varepsilon\delta\bphi] \simeq \mcF[\bphi] - \varepsilon \zeta_{i\br}[\bphi] \delta\phi^{i\br}$ to first order in $\varepsilon$, \ie $\bzeta =-\mbd \mcF$, which is the functional version of the Poincar\'e lemma.

\paragraph{Irreversibility and vorticity.}

The functional Poincaré lemma implies that eq.~\eqref{TRScond01} follows from  eq.~\eqref{TRScond02}. Thus the reversibility of dynamics~\eqref{EDPS02} is equivalent to 
\begin{equation}
\mbd \bD^{-1}\ba = 0 \ .
\label{TRScond03}
\end{equation}
Indeed, if condition~\eqref{TRScond03} is satisfied, a direct computation allows to show that $\Pss[\bphi]\equiv e^{-\mcF[\bphi]}/Z$, where $\mcF$ is obtained by replacing $\bzeta$ by $\bD^{-1}\ba$ in eq.~\eqref{GenericfreeEnergy} and $Z\equiv\int e^{-\mcF[\bphi]}\mcD\bphi$ is the normalizing constant, is the stationary distribution and is such that~\eqref{TRScond01} is fulfilled.

The functional two-form $\mbd \bD^{-1}\ba$ associated with dynamics~\eqref{EDPS02} is the central object of this article. We will denote it by $\bomega$, \ie
\begin{equation}
\bomega \equiv \mbd \bD^{-1}\ba \ ,
\label{vorticity_definition}
\end{equation}
and, just as in~\cite{o2023nonequilibrium}, we will call it the \textit{(functional) vorticity two-form}, or, by analogy with its counterpart in Markov chains~\cite{qian2016entropy}, the \textit{(functional) cycle affinity two-form}.

\paragraph{Vorticity of the AMB.} 
In order to compute the vorticity two-form~\eqref{vorticity_definition} on a concrete example, one first needs to rewrite the considered stochastic field dynamics in the form~\eqref{EDPS02} and, specifically, identify the operator $\bD$ together with the drift $\ba$.
We already did so at the end of section~\ref{subsec:SPDEframework} for our running example, the Active Model B. In particular, using the same notation as in~\eqref{functOneFormsUsingDelta}, we have $\bD^{-1}\ba=-\mu_{\br}\delta^{\br}$.
Then, the definition~\eqref{functExtder01} of the functional exterior derivative -- or equivalently the convenient formula~\eqref{extDer_convenient} -- allows to compute the vorticity two-form~\eqref{vorticity_definition}.
In the case of the AMB dynamics~\eqref{AMB}-\eqref{AMB_chemical_potential}, the latter reads (see~\cite{o2023nonequilibrium} or appendix~\ref{app:AMB_vorticity})
\begin{equation}
\bomega \equiv \mbd \bD^{-1}\ba = -\mbd\left[\mu_{\br}\delta^{\br}\right] = -\int \rmd \br [2\lambda(\rho(\br)) +\kappa'(\rho(\br))]\nabla\rho(\br) \cdot \delta^\br\wedge \nabla\delta^\br \ ,
\label{cycleAffAMB}
\end{equation}
where $\kappa'\equiv\frac{\rmd \kappa}{\rmd \rho}$ and $\nabla\delta^\br$ is the gradient of the Dirac delta at $\br$, whose action on a perturbation $\delta\rho\in T_\rho\mbF$ is $\nabla\delta^\br(\delta\rho)\equiv -\nabla\delta\rho(\br)$. 
Note that the explicit expression~\eqref{cycleAffAMB} is not entirely straightforward for a reason detailed in the introduction of section~\ref{sec:basis} (see~\cite{o2023nonequilibrium} or appendix~\ref{app:AMB_vorticity} for two slightly different step-by-step derivations of~\eqref{cycleAffAMB}).
Dynamics~\eqref{AMB}-\eqref{AMB_chemical_potential} is then reversible iff $\bomega$ is identically zero, the latter condition being fulfilled iff $2\lambda + \kappa'=0$. Whenever this is the case, then the chemical potential $\mu$ is the functional derivative of the free energy $\mcF$ given by eq.~\eqref{GenericfreeEnergy} with $\bzeta\equiv \bD^{-1}\ba$, which also reads
\begin{equation}
\mcF[\rho] = \int \left[ \frac{\alpha}{2}\rho^2 + \frac{\beta}{4}\rho^4 + \frac{\kappa}{2}|\nabla\rho|^2 \right] \rmd \br \ ,
\end{equation}
up to an irrelevant constant.
Interestingly, note that in this example that the ``coordinates'' $\frac{\delta [D^{-1}a]_{i\br}}{\delta \phi^{j\br'}}$ of $\mbd\bD^{-1}\ba$ along the two-forms $\delta^{j\br'}\wedge\delta^{i\br}$ can be generalized functions\footnote{In the case of AMB $\frac{\delta [D^{-1}a]_{\br}}{\delta \rho^{\br'}}=[2\lambda(\rho(\br)) + \kappa'(\rho(\br))]\nabla_\br\rho(\br)\cdot \nabla_{\br'}\delta(\br-\br')$, up to unimportant symmetric (under $\br \leftrightarrow \br'$) terms.} of $\br$ and $\br'$, in which case the non-local ``basis'' $\delta^{j\br'}\wedge\delta^{i\br}$ can be reduced to a local one involving Dirac deltas at $\br$ together with their derivatives at the same point.
Finally, we see that the reversibility of AMB reduces to an integrability constraint over the chemical potential $\mu$. This might be surprising at first sight as it does not seem to enforce any structure of the noise, while it is known that the later should satisfy the fluctuation-dissipation theorem at equilibrium. This constraint is in fact embedded in the single condition~\eqref{TRScond03}, where the noise statistics enters through the operator $\bD$. In the AMB dynamics~\eqref{AMB}, the FDT is enforced by the relation between the drift and noise structure, so that $\bD^{-1}\ba$ reduces to $\mu$, and reversibility of~\eqref{AMB} amounts to the integrability of the latter functional.

\paragraph{Other existing formalisms.}

To conclude this section, we emphasize that the functional exterior calculus introduced in~\cite{o2023nonequilibrium} and further developed here is by no mean the first of its kind. Indeed, other mathematical formalisms exist (see~\cite{anderson1989variational,kriegl1997convenient} and references therein). 
The first advantage of the formalism presented here is its simplicity, as its only involve functional derivatives, Dirac deltas, and the wedge product of distributions.
A second advantage is its striking analogy with the finite-dimensional exterior calculus, which allows us to straightforwardly extend known results from finite dimension, something we will extensively do throughout this article.
Finally, a major inconvenience of the alternative formalism of jet bundles~\cite{anderson1989variational}, which seems to be the most developed and widespread mathematical framework for functional differential geometry, is that it cannot handle non-local differential forms, whereas our  formalism can (see section~\ref{subsubsec:oddSubspace} and appendix~\ref{app:active_Ising}).

\subsection{Two operators associated with $\omega$}
\label{subsec:vorticityOperator}

In this section, we introduce two differential operators, both associated with the two-form $\bomega$ but in different ways. They will respectively be important in the interpretation of $\bomega$ and of the results of section~\ref{sec:basis}.
In order to emphasize the distinction between these operators, they will respectively be called the cycle-affinity operator and the vorticity operator, although $\bomega$ was indifferently called the vorticity two-form and the cycle-affinity two-form.

\subsubsection{The cycle-affinity operator}
First, we define the \textit{cycle-affinity operator}, denoted by $\whomega$, as the $L^2$-representative of $\bomega$, \ie for all $\bphi\in \mbF$, $\whomega_{\bphi}$ is such that
\begin{equation}
\bomega_{\bphi}(\delta\bphi_1,\delta\bphi_2) = \int_{\mbR^{d_1}} \delta\bphi_1 \cdot \whomega_{\bphi}  \delta\bphi_2 d\br \ ,
\label{def_affinity_operator}
\end{equation}
where the dot `$\cdot$' here stands for the canonical scalar product of $\mbR^{d_2}$. Using expression~\eqref{functExtder01}, we see that in the canonical basis of $\mbR^{d_2}$ it acts as
\begin{equation}
\left[\whomega_{\bphi} \delta\bphi\right]_j(\br') = \sum_{i=1}^{d_2}\int \left\{ \frac{\delta[D^{-1}a_\phi]_{i\br}}{\delta\phi^{j\br'}}-\frac{\delta[D^{-1}a_\phi]_{j\br'}}{\delta\phi^{i\br}} \right\} \delta\phi^{i\br} \rmd\br \ .
\end{equation}
It directly follows from the skew-symmetry of $\bomega$ in its arguments that $\whomega$ is skew-symmetric in the $L^2$-sense: 
\begin{equation}
\whomega^\dagger=-\whomega \ .
\end{equation}
In section~\ref{subsec:skewsym_operators}, $\whomega$ will allow us to translate in the more familiar language of (skew-symmetric) differential operators the somewhat abstract result of section~\ref{subsec:basisOf2forms} which is formulated with functional 2-forms.

%
%
%

\subsubsection{The vorticity operator}
\label{subsubsec:vorticityOp}

\paragraph{Definition.}
We now define the \textit{vorticity operator} $\bW$ by formally extending to this field-theoretic context the finite-dimensional construction carried out in section 3.2.3 of our paper I~\cite{o2024geometric}.
In essence, it consist in considering $\bD^{-1}$ as a Riemannian metric over $\mbF$, which allows us to Taylor expand the vector field $\ba$ in a covariant way as:
\begin{equation}
\ba_{\exp_{\bphi}(\delta\bphi)}\simeq \tau_{\bphi\to\exp_{\bphi}(\delta\bphi)}\left[\ba_{\bphi}+\wt{\nabla}_{\delta\bphi}\ba_{\bphi}\right] 
\label{Taylor_covariant}
\end{equation}
up to linear order, where $\exp_{\bphi}$ here denotes the Riemannian exponential map\footnote{The Riemannian exponential map $\exp_{\bphi}$ sends a tangent vector $\delta\bphi$ at $\bphi$ to the point $\exp_{\bphi}(\delta\bphi)$ in $\mbF$ reached at time one by the geodesic that passes at time zero through $\bphi$ with speed $\delta\bphi$.}, $\wt{\nabla}$ the Levi-Civita covariant derivative, and $\tau_{\bphi\to\exp_{\bphi}(\delta\bphi)}$ the parallel transport from $\bphi$ to $\exp_{\bphi}(\delta\bphi)$ along the geodesic joining these two points in $\mbF$.
In turn, we define the vorticity operator as the skew-symmetric part, for the metric $\bD^{-1}$, of the linear operator $\wt{\nabla}\ba_{\bphi}:\delta\bphi\in T_{\bphi}\mbF\to \wt{\nabla}_{\delta\bphi}\ba_{\bphi}\in T_{\bphi}\mbF$ appearing in the expansion~\eqref{Taylor_covariant} of the vector field $\ba$: 
\begin{equation}
\bW_{\bphi}\delta\bphi \equiv \wt{\nabla}^A_{\delta\bphi}\ba_{\bphi} \ .
\end{equation}
This construction is similar to the one conducted in hydrodynamics to define the usual vorticity of a velocity field $\bv$, albeit here in a non-Euclidean and functional context. Hence the name chosen for $\bW$.
Note that the $\bD^{-1}$-antisymmetry of $\bW$ reads
\begin{equation} 
\bD\bW^\dagger\bD^{-1} = -\bW \ ,
\end{equation}
or, in coordinates, $
D^{i_1\br_1 i_2\br_2}W^{i_3\br_3}_{i_2\br_2}[D^{-1}]_{i_3\br_3 i_4\br_4} = -W^{i_1\br_1}_{i_4\br_4}$. This symmetry property notably implies that the spectrum of $\bW$ is purely imaginary -- as is that of $\whomega$.
 Most importantly, the vorticity operator can be shown to read (see~\cite{o2024geometric}):
\begin{equation}
W^{i_1\br_1}_{i_2\br_2}=-\frac{1}{2}D^{i_1\br_1 i_3\br_3}\left[ \frac{\delta [D^{-1}a]_{i_2\br_2}}{\delta\phi^{i_3\br_3}}-\frac{\delta [D^{-1}a]_{i_3\br_3}}{\delta\phi^{i_2\br_2}}\right] \ ,
\end{equation}
or, in terms of operators,
\begin{equation}
\bW  = -\frac{1}{2}\bD \whomega \ .
\label{Omega_vs_whomega}
\end{equation}
From eq.~\eqref{Omega_vs_whomega} and the invertibility of $\bD$, we conclude that dynamics~\eqref{EDPS02} is reversible iff $\bW=0$. This fact, together with the relation of $\bW$ to the covariant expansion~\eqref{Taylor_covariant} of $\ba$, suggests to interpret the vorticity operator as the irreversible part of dynamics~\eqref{EDPS02} at the linear level and, consequently, to use the following linear dynamics, that we call the \textit{vortex} (or \textit{vorticity}) \textit{dynamics}:
\begin{equation}
\partial_t \delta\bphi = \bW_{\bphi}\delta\bphi \ ,
\label{Vortex_dynamics}
\end{equation}
as a proxy to grasp the purely irreversible behavior of dynamics~\eqref{EDPS02} in the vicinity of the state $\bphi$. We now further justify this approach.

\paragraph{A proxy for the TRS-breaking phenomenology.}
As the theory of time-reversible stochastic time-evolution is now firmly developed, it is tempting to try to build a theory of TRS-breaking dynamics on this foundation, \ie to start from a stochastic field dynamics en route to equilibrium, add a (possibly small) TRS-breaking terms to it, and study how this additional term modifies the initial equilibrium picture.
 To do so, it is natural to try to split the drift $\ba$ of dynamics~\eqref{EDPS02} into two parts, 
\begin{equation}
\ba=\barev +\bairrev \ ,
\label{drift_decomp_rev_irrev}
\end{equation} 
the first one, $\barev$, ``preserving TRS'' and the other, $\bairrev$, ``breaking it''. This would enable to first study the reversible limit $\ba\to\ba^{\rm rev}$ of dynamics~\eqref{EDPS02} using tools from equilibrium field theory. Then one could study, for instance, the deterministic dynamics $\partial_t\bphi=\bairrev$ to get an insight into which kinds of intrinsically irreversible behavior come into play on top of the equilibrium limit in the full irreversible dynamics.

In the attempt of making decomposition~\eqref{drift_decomp_rev_irrev} properly defined, a first natural approach is the following:
on the one hand, ``preserving TRS'' means for $\barev$ to be such that, as soon as $\bairrev=0$, dynamics~\eqref{EDPS02} with drift $\ba=\barev$ is reversible, while on the other hand ``breaking TRS'' means that as soon as $\bairrev\neq 0$, dynamics~\eqref{EDPS02} becomes irreversible.
According to our results  from section~\ref{subsec:RevCondFuncExtDer}, these requirements amount to
\begin{eqnarray}
\mbd\bD^{-1}\barev &=& 0 \ , \label{contraints_arev_airrev_01}\\
 \mbd\bD^{-1}\bairrev &=& \mbd\bD^{-1}\ba \ .\label{contraints_arev_airrev_02}
\end{eqnarray}
Unfortunately, such a decomposition suffers a severe ambiguity.
Indeed the constraints~\eqref{contraints_arev_airrev_01}-\eqref{contraints_arev_airrev_02} are not sufficient to uniquely determine $\barev$ and $\bairrev$ since, starting from adequate candidates $(\barev,\bairrev)$ to~\eqref{drift_decomp_rev_irrev}-\eqref{contraints_arev_airrev_02}, and adding $\bD$ times the differential $\mbd\mcF$ of a non-constant functional $\mcF[\bphi]$ to the first and subtracting it from the second, \ie $(\wt{\barev}, \wt{\bairrev})\equiv (\barev +\bD\mbd\mcF, \bairrev-\bD\mbd\mcF)$, yields another distinct, and equally-valid, pair of candidates.

To remedy this issue,  one can choose $\barev$ and $\bairrev$ to be respectively the symmetric and antisymmetric part of $\ba$ under time-reversal, which we here denote by ${}^S\ba$ and ${}^A\ba$. Once again, the conventions introduced in section~\ref{sec:framework} allow us to straightforwardly generalize the known expressions of ${}^{S,A}\ba$ from finite dimension (see for instance our paper I~\cite{o2024geometric}) to field theory, yielding
\begin{eqnarray}
\barev\equiv {}^S\ba &=& \bD\mbd\ln \Pss \ , \label{arev_TRS-sym}\\
\bairrev\equiv{}^A\ba &=& \ba-\bD\mbd\ln \Pss \ , \label{airrev_TRS-antisym}
\end{eqnarray}
where $\Pss$ is the stationary probability density of dynamics~\eqref{EDPS02}. Clearly, this choice~\eqref{arev_TRS-sym}-\eqref{airrev_TRS-antisym} provide a unique decomposition~\eqref{drift_decomp_rev_irrev} while satisfying the necessary conditions~\eqref{contraints_arev_airrev_01}-\eqref{contraints_arev_airrev_02}. The problem now is that, since the explicit expression of $\Pss$ is generically unknown, those of $\barev$ and $\bairrev$ are also unknown, thus making the choice~\eqref{arev_TRS-sym}-\eqref{airrev_TRS-antisym} generally useless for practical purposes.

Another possibility to lift the degeneracy left by the constraints~\eqref{contraints_arev_airrev_01}-\eqref{contraints_arev_airrev_02}, is to perform a functional Hodge decomposition of the one form $\bD^{-1}\ba$. Since $\mbF$ is assumed to be simply connected, this decomposition should\footnote{As far as we now, this remains unproven for functional spaces.} read $\bD^{-1}\ba=\balpha^{\rm ex} +\balpha^{\rm coex}$,
where $\balpha^{\rm ex}$ and $\balpha^{\rm coex}$ are exact and co-exact one-forms\footnote{This means that there exist a functional $\mcF$ and a functional antisymmetric contravariant tensor of order two $A^{i\br j\br'}$ such that $\balpha^{\rm ex} = \mbd \mcF$ and $\alpha^{\rm coex}_{i\br}=-\frac{\delta A^{i\br j\br'}}{\delta \phi^{j\br'}}$ (the latter formula being valid in orthonormal coordinate systems in both $\mbR^{d_1}$ and $\mbR^{d_2}$).}, respectively. The resulting decomposition~\eqref{drift_decomp_rev_irrev} is such that
\begin{eqnarray}
\barev &\equiv & \bD\balpha^{\rm ex} \ ,\\
\bairrev &\equiv &  \bD\balpha^{\rm coex} \ ,
\end{eqnarray}
and again satisfies the necessary conditions~\eqref{contraints_arev_airrev_01}-\eqref{contraints_arev_airrev_02}. Unfortunately, we are not aware of any explicit formulas to perform such a decomposition in the context of functional one-forms. Thus this alternative to~\eqref{arev_TRS-sym}-\eqref{airrev_TRS-antisym} seems to suffer a similar intractability issue.

In this article, we circumvent this ambiguity of the decomposition~\eqref{drift_decomp_rev_irrev} as follows. 
To get an insight into the reversible part of the dynamics, we suggest to choose parameters in a given dynamics~\eqref{EDPS02} in such a way that $\mbd\bD^{-1}\ba=0$. For instance for AMB, this would amount to choosing $\lambda$ and $\kappa$ such that $2\lambda+\kappa'=0$. In general, the resulting equilibrium limit will be not unique, but if the hydrodynamics under study is simple enough, we do not expect the corresponding phenomenologies to differ too much from one equilibrium limit to the other. In practice, we will not systematically use this proxy of the ``reversible part'' of the dynamics, but rather focus on getting an insight into the ``irreversible part'' of the phenomenology.
We will do this by studying the vorticity dynamics~\eqref{Vortex_dynamics}. The right-hand side of eq.~\eqref{Vortex_dynamics} is 
a ``linear approximation of the irreversible drift of~\eqref{EDPS02}'' in the sense that it is not the full linearisation of the drift,
\begin{equation}
\frac{\delta \ba}{\delta\bphi}\delta \bphi = \bW \delta\bphi + \left[\frac{\delta \ba}{\delta\bphi}-\bW \right]\delta \bphi \ ,
\end{equation}
but rather only the part of it that is responsible for TRS breakdown: the anti-symmetric part of the covariant expansion of $\ba$ for the metric $\bD^{-1}$, $\bW\delta\bphi$.
The advantages of this approach are that it is often amenable to analytical calculation (as we will see later in section~\ref{sec:basis}) and that it does not suffer the ambiguity left by~\eqref{contraints_arev_airrev_01}-\eqref{contraints_arev_airrev_02}.
Indeed, using~\eqref{Omega_vs_whomega}, \eqref{def_affinity_operator} and~\eqref{airrev_TRS-antisym}, for any $\bairrev$ satisfying~\eqref{drift_decomp_rev_irrev}-\eqref{contraints_arev_airrev_02}, we have   
\begin{eqnarray*}
\bW\delta\bphi &=& -\frac{1}{2}\bD \whomega \delta\bphi \\
&=& -\frac{1}{2}\bD\mbd[\bD^{-1}\ba](\cdot,\delta\bphi) \\
&=& -\frac{1}{2}\bD\mbd[\bD^{-1}\bairrev](\cdot,\delta\bphi) \ ,
\end{eqnarray*}
where, for any functional two-form $\bgamma$ and vector field $\delta\bphi$, we denote by $\bgamma(\cdot,\delta\bphi)$ the one-form $\delta\bphi_0\mapsto \bgamma(\delta\bphi_0,\bphi)$.
Considering the vorticity dynamics~\eqref{Vortex_dynamics} lifts the ambiguity discussed above because the set of gauge choices authorized by the requirements~\eqref{drift_decomp_rev_irrev}-\eqref{contraints_arev_airrev_02} is precisely that of exact one-forms and thus coincides with the kernel of the exterior derivative: $\mbd[\bD^{-1}(\bairrev-\bD\mbd\mcF)]=\mbd\bD^{-1}\bairrev$.

A delicate question that arises in our linear approach is the choice of the base-point $\bphi$ in eq.~\eqref{Vortex_dynamics}.
In the spirit of the approach advocated below eq.~\eqref{drift_decomp_rev_irrev}, one could choose $\bphi$ as the average profile of the equilibrium limit. 
Another possibility, if the stationary probability of the full dynamics~\eqref{EDPS02} is concentrated in the vicinity of a profile $\bphi_{\rm ss}$, is to choose the latter as the base-point in eq.~\eqref{Vortex_dynamics}. In this case indeed, dynamics~\eqref{EDPS02} in steady-state fluctuates in the vicinity of $\bphi_{\rm ss}$ and the linear approx~\eqref{Vortex_dynamics} with base-point $\bphi_{\rm ss}$ describes the irreversible behavior of these fluctuations. The effectiveness of our approach in such a case is exemplified below using AMB.

Nevertheless, as we will see in section~\ref{subsec:basisPhenomeno}, our approach is relevant beyond this peculiar situation where the stationary measure is concentrated in the vicinity of a certain profile. 
In this article in particular, the choice of profile $\bphi$ in the vicinity of which dynamics~\eqref{Vortex_dynamics} is analyzed, will not be prominent, as we will study the generic dynamics~\eqref{EDPS02} in two limit-regimes (described below) in which the exact structure of $\bphi$ is not crucial, but only allows for a more precise analysis.

\paragraph{Connection of $\bW$ with the Stratonovitch-averaged dynamics and instantons.}

Directly generalizing a finite-dimensional result~\cite{o2024geometric}, the operator $\bW$ interestingly plays a similar role for the Stratonovitch-averaged dynamics of~\eqref{EDPS02}, which reads $\partial_t\bphi = \bvss[\bphi]$, with $\bvss\equiv \bJss/\Pss=\ba-\bD\mbd\ln\Pss={}^A\ba$, \ie it is the $\bD^{-1}$-antisymmetric part of the covariant linear expansion of $\bvss$.

Besides, it worth noting that the vorticity operator $\bW$ generates the irreversible component of the instanton dynamics associated with~\eqref{EDPS02}.
Indeed, the path-probability density of dynamics~\eqref{EDPS02}, conditioned to a given initial state $\bphi_0$, satisfies
\begin{equation}
\mcP[(\bphi_t)_{t\in\mbT}] \propto e^{-S[(\bphi_t)_{t\in\mbT}]}
\end{equation}
There exist several approaches to build a covariant action $\mcS$ (see~\cite{de2023path} and references therein).
Again generalizing results from finite dimension, the approaches of DeWitt~\cite{dewitt1957dynamical} and Graham~\cite{graham1977path} in particular give actions that read
\begin{equation}
S[(\bphi_t)_{t\in\mbT}] = \int_0^{\mcT} \left[ \frac{1}{4}D_{i\br j\br'} (\dot{\phi}^{i\br} - a^{i\br}_{\blambda})(\dot{\phi}^{j\br'} - a^{j\br'}_{\blambda}) + \frac{1}{2}\DDiv_{\blambda}(\ba_{\blambda}) + \frac{c}{2}R \right] \ ,
\end{equation}
where the dot stands for partial differentiation with respect to time, $c=1/3$ and $1/6$ in the DeWitt and Graham approaches respectively, $\blambda$ is now the volume measure on $\mbF$ associated with the Riemannian metric $\bD^{-1}$, and $\DDiv_{\blambda}=\det(\bD)\frac{\delta}{\delta\phi^{i\br}} \det(\bD^{-1})$ the corresponding divergence operator. Extremizing the action then leads to the instanton dynamics:
\begin{equation}
\ddot{\phi}^{i_1\br_1} + \Gamma^{i_1\br_1}_{i_2\br_2 i_3\br_3}\dot{\phi}^{i_2\br_2}\dot{\phi}^{i_3\br_3} = D^{i_1\br_1 i_2\br_2} \frac{\delta V}{\delta\phi^{i_2\br_2}} + 2W^{i_1\br_1}_{i_2\br_2}\dot{\phi}^{i_2\br_2} \ ,
\label{instanton_dynamics}
\end{equation}
where the $\Gamma^{i_1\br_1}_{i_2\br_2 i_3\br_3}$ are the Christoffel symbols associated to the $\bD^{-1}$-Levi-Civita connection on $\mbF$, and the functional $V[\bphi]$ reads
\begin{equation}
V[\bphi]\equiv \frac{1}{2}\ba_{\blambda}\cdot\bD^{-1}\ba_{\blambda} + \DDiv_{\blambda}(\ba_{\blambda}) + cR \ .
\end{equation}
We see that reversing the time variable in the instanton eq.~\eqref{instanton_dynamics} flips the sign of the last term on the right-hand side. In other word, this last term, which is generated by twice the vorticity operator $\bW$, is the deterministically time-antisymmetric part of the instanton dynamics~\eqref{instanton_dynamics} associated with dynamics~\eqref{EDPS02}. In eq.~\eqref{instanton_dynamics}, $2\bW$ plays the role of a (cross product by a) magnetic field whose infinite-dimensional vector potential is $\bD^{-1}\ba$ (for the geometry given by $\bD^{-1}$). Interestingly, the irreversible part of the instanton dynamics, is not approximated at the linear level by (twice) the flow of $\bW$, but exactly coincides with it.

\if{
\MyPink{
From eq.~\eqref{Omega_vs_whomega} and the reversibility of $\bD$, we conclude that dynamics~\eqref{EDPS02} is reversible iff $\bW=0$. This fact, together with the relation of $\bW$ to the covariant expansion~\eqref{Taylor_covariant} of $\ba$, suggests to interpret the vorticity operator as the irreversible part of dynamics~\eqref{EDPS02} at the linear level and, consequently, to use the linear dynamics
\begin{equation}
\partial_t \delta\bphi = \bW_{\bphi}\delta\bphi 
\label{Vortex_dynamics}
\end{equation}
as a proxy to grasp the purely irreversible behavior of dynamics~\eqref{EDPS02} in the vicinity of the state $\bphi$. 
We emphasize that eq.~\eqref{Vortex_dynamics} is the linear approximation\footnote{To draw a parallel with a more intuitive situation, let us consider a finite-dimensional dynamics in $\mbR^3$ with $\bD$ being the identity matrix. The Helmholtz-Hodge decomposition of the vector field $\ba$ that drives the dynamics reads $\ba=\ba_g+\ba_c$, where $\ba_g$ and $\ba_c$ are pure gradient and curl, respectively. The dynamics would be reversible if $\ba_c$ was identically zero. Otherwise, the effect of this nonequilibrium term on the dynamics can be depicted, in the vicinity of a point $\br_0$, by the flow of the linear vector field $[\nabla\times\ba]|_{\br_0}\times (\br-\br_0)=[\nabla\times\ba_c]|_{\br_0}\times (\br-\br_0)$.} of the irreversible part of the full non-linear dynamics~\eqref{EDPS02}, and not the irreversible part of the linearized dynamics.
Thus, at least in cases where the steady-state solution of the full dynamics~\eqref{EDPS02} is concentrated in the vicinity of a certain profile $\bphi_{\rm ss}$ -- as \eg in AMB --, one can expect the phenomenology of~\eqref{EDPS02} to be heuristically
that of its reversible limit $\bW=0$ (which may be analyzed by standard techniques of equilibrium systems) on top of which fluctuations will roughly follow the ``vortex dynamics''~\eqref{Vortex_dynamics}, the irreversibility of the dynamics being apparent in this last part -- which we note is by itself an irreversible dynamics. 
}

\Magenta{
Modify according to discussion with Mike. In particular, add the relation with the instanton dynamics, whose (deterministic) irreversible term is $\propto\bW \dot{\bphi}$.
}

Besides, let us also mention that, directly generalizing a finite-dimensional result~\cite{o2024geometric}, $\bW$ plays a similar role for Stratonovitch average dynamics of~\eqref{EDPS02}, which reads $\partial_t\bphi = \bvss[\bphi]$, with $\bvss\equiv \bJss/\Pss=\ba-\bD\mbd\ln\Pss$, \ie it is the $\bD^{-1}$-antisymmetric part of the covariant linear expansion of $\bvss$.

Below we exemplify using AMB the effectiveness using~\eqref{Vortex_dynamics} as a proxy for the irreversible phenomenology of dynamics~\eqref{EDPS02}. Note that the steady state of AMB is indeed concentrated in the vicinity of a given profile (either uniform or phase-separated depending on the parameter values). 
In such cases, it is not so surprising that our linear approximation~\eqref{Vortex_dynamics} gives a good approximation of the manifestation of TRS breakdown on the system. However, more surprisingly, we will see in section~\ref{subsec:basisPhenomeno} that this approach also provides valuable insights when the steady state of the system is \MyPink{not anymore localized around a fixed profile} \Magenta{[Mike suggests exemplifying using KPZ too.]}.


}\fi

\paragraph{Two limit-regimes amenable to analytical predictions.}

Before turning to the example of AMB, let us first discuss the generic dynamics~\eqref{Vortex_dynamics}. 
We here only consider the case where $\bW_{\bphi}$ is a \textit{local differential operator} for all $\bphi$, in which case it generically reads
\begin{equation}
\bW_{\bphi} = \sum_K \bW^K(\br,[\bphi]) \frac{\partial^{|K|}}{\partial r^K} \ ,
\label{generic_vorticity_op}
\end{equation}
where $K=(k_1,\dots,k_{d_1})$ is a multi-index of natural numbers, with the conventions
 $\sum_K\equiv \sum_{k_1,\dots,k_{d_1}}$ (each sum being finite) and 
\begin{equation}
\frac{\partial^{|K|}}{\partial r^K}\equiv \frac{\partial^{k_1+\dots +k_{d_1}}}{\partial r_1^{k_1}\dots \partial r_{d_1}^{k_{d_1}}} \ .
\end{equation}
In eq.~\eqref{generic_vorticity_op}, $\bW^K$ are $d_2\times d_2$ matrices that, for a given $\bphi$, vary with $\br$ on a typical spatial scale $\ell_{W}$, which is that of typical variations in the field $\bphi$ itself, so long as the original dynamics~\eqref{EDPS02} is translation-invariant.
Dynamics~\eqref{Vortex_dynamics} has no reason to be solvable explicitly in general. However, there are two interesting limit regimes where the solution of~\eqref{Vortex_dynamics} can be approximated.
The first one corresponds to the case where the chosen perturbation $\delta\bphi$ varies on a scale $\ell_{\delta\phi}$ that is much larger than $\ell_W$, \ie $\ell_{\delta\phi} \gg \ell_W$. In this case, the dominant term in $\bW_{\bphi} \delta\bphi$ is the zero-order one and the dynamics can be approximated by $\partial_t\delta\bphi(\br,t)\simeq \bW^{(0,\dots ,0)}(\br,[\bphi])\delta\bphi(\br,t)$ and hence explicitly solved. Note that, as soon as $\bW^{(0,\dots ,0)}$ is not purely antisymmetric, then $\delta\bphi$ will tend to increase on some spatial domains and decrease on others in such a way that, after some time, $\ell_{\delta\phi}\sim\ell_W$ and the approximation will no longer be valid.
On the other hand, if $\delta\bphi$ is now such that $\ell_{\delta\phi} \ll \ell_W$, the dominant terms in $\bW_{\bphi}\delta\bphi$ are those of maximal order in derivation, and the corresponding matrices $\bW^{(k_1,\dots ,k_{d_1})}$ can be considered piece-wise constant on the spatial domain. Consequently, on each of these subdomains, the resulting approximation of $\bW$ is independent of $\br$ and thus amenable to Fourier analysis.
The existence of these two limit regimes -- that correspond to diagonal approximations of~\eqref{Vortex_dynamics} in real and Fourier space, respectively -- will be exploited several times in the rest of this article.

\paragraph{Vorticity operator and phenomenology: the example of AMB.} 
In~\cite{o2023nonequilibrium}, the vorticity two-form of AMB was used to unveil the leading-order (in weak noise amplitude) phenomenon that breaks TRS in the phase separated state.
 The latter corresponds to the permanent excitation of anisotropic modes propagating at the liquid-gas interface, either from the liquid to the gas, or the other way around, depending on the sign of $2\lambda+\kappa'$. This was shown in~\cite{o2023nonequilibrium} through the intermediate step of the analysis, \via $\bomega$, of the structure of the stationary probability current.

Interestingly, this phenomenology can also be uncovered -- and even specified -- by the dynamics~\eqref{Vortex_dynamics} generated by the vorticity operator. In the case of the AMB dynamics~\eqref{AMB}-\eqref{AMB_chemical_potential}, using eqs.~\eqref{cycleAffAMB}, \eqref{def_affinity_operator}, and~\eqref{Omega_vs_whomega}, the latter is readily shown to act as (see appendix~\ref{app:AMB_vorticity} for the detailed derivation):
\begin{equation}
\bW_\rho\delta\rho =\frac{1}{2}M \Delta \left\{ (2\lambda+\kappa')\nabla\rho \cdot \nabla\delta\rho + \nabla\cdot \left[\delta\rho (2\lambda+\kappa')\nabla\rho \right]\right\} \ ,
\label{full_vorticity_op_AMB}
\end{equation}
where $\delta\rho$ is a perturbation around a given profile $\rho$.
In the vicinity of the average phase-separated profile $\rhoss$, we start by focusing on the second of the two limit-regimes mentioned above, \ie we consider a perturbation $\delta\rho$ that varies on much smaller spatial scales than $\rhoss$ does. The action of $\bW_{\rhoss}$ on $\delta\rho$ can then be approximated as:
\begin{equation}
\bW_{\rhoss}\delta\rho \simeq M (2\lambda+\kappa')\nabla\rhoss \cdot \nabla\Delta\delta\rho \ .
\label{approx_vortex_dyn_AMB}
\end{equation}
First, we see that, at least at this linear level, the effect of activity is localized at the liquid-gas interface, where $\nabla\rhoss\neq 0$, as reported in~\cite{nardini2017entropy,o2023nonequilibrium}.
Further, eqs.~\eqref{Vortex_dynamics} and~\eqref{approx_vortex_dyn_AMB} recover the aforementioned propagating-wave phenomenology uncovered in~\cite{o2023nonequilibrium}. This is most readily seen in dimension $d_1=1$ by injecting a harmonic perturbation $\delta\rho\propto e^{i(wt-kx)}$ in eqs.~\eqref{Vortex_dynamics}-\eqref{approx_vortex_dyn_AMB}, leading to the dispersion relation:
\begin{equation}
w = M (2\lambda+\kappa')(\partial_x\rhoss) k^3 \ .
\label{AMB_irrev_dispersion_rel}
\end{equation}
From eq.~\eqref{AMB_irrev_dispersion_rel}, we indeed recover the fact that waves propagate from the gas to the liquid if $2\lambda+\kappa'>0$, and the other way around if $2\lambda+\kappa'<0$, as predicted and measured in~\cite{o2023nonequilibrium}. But it is worth noting that this method -- of analyzing the leading-order irreversible phenomenon through the vorticity operator --
goes further than the one used in~\cite{o2023nonequilibrium} as it also predicts the dispersion relation~\eqref{AMB_irrev_dispersion_rel} of these waves.

Let us now examine the opposite limit-regime, \ie when $\delta\rho$ varies on much larger scales than $\rhoss$ does. Assuming $2\lambda+\kappa'$ constant for simplicity, the linear dynamics~\eqref{Vortex_dynamics} becomes:
\begin{equation}
\partial_t\delta\rho= \delta\rho M(2\lambda+\kappa') \Delta^2\rhoss \ ,
\label{real_space_diagonal_limit_equation}
\end{equation}
the solution of which reads
\begin{equation}
\delta\rho(x,t)=\delta\rho(0,x)\exp\left[M(2\lambda+\kappa') \Delta^2\rhoss(x) \right] \ .
\label{real_space_diagonal_limit_solution}
\end{equation}
Consequently, if $2\lambda+\kappa'>0$ for instance, $\delta\rho$ will tend to accumulate on the liquid side of the interface and be depleted on the gaseous side, thus ``compressing'' the liquid droplet (see fig.~\ref{fig:steadyState_and_derivatives}). Conversely, if $2\lambda+\kappa'<0$, the liquid droplet will tend to spread. 
Of course, after some time, the typical spatial scale of $\delta\rho$ will no longer be much greater than that of $\rhoss$ and equation~\eqref{real_space_diagonal_limit_equation} will cease to be a correct approximation, hence preventing $\delta\rho$ from blowing up at long times. 
Moreover, in the full dynamics~\eqref{AMB}-\eqref{AMB_chemical_potential}, the counterpart of the term of the right-hand side of eq.~\eqref{real_space_diagonal_limit_equation}, which pumps mass from or towards the droplet, will be counterbalanced by a reversible term that stabilizes the stationary profile.
These predictions of the influence of activity on the droplet size are in accordance with the binodal shift, caused by activity, in AMB~\cite{o2023nonequilibrium,wittkowski2014scalar}.


We conclude this section by emphasizing that, although we used a linear equation to conduct our analysis, the naive linearization of dynamics~\eqref{AMB}-\eqref{AMB_chemical_potential} alone (in which reversible and irreversible ingredients are entangled) couldn't have been directly used to predict the leading order irreversible phenomenon around a phase-separated profile as we just did.

\begin{figure}[h!]
\begin{center}
\includegraphics[scale=0.5]{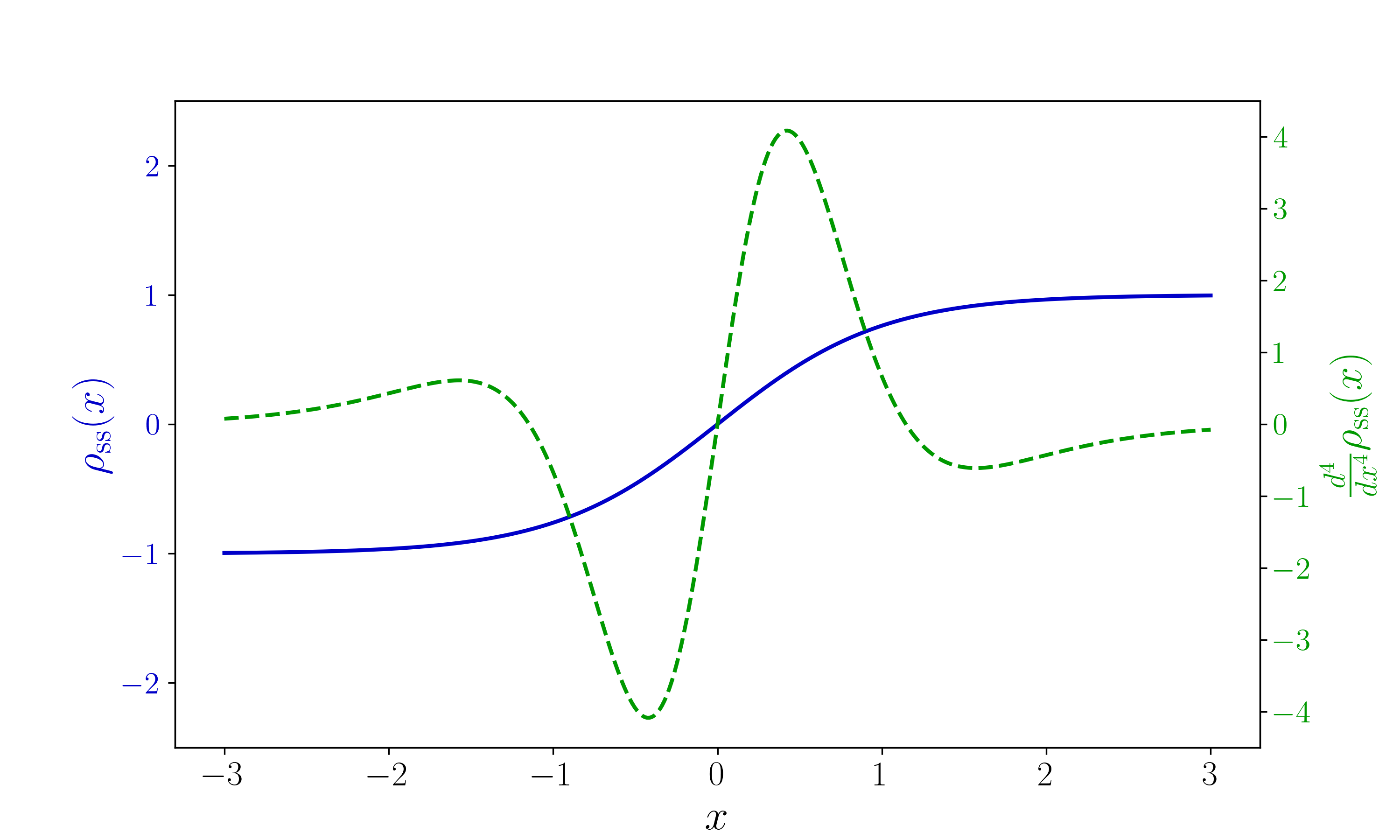}
\caption{\label{fig:steadyState_and_derivatives} 
A schematic picture of the left boundary of the phase-separated, steady-state, average profile $\rhoss$ of AMB and its derivative of order four. 
We see that the fourth derivative $\frac{d^4}{dx^4}\rhoss$, which coincides with $\Delta^2\rhoss$ is $d_1=1$ dimension, displays a positive (respectively negative) peak on the liquid (respectively gaseous) side of the interface. This justifies the qualitative interpretation of eq.~\eqref{real_space_diagonal_limit_solution} above. As this is only for illustrative purpose we took the function $\tanh(x)$ instead of a real estimate of $\rhoss(x)$ (the two curves in fact becoming identical in the limit $2\lambda+\kappa' \to 0$).
}
\end{center}
\end{figure}

\newpage

\section{From explicit reversibility condition(s) to classification of nonequilibrium field theories}
\label{sec:basis}

As announced in the introduction, the integrability condition~\eqref{intro_Schwarz} for a functional one-form $\bzeta(\br,[\bphi])$, or even its more geometric version $\mbd \bzeta=0$, might not be straightforward to use when looking for the concrete reversibility condition(s) of a given field dynamics. 
This issue was not visible on the AMB example reported in section~\ref{subsec:RevCondFuncExtDer} because we directly gave a convenient expression~\eqref{cycleAffAMB} of its vorticity two-form $\bomega\equiv \mbd \bD^{-1}\ba$. Indeed, the explicit calculation of $\bomega$, detailed in~\cite{o2023nonequilibrium}, first leads to the expression
\begin{equation}
\bomega = -\int \rmd \br \, 2\lambda \nabla\rho \cdot\delta^\br \wedge  \nabla\delta^\br - \int \rmd \br \, \kappa \delta^\br\wedge \Delta\delta^\br \ .
\label{problematic_2form}
\end{equation}
At this stage, it would be wrong to conclude that, for $\bomega$ to vanish, $\lambda(\rho)$ and $\kappa(\rho)$ both need to be identically zero.
This is because the elementary two-forms $\delta^\br \wedge \nabla \delta^\br$ and $\delta^\br \wedge \Delta\delta^\br$ are not \textit{independent}\footnote{What we here mean by ``independent'' will be made more precise in section~\ref{subsec:basisOf2forms}.} as they differ by a total derivative. Indeed, for an arbitrary pair of fluctuations $\delta\rho_1,\delta\rho_2$, 
\begin{eqnarray*}
\delta^\br \wedge \Delta\delta^\br(\delta\rho_1,\delta\rho_2)&\equiv & \delta\rho_1^\br\Delta\delta\rho_2^\br - \delta\rho_2^\br\Delta\delta\rho_1^\br \\
&=& \nabla\cdot \left[ \delta\rho_1^\br\nabla\delta\rho_2^\br - \delta\rho_2^\br\nabla\delta\rho_1^\br \right] \\
&=& -\nabla\cdot \left[ \delta^\br\wedge\nabla\delta^\br(\delta\rho_1,\delta\rho_2)\right] \ ,
\end{eqnarray*}
where the minus sign in the last expression comes from the sign convention of the gradient of the Dirac delta: $\nabla\delta^\br(\delta\rho)=-\nabla\delta\rho^\br$. This relation, that we may summarize as 
\begin{equation}
\delta^\br \wedge \Delta\delta^\br = -\nabla\cdot [\delta^\br \wedge \nabla\delta^\br] \ ,
\label{relation_a_generaliser}
\end{equation}
can then be injected into eq.~\eqref{problematic_2form}. After integrating by parts the factorized gradient in the resulting equation, one gets expression~\eqref{cycleAffAMB} of the vorticity two form, from which one can easily show that $\bomega=0$ (\ie $\bomega_\rho=0$ for all $\rho$) iff $2\lambda+\kappa'=0$, which is the correct integrability condition for $D^{-1}a=-\mu$ and hence the right reversibility condition for Active Model B.

The example of AMB shows that, at least for a field dynamics~\eqref{EDPS02} whose one-form $\bzeta\equiv \bD^{-1}\ba$ is expressed as a gradient expansion in $\bphi$, one needs to express the vorticity $\bomega$ as a sum of integrals against elementary two-forms that are independent -- in the appropriate sense -- in order to determine the concrete reversibility condition(s).
To be even more systematic, we would ideally like to determine a \textit{basis} of 2-forms, \ie a family of independent elementary two-forms on which $\mbd\bzeta$ can always be developed, whatever $\bzeta(\br,[\bphi])$.
In section~\ref{subsec:basisOf2forms}, we completely solve this problem in $d_1=1$ spatial dimension for the case of a one-form $\bzeta(x,[\bphi])$ that is local in $\bphi$.
We then partially extend in~\ref{subsec:higher_dimension_d1} our result to higher dimensions by determining a free family of two-forms. The latter is incomplete in the sense that it does not span the entire space of vorticity two-forms of arbitrary field dynamics. However this family turns out to be sufficient in many practical cases.

To shed light on our notion of a basis of two-forms, we translate it in the more common language of differential operators in section~\ref{subsec:skewsym_operators}. More precisely, we use the notion of the affinity operator~\eqref{def_affinity_operator} introduced in section~\ref{subsec:vorticityOperator} to show that this basis of functional 2-forms is in one-to-one correspondence with a basis of antisymmetric operators.

To get an insight into the phenomenology associated to each basis element, we then analyse in section~\ref{subsec:basisPhenomeno} the linear dynamics~\eqref{Vortex_dynamics} generated by each respective vorticity operator.
This approach will show that these elementary two-forms can be classified into three subfamilies whose respective members have similar phenomenologies.
For each of these subfamilies, we also give concrete examples whose known phenomenology are in good agreement with our analysis based on the 
flow of their vorticity operator as captured by~\eqref{Vortex_dynamics}.

Finally, in section~\ref{subsec:discussion_extension}, we summarize our results, discuss their limitations, and elaborate on their possible extension to more general situations.

\newpage

\subsection{A basis of 2-forms in one spatial dimension}
\label{subsec:basisOf2forms}

\paragraph{Local one-forms and two-forms.}
At the end of section~\ref{sec:framework}, we defined a notion of locality for a map $\mu(\br,[\bphi])$ that is both a function of $\br$ and a functional of $\bphi$. We now extend this notion for one- and two- forms.
From now on, a functional one-form $\bzeta_{\bphi}=\zeta_{i\br}[\bphi]\delta^{i\br}$ is said to be a \textit{local one-form} when its coordinate function(al)s $\zeta_{i\br}[\bphi]$ are local maps. We denote by $\Omega^1_{\rm loc}(\mbF)$ the space of local one-forms over $\mbF$.
We define the notion of local two-forms in a slightly different way: a two-form $\bgamma_{\bphi}=\gamma_{i\br j\br'}[\bphi] \delta^{i\br}\wedge\delta^{j\br'}$ is a \textit{local two-form} when it can be written as 
\begin{equation}
\bgamma=\sum_K\gamma^{(K)}_{ij\br}[\bphi]\delta^{i\br}\wedge \frac{\partial^{|K|}}{\partial r^K}\delta^{j\br} \ ,
\label{generic_local_two-form}
\end{equation} 
where the $\gamma^{(K)}_{ij\br}[\bphi]$ depend smoothly on $\bphi$ and $\br$, but do not have to be local as maps\footnote{We could alternatively require the $\gamma^{(K)}_{ij\br}[\bphi]$ to be local as maps. All the results presented in this article would still be valid. But  as this constraint is not mandatory for our study, we lift it from the definition of local two-forms to make our results slightly more general.} of the variables $\br$ and $\bphi$. 
In other words, $\bgamma$ is local two-form whenever $\bgamma_{\bphi}(\delta\bphi_1,\delta\bphi_2)$ is an integral over $\br$ of a function that depends only on the value of each $\delta\bphi_i$ and its derivatives (up to a finite order) at $\br$, whatever the functional dependence of $\bgamma_{\bphi}$ with respect to $\bphi$. 
In eq.~\eqref{generic_local_two-form}, we use the same notations as in eq.~\eqref{generic_vorticity_op} and we implicitly sum over $i,j$ and integrate over $\br$ (we integrate only once, despite $\br$ appearing twice as an upper index).
Note that the relation between $\gamma_{i\br j\br'}$ and $\gamma^{(K)}_{ij\br}$ reads
\begin{equation}
\gamma_{i\br j\br'}=\sum_K\gamma^{(K)}_{ij\br}(-1)^{|K|}\frac{\partial^{|K|}}{\partial r'^K}\delta(\br-\br') \ .
\end{equation}
We denote by $\Omega^2_{\rm loc}(\mbF)$ the space of local two-forms over $\mbF$.

Before turning to the space of local exact two-forms $\mbd\Omega^1_{\rm loc}(\mbF)$, to which the vorticity $\bomega$ belongs when $\bD^{-1}\ba$ is local, we first study in detail the space of local two-forms $\Omega^2_{\rm loc}(\mbF)$, the latter space containing the former:
\begin{equation}
\mbd\Omega^1_{\rm loc}(\mbF)\subset \Omega^2_{\rm loc}(\mbF) \ .
\label{inclusion_exact_two-forms}
\end{equation}

\paragraph{A basis of the space of local two-forms $\Omega^2_{\rm loc}(\mbF)$.}
Until the end of section~\ref{subsec:basisOf2forms}, we focus on the case of $d_1=1$ spatial dimension, where a local two-form generically reads
\begin{equation}
\bgamma=\sum_{k=0}^q\gamma^{(k)}_{ij x}[\bphi]\delta^{i x}\wedge \partial^k\delta^{j x} \ ,
\label{generic_local_two-form_1d}
\end{equation} 
with the abbreviation $\partial^k\equiv\frac{d^{k}}{d x^k}$ is used from now on to lighten notations.
This expression of a generic local one-form $\bgamma$ suffers a redundancy issue as it involves, for instance, the elementary two-forms $\delta^{i x}\wedge \partial^{2}\delta^{i x}$, which can be ``decomposed'' along $\delta^{i x}\wedge \partial\delta^{i x}$ upon factorizing a spatial derivative~\eqref{relation_a_generaliser} and integrating by parts, as explained in the introduction of section~\ref{sec:basis}. 
We now solve this redundancy problem by determining a family of elementary two-forms, that are independent of each other in the appropriate sense, and on which any local two-form can be uniquely decomposed.
We do this by first showing in appendix~\ref{app:proof_Basis_form} that (the $d_1=1$ version of) relation~\eqref{relation_a_generaliser} generalizes as follows:
\begin{eqnarray}
\delta^i\wedge\partial^{2\ell}\delta^j + \delta^j\wedge\partial^{2\ell}\delta^i  &=& - \sum_{k=0}^{\ell-1}b_{2k+1}^{2\ell}\partial^{2\ell-2k-1}_x\left[ \delta^i\wedge\partial^{2k+1}\delta^j + \delta^j\wedge\partial^{2k+1}\delta^i \right]  \ ,
\label{core:decompo_elementary_2forms_plus} \\
\delta^i\wedge\partial^{2\ell+1}\delta^j - \delta^j\wedge\partial^{2\ell+1}\delta^i  &=& - \sum_{k=0}^{\ell}c_{2k}^{2\ell+1}\partial^{2\ell-2k+1}_x \left[ \delta^i\wedge\partial^{2k}\delta^j - \delta^j\wedge\partial^{2k}\delta^i \right] \ ,
\label{core:decompo_elementary_2forms_minus}
\end{eqnarray}
the coefficients $b^{2\ell}_{2k+1}$ and $c^{2\ell+1}_{2k}$ being given by eqs.~\eqref{expr_coeff_b}-\eqref{expr_coeff_a} of appendix~\ref{app:factorisation_operators}.
In turn, injecting these relations back into eq.~\eqref{generic_local_two-form_1d}, integrating by parts several times, and rearranging the terms (see appendix~\ref{app:proof_Basis_form} for details) allows us to prove that the generic local two-form $\bgamma$ can be re-written as
\begin{eqnarray}
\bgamma &=& \sum_{\ell=0}^{\lfloor q/2 \rfloor} \sum_{1\leq i< j \leq d_2} \int \rmd x \: \alpha_{ij}^{(\ell)}(x,[\bphi]) \left(\delta^{ix}\wedge \partial^{2\ell}\delta^{jx} - \delta^{jx}\wedge \partial^{2\ell}\delta^{ix} \right) \notag\\
& + &  \sum_{\ell=0}^{\lfloor (q-1)/2\rfloor} \Bigg\{ \sum_{1\leq i < j\leq d_2}\int \rmd x \: \beta_{ij}^{(\ell)}(x,[\bphi])\left( \delta^{ix}\wedge\partial^{2\ell+1}\delta^{jx} + \delta^{jx}\wedge\partial^{2\ell+1}\delta^{ix}\right) \label{eq:generalLoc2Form01} \\
& &   \quad \qquad + \sum_{i=1}^{d_2} \int \rmd x \: \beta_{ii}^{(\ell)}(x,[\bphi]) \delta^{ix}\wedge\partial^{2\ell+1}\delta^{ix} \Bigg\} \notag \ ,
\end{eqnarray}
where the ``coordinate functions'' $\alpha_{ij}^{(\ell)}(x,[\bphi])$ and $\beta_{ij}^{(\ell)}(x,[\bphi])$ respectively read
\begin{equation}
 \alpha_{ij}^{(\ell)}= {}^A\gamma^{(2\ell)}_{ij} + \sum_{k=\ell}^{\left \lfloor (q-1)/2 \right \rfloor}c^{2k+1}_{2\ell}\partial_x^{2k+1-2\ell} \left[{}^A\gamma_{ij}^{(2k+1)}\right]
\label{eq:expressionAlpha_two-form}
\end{equation}
and
\begin{equation}
\beta_{ij}^{(\ell)}= {}^S\gamma_{ij}^{(2\ell+1)}+\sum_{k=\ell+1}^{\left \lfloor q/2 \right \rfloor}b^{2k}_{2\ell+1}\partial_x^{2k-2\ell-1}\left[ {}^S\gamma_{ij}^{(2k)} \right] \ ,
 \label{eq:expressionBeta_two-form}
\end{equation}
the matrices ${}^{S/A}\bgamma^{(k)}$ being respectively the symmetric and antisymmetric parts of $\bgamma^{(k)}$, \ie ${}^{S/A}\gamma_{ij}^{(k)}\equiv (\gamma^{(k)}_{ij} \pm \gamma^{(k)}_{ji})/2$.
This result implies that the family of functional 2-forms given by
\begin{equation}
(\underbrace{[\delta^{ix}\wedge \partial^{2\ell}\delta^{jx} - \delta^{jx}\wedge \partial^{2\ell}\delta^{ix}]}_{\rm antisymmetric}, \underbrace{[\delta^{ix}\wedge\partial^{2\ell+1}\delta^{jx} + \delta^{jx}\wedge\partial^{2\ell+1}\delta^{ix}], \delta^{ix}\wedge\partial^{2\ell+1}\delta^{ix}}_{\rm symmetric})_{1 \leq i < j \leq d_2, \ell\in \mbN,x\in\mbR} \ 
\label{1dBasis2Forms}
\end{equation}
is \textit{generative} for the space $\Omega^2_{\rm loc}(\mbF)$, the ``linear combinations'' over this family consisting in  summing and integrating over its discrete and continuous indices, respectively. 
Under eq.~\eqref{1dBasis2Forms}, we attributed names to subfamilies whose meaning will be justified soon. 

Importantly, in addition to its generative aspect, we further prove in appendix~\ref{subsec:free_family} that the family~\eqref{1dBasis2Forms} is also \textit{free} in the sense that
\begin{equation}
\bgamma=0 \Leftrightarrow \left[(\alpha_{ij}^{(\ell)})_{1 \leq i < j \leq d_2, \ell\in \mbN}=0 \ \textrm{and} \  (\beta_{ij}^{(\ell)})_{1 \leq i \leq j \leq d_2, \ell\in \mbN}=0 \right]
\label{CNS_vanishing_loc_2form}
\end{equation}
 This family can thus be seen \textit{a basis} of the space of two-forms $\Omega^2_{\rm loc}$.
 In other words, in $d_1=1$ spatial dimension, a local two-form $\gamma\in\Omega^2_{\rm loc}$ can be uniquely decomposed over the family~\eqref{1dBasis2Forms}.


Note that eqs.~\eqref{eq:expressionAlpha} and \eqref{eq:expressionBeta} obey the symmetries $\alpha_{ij}^{(\ell)}=-\alpha_{ji}^{(\ell)}$ and $\beta_{ij}^{(\ell)}=\beta_{ji}^{(\ell)}$, respectively. This suggests reformulating the rather lengthy expression~\eqref{eq:generalLoc2Form01} of $\bgamma$ in a more compact form, albeit at the price of making the decomposition over the basis~\eqref{1dBasis2Forms} less visible:
%
%
\begin{eqnarray}
\bgamma = \sum_{\ell=0}^{\lfloor q/2 \rfloor} \int \rmd x \: \balpha^{(\ell)}(x,[\bphi]) : \bdelta^{x}\wedge \partial^{2\ell}\bdelta^{x}  \notag  +  \sum_{\ell=0}^{\lfloor (q-1)/2\rfloor} \int \rmd x \: \bbeta^{(\ell)}(x,[\bphi]): \bdelta^{x}\wedge\partial^{2\ell+1}\bdelta^{x} \ ,  \label{eq:generalLoc2Form02} 
\end{eqnarray}
where `$:$' represents the full tensor contraction between the antisymmetric matrix $\balpha^{(\ell)}\equiv(\alpha_{ij}^{(\ell)})_{i,j=1\dots d_2}$ (respectively the symmetric matrix $\bbeta^{(\ell)}\equiv(\beta_{ij}^{(\ell)})_{i,j=1\dots d_2}$) and the $\mbR^{d_2}\otimes\mbR^{d_2}$-valued 2-forms $[\bdelta^x\wedge \partial^k \bdelta^x]^{ij}\equiv \delta^{ix}\wedge \partial^k\delta^{jx}$. 
The symmetry properties of the matrices $\balpha$ and $\bbeta$ justifies the names given in~\eqref{1dBasis2Forms} to each of the two subfamilies: the space spanned by the antisymmetric (respectively symmetric) family of elementary 2-forms is obtained by contracting $\bdelta^{x}\wedge \partial^{2\ell}\bdelta^{x}$ (respectively $\bdelta^{x}\wedge \partial^{2\ell+1}\bdelta^{x}$) with antisymmetric matrices (respectively symmetric matrices).

\paragraph{Coordinates of an exact local two-form.}

Let us now consider a local 1-form $\bzeta_{\bphi} = \zeta_{ix}[\bphi]\delta^{ix}$ where $\zeta_{ix}[\bphi] = \zeta_i\left(\bphi(x),\partial_x\bphi(x),\dots,\partial^q_x \bphi(x)\right)$ for a given integer $q$.
%
Using formula~\eqref{extDer_convenient}, the functional exterior derivative of $\bzeta$ reads:
\begin{eqnarray*}
\mbd \bzeta &=& \frac{\delta \zeta_{i x}}{\delta \phi^{j y}}\delta^{j y}\wedge \delta^{i x} \\
&=& \sum_{k=0}^q\frac{\partial \zeta_{ix}}{\partial(\partial^k\phi^j)}(-1)^k \partial^k_y\delta(x-y) \delta^{j y}\wedge \delta^{i x} \\
&=& \sum_{k=0}^q (-1)^{k+1} \frac{\partial \zeta_{ix}}{\partial(\partial^k\phi^j)} \delta^{ix}\wedge \partial^k\delta^{jx} \ .
\end{eqnarray*}
where the last equality is straightforwardly obtained by evaluating each side on a pair of perturbations $\delta\bphi_1,\delta\bphi_2$.
As announced in~\eqref{inclusion_exact_two-forms}, $\mbd\bzeta$ is a local two-form, hence we can decompose it on the basis~\eqref{1dBasis2Forms}.
When compared to the generic local two-form $\bgamma$, using eq.~\eqref{generic_local_two-form_1d}, we have the relation 
\begin{equation}
\gamma^{(k)}_{ij}\leftrightarrow (-1)^{k+1} \frac{\partial \zeta_{ix}}{\partial(\partial^k\phi^j)}  \ .
\end{equation}
Using this relation, together with eqs.~\eqref{eq:expressionAlpha_two-form} and~\eqref{eq:expressionBeta_two-form}, we deduce the ``coordinate functions'' $\alpha_{ij}^{(\ell)}(x,[\bphi])$ and $\beta_{ij}^{(\ell)}(x,[\bphi])$  of $\mbd\bzeta$, which respectively read
\begin{equation}
 \alpha_{ij}^{(\ell)}=\frac{1}{2}\left[\frac{\partial \zeta_j}{\partial(\partial^{2\ell}\phi^i)}-\frac{\partial \zeta_i}{\partial(\partial^{2\ell}\phi^j)}-\sum_{k=\ell}^{\left \lfloor (q-1)/2 \right \rfloor}c^{2k+1}_{2\ell}\partial_x^{2k+1-2\ell} \Big(\frac{\partial \zeta_j}{\partial(\partial^{2k+1}\phi^i)}-\frac{\partial \zeta_i}{\partial(\partial^{2k+1}\phi^j)}\Big) \right]
\label{eq:expressionAlpha}
\end{equation}
and
\begin{equation}
\beta_{ij}^{(\ell)}=\frac{1}{2}\left[\frac{\partial \zeta_j}{\partial(\partial^{2\ell+1}\phi^i)} +\frac{\partial \zeta_i}{\partial(\partial^{2\ell+1}\phi^j)}-\sum_{k=\ell+1}^{\left \lfloor q/2 \right \rfloor}b^{2k}_{2\ell+1}\partial_x^{2k-2\ell-1}\Big(\frac{\partial \zeta_j}{\partial(\partial^{2k}\phi^i)} + \frac{\partial \zeta_i}{\partial(\partial^{2k}\phi^j)} \Big)\right] \ .
 \label{eq:expressionBeta}
\end{equation}
When $\bzeta$ is the one-form $\bD^{-1}\ba$ associated to dynamics~\eqref{EDPS02}, this allows reformulating the reversibility~\eqref{TRScond03} in an explicit, readily applicable, form:
\begin{equation}
 \bomega\equiv \mbd \bD^{-1}\ba =0 \Leftrightarrow  \left[(\alpha_{ij}^{(\ell)})_{1 \leq i < j \leq d_2, \ell\in \mbN}=0 \ \textrm{and} \  (\beta_{ij}^{(\ell)})_{1 \leq i \leq j \leq d_2, \ell\in \mbN}=0 \right] \ ,
\end{equation}
where the $\balpha^{(\ell)}$ and $\bbeta^{(\ell)}$ are respectively given by~\eqref{eq:expressionAlpha} and~\eqref{eq:expressionBeta} with $\bzeta\equiv \bD^{-1}\ba$.

To illustrate this result, let us show that it directly gives the convenient expression~\eqref{cycleAffAMB} of the vorticity of the one-dimensional AMB and thus, in particular, its reversibility condition as well. Since the density field $\rho$ is scalar valued, the basis~\eqref{1dBasis2Forms} reduces to the subfamily $(\delta^x\wedge\partial^{2\ell+1}\delta^x)_{\ell\in\mbN,x\in\mbR}$. Further, $\zeta=D^{-1}a=-\mu$ is a local functional of $\rho$ of order $q=2$. Hence the only possible term in the vorticity two-form is $\int \beta_{11}^{(0)} \delta^x\wedge \partial_x\delta^x$. A direct application of formula~\eqref{eq:expressionBeta} then gives
\begin{equation}
\beta_{11}^{(0)} = \frac{\partial \zeta}{\partial(\partial_x\rho)}- b_1^2 \partial_x \frac{\partial \zeta}{\partial (\partial_x^2\rho)} = - \left[ 2\lambda \partial_x\rho - \partial_x(-\kappa) \right] = -(2\lambda+\kappa')\partial_x\rho \ ,
\end{equation} 
where we used formula~\eqref{expr_coeff_b} \&~\eqref{expr_coeff_a} to get that $b^2_1=1$. Thus, as expected, we recover the (one-dimensional version of the) convenient expression of the AMB vorticity~\eqref{cycleAffAMB}, from which we can directly conclude that reversibility is equivalent to $2\lambda+\kappa'=0$.

\paragraph{Further decomposition of the space of local two-forms.}
Our partitioning of the basis~\eqref{1dBasis2Forms} between the antisymmetric and symmetric subfamilies already provides a decomposition of $\Omega^2_{\rm loc}(\mbF)$. But, as we will see in section~\ref{subsubsec:oddSubspace}, it is physically meaningful to further decompose the symmetric subfamily as follows.
In general, the field $\bphi$ is not a ``true'' vector field\footnote{If $\bphi$ is a true vector field, then $\bphi(\br)$ is a vector, a geometrical object on its own that can be equally described in any basis. On the other hand, if for instance the field consists in a density field $\rho$ and a polarity field $\bp$ stacked together, $\bphi=(\rho,\bp)$, then all the possible bases of $\mbR^4$ are not equivalent to describe $\bphi(r)$: those mixing the first coordinate with the others are not physically meaningful.} but is rather made out of $n$ distinct physical fields (of mass density and polarity for instance) stacked together, \ie $\bphi= (\bphi_{(1)},\dots, \bphi_{(n)})^\top$.
The set of coordinate indices $I=\{1,\dots, d_2\}$ of $\bphi=(\phi^i)_{i\in I}$ can then be partitioned as $I=\cup_{k=1}^n I_k$ , where $I_k$ is the set indexing the coordinates of $\bphi_{(k)}$. In turn, this allows partitioning of the symmetric family into two subfamilies.
In the first one, that we call the \textit{self-symmetric} family, the coordinate indices $i$ and $j$ appearing in the expression of a given elementary two-form belong to a common subset $I_k$:
\begin{equation}
([\delta^{ix}\wedge\partial^{2\ell+1}\delta^{jx} + \delta^{jx}\wedge\partial^{2\ell+1}\delta^{ix}], \delta^{ix}\wedge\partial^{2\ell+1}\delta^{ix})_{I(i)= I(j), 1 \leq i < j \leq d_2, \ell\in \mbN,x\in\mbR} \ .
\end{equation}
where $I(i)$ denotes the subset $I_k$ of $I$ to which a given index $i$ belongs.
In the second one, that we call the \textit{inter-symmetric} family, $i$ and $j$ belong to two different subset of $I$:
\begin{equation}
(\delta^{ix}\wedge\partial^{2\ell+1}\delta^{jx} + \delta^{jx}\wedge\partial^{2\ell+1}\delta^{ix})_{I(i)\neq I(j),1 \leq i < j \leq d_2, \ell\in \mbN,x\in\mbR} \ ,
\end{equation}
The components of $\bgamma$ along the self-symmetric subfamily then correspond to the diagonal blocks\footnote{The $n$ diagonal blocks of $\bbeta^{(\ell)}$ are $(\beta^{(\ell)}_{ij})_{i,j\in I_k}$, $k=1,\dots, n$.} of each $\bbeta^{(\ell)}$ in~\eqref{eq:generalLoc2Form02}, while those along the inter-symmetric family are given by all the other components of each $\bbeta^{(\ell)}$. 
Physically, we will see in section~\ref{subsubsec:oddSubspace} that the elements of the inter-symmetric family in the vorticity $\bomega$ of dyanmics~\eqref{EDPS02} are due to TRS-breaking \textit{interactions} between distinct fields $\bphi_{(i)}\neq\bphi_{(j)}$, while any element of the self-symmetric family originates from a field $\bphi_{(i)}$ whose evolution is irreversible \textit{on its own}.

At the level of vector spaces, our partitioning of the basis~\eqref{1dBasis2Forms} can thus be formulated as the decomposition
\begin{equation}
\Omega^2_{\rm loc}(\mbF) = \mcA \oplus \mcS_{\rm self} \oplus \mcS_{\rm inter} \ ,
\label{Decomposition_space_local_vorticities}
\end{equation}
where $\mcA$, $\mcS_{\rm self}$, and $\mcS_{\rm inter}$ are respectively the antisymmetric, self-symmetric and inter-symmetric subspaces, each one being the linear span\footnote{The linear span of a given family is to be understood here as 
integrating and summing (with respect to discrete and continuous indices, respectively) the elements of this family against function(als) $f_{ij}(x,[\bphi])$.} of the subfamily of~\eqref{1dBasis2Forms} that bears the corresponding name.

\subsection{Partial extension to higher spatial dimension $d_1$}
\label{subsec:higher_dimension_d1}

We now discuss how the results of section~\ref{subsec:basisOf2forms} generalize, at least partially, to higher spatial dimension $d_1>1$.
Let us consider the family of two-forms obtained by grouping together $d_1$ copies\footnote{We exclude from the additional copies the elements that are of order $n=0$ in derivation.} of the family~\eqref{1dBasis2Forms} where each derivative\footnote{Recall that, to lighten notations in the previous section, we denoted the distribution $\partial_x^k\delta^{ix}$ by $\partial^k\delta^{ix}$.} $\partial_x$ is replaced by $\partial_{r^k}$ in the $k^{\rm th}$ copy, \ie
\begin{equation}
([\delta^{i\br}\wedge \partial^{2n}_{r^k}\delta^{j\br} - \delta^{j\br}\wedge \partial^{2n}_{r^k}\delta^{i\br}],[\delta^{i\br}\wedge\partial^{2\ell+1}_{r^k}\delta^{j\br} + \delta^{j\br}\wedge\partial^{2\ell+1}_{r^k}\delta^{i\br}], \delta^{i\br}\wedge\partial^{2\ell+1}_{r^k}\delta^{i\br})_{1\leq k\leq d_1, 1 \leq i < j \leq d_2, n\in \mbN^*, \ell\in \mbN, x\in\mbR^{d_1}} \ .
\label{1dBasis2Forms_higherd1}
\end{equation}
This family turns out to be free in the same sense as its one dimensional counterpart~\eqref{1dBasis2Forms} (see the end of appendix~\ref{subsec:free_family}).
Further, it can be shown\footnote{Indeed, at each step of section~\ref{subsec:basisOf2forms}, this amounts to replace all the spatial derivatives by $\frac{\partial}{\partial r^n}$, and to sum from $n=1$ to $d_1$.} that it is generative for exterior derivatives of
one-forms $\bzeta(\br,[\bphi])$ that locally depend on $\bphi$ but not through any cross derivative, \ie which are of the form
\begin{equation}
\bzeta(\br,[\bphi]) = \bzeta\left(\bphi(\br),\{\partial_{r^k}\bphi(\br),...,\partial_{r^k}^q\bphi(\br)\}_{k=1,\dots,d_1}\right) \ .
\end{equation}
For this particular subspace of exterior derivatives of one-forms, the family~\eqref{1dBasis2Forms_higherd1} is thus a basis.
However, the family~\eqref{1dBasis2Forms_higherd1} can generate neither the exterior derivative of local one-forms whose dependence in $\bphi$ involve cross derivatives, like \eg $\frac{\partial^2\bphi(\br)}{\partial r^1\partial r^2}$, nor consequently the whole space of local two-forms. The family~\eqref{1dBasis2Forms_higherd1} hence generalizes~\eqref{1dBasis2Forms} and its properties only partially to dimensions $d_1>1$. 


We can extend the partition of~\eqref{1dBasis2Forms} into antisymmetric, self-, and inter- symmetric subfamilies to~\eqref{1dBasis2Forms_higherd1}, by the same procedure\footnote{For instance, the antisymmetric subfamily in dimension $d_1>1$ is obtained by gathering the $d_1$ ``copies'' of the one dimensional antisymmetric family, one for each direction of space $r^k$.} through which we defined~\eqref{1dBasis2Forms_higherd1} from~\eqref{1dBasis2Forms}, and still denote the corresponding subspace by $\mcA$, $\mcS_{\rm self}$, and $\mcS_{\rm inter}$, respectively.
In $d_1>1$ spatial dimension, the decomposition~\eqref{Decomposition_space_local_vorticities} thus generalizes to
\begin{equation}
\Omega^2_{\rm loc} (\mbF) = \mcA \oplus \mcS_{\rm self} \oplus \mcS_{\rm inter} \oplus \mcC \ .
\label{Decomposition_space_local_vorticities_higher_d1}
\end{equation}
Here $\mcC$ is a complement of $\mcA\oplus\mcS_{\rm self}\oplus\mcS_{\rm inter}$ in $\Omega^2_{\rm loc} (\mbF)$ and hence allows to decompose, in particular, the exterior derivative of any local one-form $\bzeta$ that depends on cross-derivatives of $\bphi$. Finding a basis of a suitable complement $\mcC$ is an interesting challenge that is left for future work.

Despite its incompleteness in generic situations, the family~\eqref{1dBasis2Forms_higherd1} is sufficient to decompose the vorticity two-form of many celebrated physical models. 
For instance, this is the case for Active Model B in arbitrary spatial dimension $d_1$ since its one-form $\bD^{-1}\ba=-\mu$, where $\mu(\br,[\bphi])$ is given by eq.~\eqref{AMB_chemical_potential}, depends locally on $\rho$ but not through any cross derivative. The resulting vorticity $\bomega=-\mbd \mu$, given by eq.~\eqref{cycleAffAMB}, is generated by the basis elements $(\delta^{\br}\wedge \partial_{r^k}\delta^{\br})_{k=1,\dots,d_1}$ of the self-symmetric family. Another example is given by the KPZ equation, whose vorticity is the same as that of AMB~\cite{o2023nonequilibrium}.

After this short detour through higher spatial dimensions, we will stick to the $d_1=1$ case in the rest of section~\ref{sec:basis} for simplicity, keeping in mind that most of the following results can be generalized from the family~\eqref{1dBasis2Forms} to~\eqref{1dBasis2Forms_higherd1}.

\subsection{Relation between functional two-forms and skew-symmetric differential operators}
\label{subsec:skewsym_operators}

We have seen in section~\ref{subsec:RevCondFuncExtDer} that the question of the reversibility of dynamics~\eqref{EDPS02} is geometrically a question of exactness (or integrability) of the functional one-form $\bD^{-1}\ba$. Consequently, the language of functional differential forms appears to be, mathematically, the most natural one to tackle it. Nevertheless, functional differential forms are not commonly used in physics and even less the corresponding formalism introduced in this article.
Hence, before turning to the study of the phenomenology associated to each element of the basis~\eqref{1dBasis2Forms}, 
we give here an alternative description in the more common language of differential operators.
%

To this end, let us start by extending to the whole space $\Omega^2(\mbF)$ the map ``\textit{hat}'', defined in section~\ref{subsec:vorticityOperator}, that associates to a vorticity two-form $\bomega\in\mbd\Omega^1(\mbF)$ the corresponding cycle affinity operator $\whomega$. For any two-form $\bgamma\in \Omega^2(\mbF)$, the corresponding ``cycle-affinity operator'' $\whgamma\equiv hat(\gamma)$
is obtained by replacing $\bomega$ and $\whomega$ in definition~\eqref{def_affinity_operator} respectively by $\bgamma$ and $\whgamma$.
We now restrict our attention to two-forms $\bgamma$ that are local, $\bgamma\in\Omega^2_{\rm loc}(\mbF)$. Then, injecting expression~\eqref{eq:generalLoc2Form02} into definition~\eqref{def_affinity_operator} and performing several integrations by parts, we get:
\begin{equation}
[\whgamma(\delta\bphi)]_i = \sum_{\ell=0}^{\lfloor q/2 \rfloor}\sum_{j=1}^{d_2} \left[\alpha_{ij}^{(\ell)} \partial_x^{2\ell} \delta\phi^j + \partial_x^{2\ell}(\alpha_{ij}^{(\ell)}\delta\phi^j) \right]
  -  \sum_{\ell=0}^{\lfloor (q-1)/2\rfloor}\sum_{j=1}^{d_2} \left[ \beta_{ij}^{(\ell)} \partial_x^{2\ell+1} \delta\phi^j + \partial_x^{2\ell+1}(\beta_{ij}^{(\ell)}\delta\phi^j) \right] \ ,
\end{equation}
or, in terms of operators:
\begin{equation}
\whgamma = \sum_{\ell=0}^{\lfloor q/2 \rfloor} \left[\balpha^{(\ell)} \partial_x^{2\ell}  + \partial_x^{2\ell}\balpha^{(\ell)} \right]
  -  \sum_{\ell=0}^{\lfloor (q-1)/2\rfloor} \left[ \bbeta^{(\ell)} \partial_x^{2\ell+1} + \partial_x^{2\ell+1}\bbeta^{(\ell)}\right]
  \label{genericCycleOperator}
\end{equation}
Each differential operator in the first (respectively second) sum on the right-hand of equation~\eqref{genericCycleOperator} is the cycle affinity operator associated to a two-form that belongs to the antisymmetric (respectively symmetric) subspace. Moreover, we see that each one is, up to a factor $1/2$, the skew-symmetric part of an operator that applies an even- (respectively odd-) degree derivative and multiplies by a skew-symmetric (respectively symmetric) matrix.

A natural question one could ask is whether the latter combinations are the only possible ones from which to build skew-symmetric operators. It turns out they are not, but that the missing operators can be written as linear combination of those appearing in~\eqref{genericCycleOperator}. To see this, let us consider a generic differential operator 
$\mcL=\sum_k\bL_k(x,[\bphi]) \partial_x^{k}$ that, for a given $\bphi\in\mbF$, acts on perturbations $\delta\bphi\in T_{\bphi}\mbF$, the $\bL_k$'s being square matrices that depend on $x$ and $\bphi$.
The skew-symmetric part $\mcL^A$ of $\mcL$ reads
\begin{eqnarray}
\mcL^A =\frac{1}{2} \sum_k\left[\bL_k \partial_x^{k} - (-1)^k \partial_x^{k} \bL_k^\top \right] \ ,
\label{generic_antisym_op}
\end{eqnarray}
where $\bL_k^\top$ is the transpose of $\bL_k$.
Splitting each matrix $\bL_k$ into its symmetric and skew-symmetric parts, $\bL_k=\bL^S_k+\bL^A_k$, and rearranging the terms, we get
\begin{eqnarray}
 \mcL^A  = \frac{1}{2}\sum_k \left\{ \left[ \bL^A_k \partial_x^{k} + (-1)^k \partial_x^{k} \bL^A_k \right] + \left[\bL^S_k\partial_x^{k}-(-1)^k \partial_x^{k}\bL^S_k \right]  \right\} \ .
\end{eqnarray}
Distinguishing derivatives of even and odd degrees, we finally obtain
\begin{eqnarray}
 \mcL^A  &=& \frac{1}{2}\sum_\ell \Big\{\left[\bL^A_k\partial_x^{2\ell}+ \partial_x^{2\ell}\bL^A_k \right] + \left[ \bL^S_k \partial_x^{2\ell} -  \partial_x^{2\ell} \bL^S_k \right] \notag \\
& & + \left[\bL^A_k\partial_x^{2\ell+1}- \partial_x^{2\ell+1}\bL^A_k \right] + \left[ \bL^S_k \partial_x^{2\ell+1} + \partial_x^{2\ell+1} \bL^S_k \right] \Big\}
\label{skewOp_decomp}\ .
\end{eqnarray}
We thus see that any skew-symmetric differential operator can be written as the superposition of 4 types of elementary skew-symmetric operators:
\begin{gather}
\bA\partial_x^{2\ell}+ \partial_x^{2\ell}\bA\ ,\label{op_type_Ap} \\
\bS \partial_x^{2\ell} -  \partial_x^{2\ell} \bS \ , \label{op_type_Sm}\\
\bA\partial_x^{2\ell+1}- \partial_x^{2\ell+1}\bA \ ,  \label{op_type_Am}\\
\bS \partial_x^{2\ell+1} + \partial_x^{2\ell+1} \bS \ , \label{op_type_Sp} 
\end{gather}
where $\bA$ and $\bS$ are matrix-valued function(al)s of $x$ and $\bphi$ that are antisymmetric and symmetric, respectively.
Let us denote by $\pmcA,\mmcS, \mmcA$, and $\pmcS$ the vector spaces respectively generated by linear combinations of antisymmetric operators of the types~\eqref{op_type_Ap}-\eqref{op_type_Sp}. For instance, the elements of $\pmcA$ are of the form $\sum_\ell \bA_\ell\partial_x^{2\ell}+ \partial_x^{2\ell}\bA_\ell$, where the sum is finite and all the $\bA_\ell$ are antisymmetric matrices that depend on $x$ and $\bphi$.

Note that the operators that make up the generic cycle affinity operator~\eqref{genericCycleOperator} belong either to $\pmcA$ or $\pmcS$.
As the matrices $\balpha^{(\ell)},\bbeta^{(\ell)}$ in the expression~\eqref{eq:generalLoc2Form02} of $\bgamma$ can be chosen arbitrarily in the spaces of antisymmetric and symmetric matrices, respectively, we conclude that, as suggested by the notations, $\pmcA$ and $\pmcS$ are the images under \textit{hat} of $\mcA$ and $\mcS\equiv \mcS_{\rm self}\oplus\mcS_{\rm inter}$, respectively.

We show in appendix~\ref{appendix:another_decomp_of_operators} the following decomposition formulas for operators of the form~\eqref{op_type_Sm} and~\eqref{op_type_Am}:
\begin{equation}
\bS \partial_x^{2\ell}-\partial_x^{2\ell} \bS = -\sum_{k=0}^{\ell-1}b^{2\ell}_{2k+1} \left[ \left(\partial_x^{2\ell-2k-1}\bS\right)\partial_x^{2k+1} + \partial_x^{2k+1}\left(\partial_x^{2\ell-2k-1}\bS\right)\right] \ ,
\label{type_2_decompos}
\end{equation}
and
\begin{equation}
\bA \partial_x^{2\ell+1}-\partial_x^{2\ell+1} \bA = -\sum_{k=0}^{\ell}c^{2\ell+1}_{2k} \left[ \left(\partial_x^{2\ell+1-2k}\bA\right)\partial_x^{2k} + \partial_x^{2k}\left(\partial_x^{2\ell+1-2k}\bA\right)\right] \ ,
\label{type_3_decompos}
\end{equation}
with the same coefficients $b_k^\ell$ and $c_k^\ell$ as in eqs.~\eqref{eq:expressionAlpha_two-form}-\eqref{eq:expressionBeta_two-form}.
We thus see that the operators of the form~\eqref{op_type_Sm}(respectively~\eqref{op_type_Am}) can be written as a superposition of operators of the form~\eqref{op_type_Sp} (respectively~\eqref{op_type_Ap}).
In other words, we have the inclusions $\mmcS\subset \pmcS$ and $\mmcA\subset\pmcA$ and it follows that the vector space of antisymmetric differential operators~\eqref{generic_antisym_op}, that we denote by $\mbL^A$, can be decomposed as 
$\mbL^A = \pmcA + \pmcS$.
In particular, this implies that the map \textit{hat} is surjective from $\Omega^2_{\rm loc}$ to $\mbL^A$, \ie that any element in $\mbL^A$ is of the form $\whgamma$, with $\bgamma\in\Omega^2_{\rm loc}$.
But an element $\whgamma\in\mbL^A$ vanishes iff, for all $\delta\bphi_1,\delta\bphi_2$, $\int\delta\bphi_1\cdot \whgamma(\delta\bphi_2) \rmd x = 0$, which also reads $\bgamma(\delta\bphi_1,\delta\bphi_2)=0$, by definition of \textit{hat}. Using~\eqref{CNS_vanishing_loc_2form} we conclude that
\begin{equation}
\whgamma=0 \Leftrightarrow \balpha^{(\ell)}=\bbeta^{(\ell)}=0 \ .
\label{CNS_cancellation_antisym_op}
\end{equation}
and  $\ker(hat)=\{0\}$. Consequently we have the direct-sum decomposition of $\mbL^A$:
\begin{equation}
\mbL^A =\pmcA\oplus \pmcS \ ,
\label{direct_sum_antisym_op}
\end{equation}
and \textit{hat} is a linear isomorphism between $\Omega^2_{\rm loc}$ and $\mbL^A$.

%
%

To make the correspondence between decomposition~\eqref{direct_sum_antisym_op} and its counterpart for two-forms even clearer,
let us consider a generic $\bgamma\in\Omega^2_{\rm loc}(\mbF)$, but start from its expression~\eqref{generic_local_two-form_1d} rather than its decomposition~\eqref{eq:generalLoc2Form02} in the basis~\eqref{1dBasis2Forms}.
On the one hand, we can decompose each matrix $\bgamma^{(k)}={}^S\bgamma^{(k)}+{}^A\bgamma^{(k)}$ into its symmetric and antisymmetric parts, as in section~\ref{subsec:basisOf2forms}. Then, we notice that formulas~\eqref{core:decompo_elementary_2forms_plus} \&~\eqref{core:decompo_elementary_2forms_minus} mean that the components of $\bgamma$ of the form ${}^A\bgamma^{(2\ell+1)}_x:\bdelta^x\wedge\partial^{2\ell+1}\bdelta^x$ and ${}^S\bgamma^{(2\ell)}_x:\bdelta^x\wedge\partial^{2\ell}\bdelta^x$ can be respectively decomposed along the elementary two-forms $(\bdelta^x\wedge\partial^{2k}\bdelta^x)_{k\leq \ell}$ and $(\bdelta^x\wedge\partial^{2k+1}\bdelta^x)_{k\leq\ell}$, leading to the expression~\eqref{eq:generalLoc2Form02}, with coordinates~\eqref{eq:expressionAlpha_two-form}-\eqref{eq:expressionBeta_two-form}.
On the other-hand, applying the \textit{hat} map to the expression~\eqref{generic_local_two-form_1d} of $\bgamma$ gives an alternative formula for $\whgamma$:
\begin{equation}
\whgamma = \sum_k\left\{(-1)^k \bgamma^{(k)} \partial^k - \partial^k[\bgamma^{(k)}]^\top \right\} \ ,
\label{a_last_generic_skew_op}
\end{equation}
which coincides with the generic expression of a skew-symmetric differential operator~\eqref{generic_antisym_op}, with $\bL_k=2(-1)^k \bgamma^{(k)}$. Repeating for $\whgamma$ the reasoning we had above for the operator $\mcL^A$, we end up decomposing the operators $[{}^S\bgamma^{(2\ell)}]\partial^{2\ell} - \partial^{2\ell}[{}^S\bgamma^{(2\ell)}]\in \mmcS$ and $[{}^A\bgamma^{(2\ell+1)}]\partial^{2\ell+1} - \partial^{2\ell+1}[{}^A\bgamma^{(2\ell+1)}]\in \mmcA$ respectively using formula~\eqref{type_2_decompos} and~\eqref{type_3_decompos}, which gives back the expression~\eqref{genericCycleOperator} of $\whgamma$. 
Hence, the decomposition formulas~\eqref{type_2_decompos} \&~\eqref{type_3_decompos} in $\mbL^A$ (and the corresponding direct sum decomposition~\eqref{direct_sum_antisym_op}) mirror formulas~\eqref{core:decompo_elementary_2forms_plus} \&~\eqref{core:decompo_elementary_2forms_minus} in $\Omega^2_{\rm loc}(\mbF)$ (and its corresponding decomposition), this exact correspondence being given by the isomorphism \textit{hat}.

As an example of this correspondence, let us consider the cycle affinity operator $\whomega$ of AMB in $d_1=1$ dimension. Using definition~\eqref{def_affinity_operator} together with expression~\eqref{problematic_2form} of $\bomega$ yields, after some integrations by parts,
\begin{equation}
\whomega_{\rho} \delta\rho = \left\{\partial_x \left[ (2\lambda \partial_x\rho)\delta\rho\right] + 2\lambda (\partial_x\rho) \delta\rho \right\} + \left\{-\kappa\partial_x^2\delta\rho + \partial_x^2(\kappa \delta\rho)\right\} \ .
\label{whomega_AMB_temp}
\end{equation}
Let us denote by $\whomega^{(1)}$ and $\whomega^{(2)}$ the two operators within each pair of curly brackets on the right-hand side of eq.~\eqref{whomega_AMB_temp}. We see that $\whomega^{(1)}\in \pmcS$ while $\whomega^{(2)}\in\mmcS$. Consequently, these operators are not independent\footnote{By ``independent'', we here mean that $\whomega=0$ is not equivalent to $\whomega^{(1)}=\whomega^{(2)}=0$, since $\pmcS\cap\mmcS=\mmcS\neq 0$.} and the latter can be re-written as 
\begin{equation}
\whomega^{(2)}_\rho \delta\rho  \equiv -\kappa\partial_x^2\delta\rho + \partial_x^2(\kappa \delta\rho) = \delta\rho \partial_x^2\kappa + 2(\partial_x \kappa ) (\partial_x \delta\rho) = \partial_x \left[(\partial_x\kappa)\delta\rho \right] + (\partial_x\kappa)(\partial_x\delta\rho) \ ,
\end{equation}
so that the whole cycle affinity operator also reads
\begin{equation}
\whomega_{\rho}\delta\rho = \partial_x \left[  (2\lambda+\kappa')(\partial_x\rho)\delta\rho \right] + (2\lambda+\kappa')(\partial_x\rho)\partial_x\delta\rho \ . 
\label{whomega_AMB_final}
\end{equation}
This mirrors, at the level of the skew-symmetric operator $\whomega$, what happens for the vorticity 2-form $\bomega$, where the second term on the right-hand side of the expression~\eqref{problematic_2form} can be decomposed along the same elementary two-form as the first one.
Comparing the form~\eqref{whomega_AMB_final} to the generic expression~\eqref{genericCycleOperator}, $\whomega$ has a single non-zero component: $\beta^{(0)}_{11}=-(2\lambda+\kappa')\partial_x\rho$, along the elementary operator of order $\ell=0$ of the self-symmetric family. The fact that the latter is ``free'' (in the sense of~\eqref{CNS_cancellation_antisym_op}) provides an alternative way to find the reversibility condition of AMB: $2\lambda+\kappa'=0$.

Hence, the results of this section allow the irreversiblity of dynamics~\eqref{EDPS02} -- in $d_1=1$ spatial dimension and with a one-form $\bD^{-1}\ba$ that is local -- to be seen from the point of view of the cycle affinity operator $\whomega$. We see that, from the four different types of differential operators~\eqref{op_type_Ap}-\eqref{op_type_Sp}, two are redundant, while the other two are ``free'' in an appropriate sense~\eqref{CNS_cancellation_antisym_op} which allows us, in particular, to 
directly deduce the explicit reversibility condition(s) of the dynamics. These properties of antisymmetric differential operators mirror (and hopefully illuminate) the corresponding ones found earlier for local two-forms.

\subsection{Phenomenology of each vorticity subspace}
\label{subsec:basisPhenomeno}

In this section, we examine what phenomenology one could expect from a dynamics~\eqref{EDPS02} whose vorticity two-form $\bomega$ belongs exclusively to one of the three subfamilies of~\eqref{1dBasis2Forms} identified in~\ref{subsec:basisOf2forms}.
Of course, determining the precise phenomenology of a field theory only based on the corresponding $\bomega$ is not completely possible, as the latter only contains linearized information about what makes the dynamics irreversible. 
However, we will see that tackling this problem through the vorticity dynamics~\eqref{Vortex_dynamics} -- even with rather generic forms of this operator -- allows us to qualitatively recover celebrated, typically out-of-equilibrium, behaviors, from the flocking state of aligning self-propelled agents to the run-and-chase dynamics of non-reciprocally interacting mixtures, through nonequilibrium pattern formation.

Once again, we here focus our attention only on $d_1=1$ dimensional dynamics obeying~\eqref{EDPS02} whose one-form $\bD^{-1}\ba$ is local in $\bphi$ so that the corresponding vorticity two-form $\bomega\equiv\mbd\bD^{-1}\ba$ can be expended on the basis~\eqref{1dBasis2Forms}.
Even though generically the vorticity has no reason to be generated by only one of the subfamilies identified in~\eqref{1dBasis2Forms}, 
we will assume below that this is the case and examine each subspace successively.
As our analysis is essentially qualitative, we expect the behavior of a superposition of vorticities from distinct subfamilies of~\eqref{1dBasis2Forms} to be in large part the superposition of the corresponding phenomenologies.
Note however that the methods used below to analyze each subfamily independently could easily be adapted to a superposed case.
 
We start by rewriting here for convenience the dynamics~\eqref{Vortex_dynamics} generated by the vorticity operator:
\begin{equation}
\partial_t\delta\bphi = \bW_{\bphi} \delta\bphi \ ,
\label{BOmega_dynamics_1d}
\end{equation} 
which encapsulates the time-irreversible part of dynamics~\eqref{EDPS02} in the vicinity of a given $\bphi$.
Let us recall that, as stated in section~\ref{subsec:vorticityOperator}, the vorticity and cycle affinity operators are related by $\bW=-\bD\whomega/2$. 
 
For simplicity, we assume from now on the diffusion operator $\bD$ to be proportional to the identity, unless otherwise stated explicitly. We can further assume this (positive) constant be equal to one, as any other one could be subsumed in the coordinates of the generic vorticities we will consider.
We will briefly comment below how our results generalize to other situations. In particular, having $\bD\propto\partial^2_x$ (which is often the case for conserved fields) doesn't change the qualitative results found below.

For concreteness, as well as setting $d_1 = 1$ we focus on the case $d_2 = 2$ in this section. (Recall that $d_2$ is the dimension of the order parameter space.)
Nonetheless we will argue below -- notably through the example of non-reciprocal flocking in section~\ref{subsubsec:evenSubspace} -- that, in higher $d_2$ dimensions, 
the typical phenomenology associated to each subspace of~\eqref{1dBasis2Forms} remains very similar \textit{in} $\mbR^{d_2>2}$ as compared to $\mbR^2$. However when $\bphi$ is not a ``true'' $\mbR^{d_2}$-valued vector field but rather a collection of various fields (of \eg number density and polarity; see the example of Active Ising Model at the end of this section) stacked together, the way this phenomenology is observed \textit{in the ``physical space''} $\mbR^{d_1}$ may appear quite different.

\subsubsection{The antisymmetric subspace and non-reciprocal interactions}
\label{subsubsec:evenSubspace}
Let us consider a cycle affinity operator that belongs to the antisymmetric subspace and is of homogeneous degree: $\whomega=\balpha \partial_x^{2\ell}+\partial_x^{2\ell}\balpha$, with $\balpha$ an antisymmetric matrix that depends on\footnote{Throughout section~\ref{subsec:basisPhenomeno}, we consider that $\bphi$ is fixed and we study the linear dynamics~\eqref{BOmega_dynamics_1d} around this point in function space. That is why we consider $\balpha$ as a function of $x$ only.} the space variable $x$. Since $\bD$ is assumed to be the identity, the vorticity operator reads 
\begin{equation}
\bW = -\frac{1}{2}\left[\balpha \partial_x^{2\ell}+\partial_x^{2\ell}\balpha\right] \ .
\label{Omega_shape_evenSpace}
\end{equation}

Let us first consider the case $\ell=0$ which is slightly degenerate. In this case, $\bW$ acts as a multiplication by the antisymmetric matrix $-\balpha$, which generically depends on $x$. The solution of~\eqref{BOmega_dynamics_1d} hence corresponds to a field $\delta\bphi$ that rotates on site at each $x\in\mbR$ -- \ie all $\delta\bphi(x)$ independently rotate in $\mbR^{2}$ -- in a direction (clockwise or counter-clockwise) and with a speed that are given by $-\alpha_{12}(x)=\alpha_{21}(x)$.

Let us now turn to the case $\ell\geq 1$. 
We start by considering the first of the two limit-regimes described in section~\eqref{subsec:vorticityOperator}, where the fluctuation $\delta\bphi$ initially varies on a spatial scale $\ell_{\delta\phi}$ that is much larger than that for spatial variations of the operator $\bW$ (which is itself that of the matrix $\balpha$ and in turn depends on the chosen base-state $\bphi$), $\ell_{\delta\phi}\gg \ell_W$. In such a regime, $\bW$ again acts as multiplication by an antisymmetric matrix, 
which is $-\partial_x^{2\ell}\balpha/2$ in this case. Consequently, the phenomenology is the same as that of the $\ell=0$ case, although the role previously played by $\balpha$ is now taken by $\partial^{2\ell}_x\balpha/2$. Note that if $\balpha$ is uniform, then $\ell_W=\infty$ and this first regime does not exist.

In the opposite limit-regime, where we consider fluctuations $\delta\bphi$ that vary on much smaller scales than $\bW$, $\ell_{\delta\phi}\ll \ell_W$, we partition the physical space $\mbR^{d_1}=\mbR$ into subdomains where $\balpha$ can be considered uniform and, consequently, where the vorticity operator approximately reads $\bW \simeq -\balpha\partial_x^{2\ell}$. In Fourier space, the latter thus acts as a matrix multiplication by:
\begin{equation}
\bW_k = W_k \begin{bmatrix} 0 & -1 \\ 1 & 0 \end{bmatrix} \ , 
\label{Omega_alpha_family}
\end{equation}
where $W_k\equiv (-1)^\ell \alpha_{12}k^{2\ell}$ and $k$ is the wave number.
If we focus on a generic harmonic mode that initially reads
\begin{equation}
\delta\bphi_k(x,t=0) = \begin{bmatrix} \rho_1 \cos(kx+\theta_1) \\  \rho_2 \cos(kx+\theta_2) \end{bmatrix} \ ,
\label{bpsik_initial}
\end{equation}
then the solution at time $t$ is
\begin{eqnarray}
\delta\bphi_k(x,t) &=& e^{t\bW_k} \begin{bmatrix} \rho_1 \cos(kx+\theta_1) \\  \rho_2 \cos(kx+\theta_2) \end{bmatrix} \notag\\
&=& \begin{bmatrix} \rho_1 \cos(kx+\theta_1) \\  \rho_2 \cos(kx+\theta_2) \end{bmatrix} \cos(W_k t) + \begin{bmatrix} -\rho_2 \cos(kx+\theta_2) \\  \rho_1 \cos(kx+\theta_1) \end{bmatrix} \sin(W_k t) \ .
\label{solutionAlphaFam01}
\end{eqnarray}
Compared to the previous limit-regime, $\bW$ can still be seen as a generator of rotations but the rotation speed now increases with the wavenumber $k$. Moreover, the vector $\delta\bphi_k(x,t)\in\mbR^2$ always rotates at the same speed and in the same direction at all $x$ in a given domain within which $\alpha_{12}$ is considered uniform.
Interestingly, the solution~\eqref{solutionAlphaFam01} can also be written as the superposition of propagating waves (see appendix~\ref{app:alphaFamily}):
\begin{equation}
\delta\bphi_k(x,t) = \rho_+  \begin{bmatrix} \cos(kx+W_k t+\theta_+) \\  \sin(kx+W_k t+\theta_+) \end{bmatrix} + \rho_-  \begin{bmatrix} \cos(kx-W_k t-\theta_-) \\  -\sin(kx-W_k t-\theta_-) \end{bmatrix} \ .
\label{solutionAlphaFam02}
\end{equation}
In eq.~\eqref{solutionAlphaFam02}, $\rho_\pm$ and $\theta_\pm$ are the polar coordinates of the vectors $(U\pm V)/2 \in \mbR^2$, \ie
\begin{equation}
\frac{U\pm V}{2} = \rho_\pm \begin{bmatrix} \cos \theta_\pm \\ \sin \theta_\pm \end{bmatrix} \ ,
\end{equation}
where
\begin{equation}
U \equiv \begin{bmatrix} \rho_1 \cos(\theta_1) \\  \rho_2 \cos(\theta_2) \end{bmatrix} \quad \text{and} \quad V \equiv \begin{bmatrix} -\rho_2 \sin(\theta_2) \\  \rho_1 \sin(\theta_1) \end{bmatrix} \ .
\end{equation}
We see in eq.~\eqref{solutionAlphaFam02} that the $`+'$ (respectively $`-'$) component of $\delta\bphi_k$ (which are respectively defined as the first and second term on the right-hand side of eq.~\eqref{solutionAlphaFam02}), which is positively (respectively negatively) circularly polarized, propagates on $\mbR$ at speed $W_k/k=(-1)^\ell\alpha_{12}k^{2\ell-1}$ in the direction given by $\textrm{sign}(-W_k)$ (respectively $\textrm{sign}(W_k)$).
To further interpret these facts, let us assume that, for instance, $W_k>0$. Then we see that in both the $`+'$ and $`-'$ component in~\eqref{solutionAlphaFam02}, $\delta\phi_k^2$ has a phase delay of $\pi/2$ with respect to $\delta\phi_k^1$ in the direction of propagation. In other words, if we interpret the components $\delta\phi_{1,2}$ as representing two different ``species'', then everything happens as if $\delta\phi_k^2$ \textit{was chasing} $\delta\phi_k^1$ (see figure~\ref{fig:basis_phenomeno}). 


Interestingly, if $\alpha_{12}$ changes its sign over $\mbR$, a given component chases the other on all sub-domains of a given sign of $\alpha_{12}$, while the contrary happens on the sub-domains of opposite sign. Besides, if $\alpha_{12}$ has a constant sign over all $\mbR$, but if the sign of $\partial^{2\ell}_x\alpha_{12}$ is different from that of $\alpha_{12}$, then the attribution of the pray and predator roles depends not on the spatial region but on the spatial scale.

\begin{figure}[h!]
\begin{center}
\begin{tikzpicture}
\def \xpic{3.5} 
\def \ypic{1.5}
\def \xdistar{2} 
\def \ydistar{0}
\def \l{1} 
\def \dx{0.25} 

\path(-\xpic,+\ypic)node[anchor=center]{\includegraphics[scale=0.25]{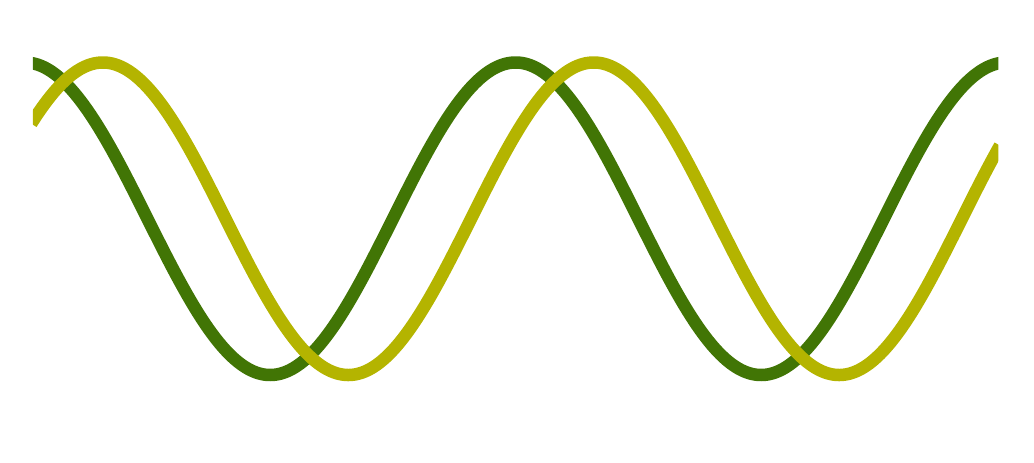}};
\draw[->,ultra thick,colorL] (-\xpic+\xdistar,+\ypic+\ydistar)--(-\xpic+\xdistar+\l,+\ypic+\ydistar);

\path(-\xpic,-\ypic)node[anchor=center]{\includegraphics[scale=0.25]{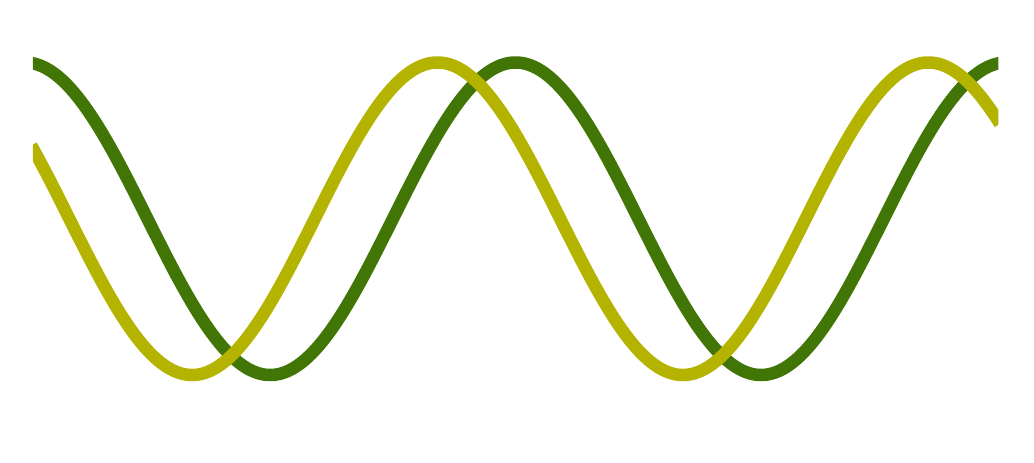}};
\draw[->,ultra thick,colorL] (-\xpic-\xdistar-\dx,-\ypic+\ydistar)--(-\xpic-\xdistar-\l-\dx,-\ypic+\ydistar);

\path(+\xpic,+\ypic)node[anchor=center]{\includegraphics[scale=0.25]{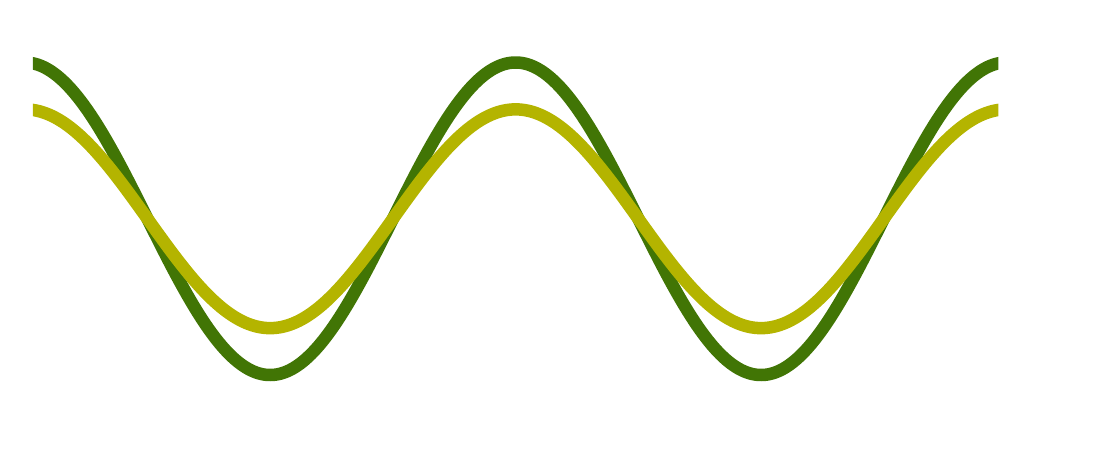}};
\draw[->,ultra thick,colorR] (+\xpic+\xdistar,+\ypic+\ydistar)--(+\xpic+\xdistar+\l,+\ypic+\ydistar);

\path(+\xpic,-\ypic)node[anchor=center]{\includegraphics[scale=0.25]{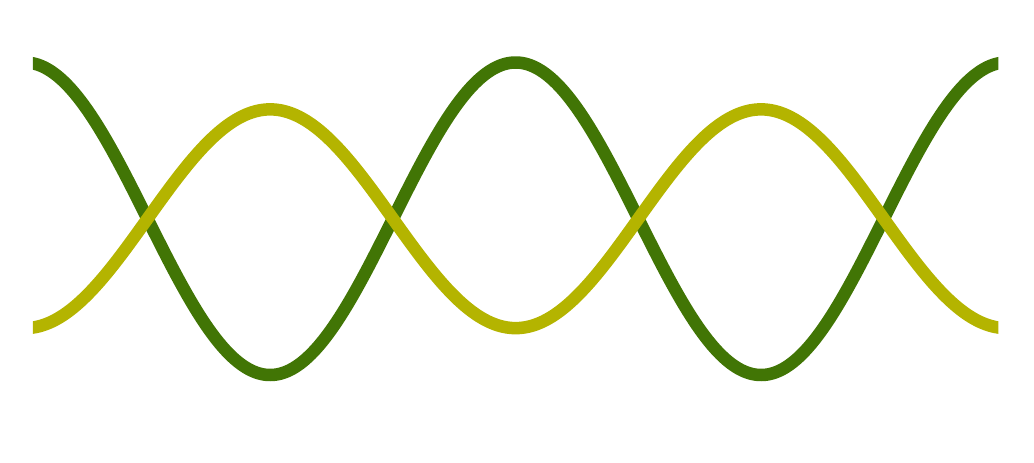}};
\draw[->,ultra thick,colorR] (+\xpic-\xdistar-\dx,-\ypic+\ydistar)--(+\xpic-\xdistar-\l-\dx,-\ypic+\ydistar);

\def \LX{2.5}
\def \CX{3.5}
\draw [ultra thick,dashed] (-\CX-\LX-\dx/2,0)--(-\CX+\LX-\dx/2,0);
\draw [ultra thick,dashed] (\CX-\LX-\dx/2,0)--(\CX+\LX-\dx/2,0);

\draw [ultra thick] (-\dx/2,2.5)--(-\dx/2,-2.5); 

\def \X{5.5}
\def \Y{2.5}
\draw [ultra thick] (-\dx/2,2.5)--(-\dx/2+\X,2.5+\Y);
\draw [ultra thick] (-\dx/2,2.5)--(-\dx/2-\X,2.5+\Y);

\def \xx{0.25}
\def \dY{-0.5}
\path(-\dx/2-0.,2.5+\Y+1+\dY)node[anchor=center]{\includegraphics[scale=0.25]{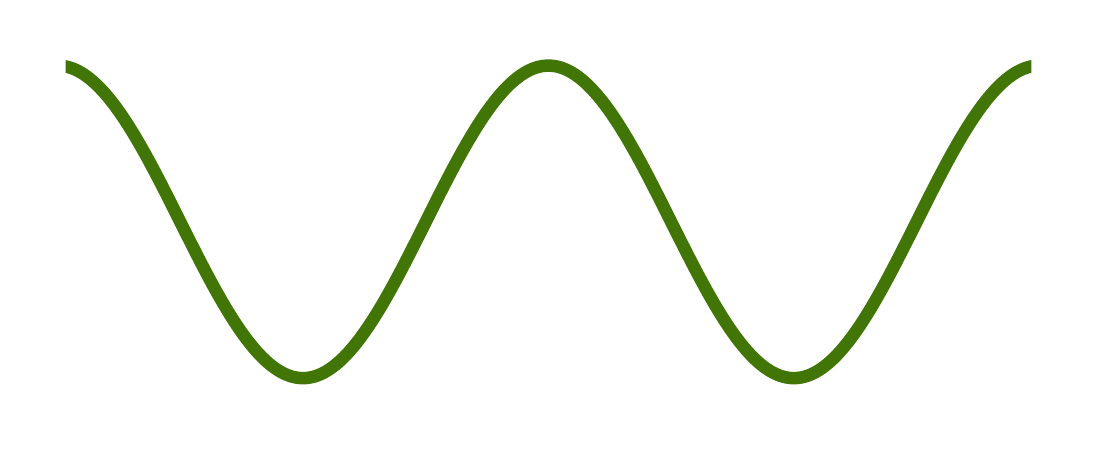}};
\draw[->,ultra thick,colorUp] (-\dx/2-0.15+\xdistar+\xx , 2.5+\Y+1+\ydistar+\dY)--(-\dx/2-0.15+\xdistar+\l+\xx , 2.5+\Y+1+\ydistar+\dY);
\draw[->,ultra thick,colorUp,densely dotted] (-\dx/2-0.15-\xdistar-\dx+\xx , 2.5+\Y+1+\ydistar+\dY)--(-\dx/2-0.15-\xdistar-\dx-\l+\xx , 2.5+\Y+1+\ydistar+\dY);

\def \vy{0.2}
\path(-\dx/2,2.5+\Y-0.8+\vy+\dY)node[anchor=center,draw=colorUp,rounded corners]{\LARGE \textcolor{colorUp}{$\mcS_{\rm self}$}};
\path(-\dx/2-0.15-5,2.5+\Y-0.8-\vy+\dY)node[anchor=center,draw=colorL,rounded corners]{\LARGE \textcolor{colorL}{$\mcA$}};
\path(-\dx/2+0.15+5,2.5+\Y-0.8-\vy+\dY)node[anchor=center,draw=colorR,rounded corners]{\LARGE \textcolor{colorR}{$\mcS_{\rm inter}$}};

\end{tikzpicture}
\caption{\label{fig:basis_phenomeno} 
A schematic picture of the (short-scale) phenomenology we expect from each subspace of the decomposition $\Omega^2_{\rm loc}(\mbF)=\mcA\oplus \mcS_{\rm self} \oplus \mcS_{\rm inter}$. The dark- and light- green curves represent the components $\delta\phi^1$ and $\delta\phi^2$, respectively, while the arrows give their direction of propagation.
The left (respectively right) panel illustrates the behavior of $\delta\phi^{1,2}$ when $\bomega$ belongs to the antisymmetric subspace $\mcA$ (inter-symmetric subspace $\mcS_{\rm inter}$, respectively) for given order $\ell$ and coefficient $\alpha_{12}$ (resp. $\beta_{12}$): the components propagate in a common direction which depends on whether their phase difference $\Delta\theta$ belongs to $(-\pi,0)$ or $(0,\pi)$ (respectively to $(-\pi/2,\pi/2)$ or $(\pi/2,3\pi/2)$). The top panel illustrates the behavior of a component $\delta\phi^i$ when $\bomega$ belongs to the self-symmetric subspace $\mcS_{\rm self}$: it propagates in a direction that solely depends on the order $\ell$ and the coefficient $\beta_{ii}$. 
}
\end{center}
\end{figure}
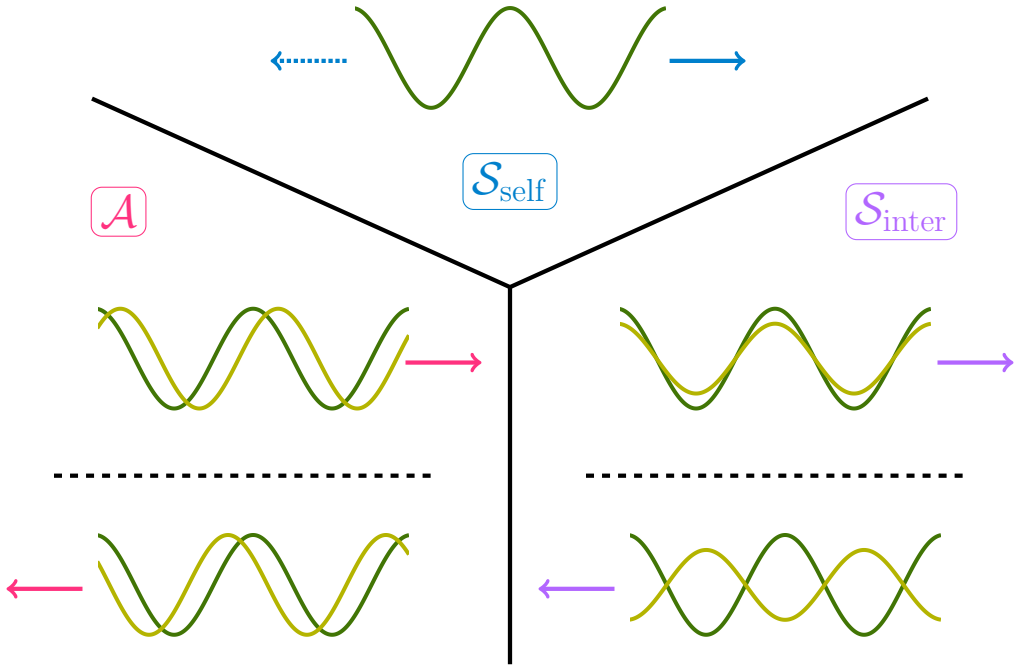

A lot of attention has been recently devoted to non-reciprocally interactive mixtures, notably for their relevance in biology. When
the interacting fields are scalar, the non-reciprocity of their interactions has been identified as a generic route to oscillating and traveling patterns, from the non-reciprocal Swift-Hohenberg model~\cite{fruchart2021non} (NRSH) to the non-reciprocal Cahn-Hilliard system~\cite{saha2020scalar,you2020nonreciprocity,frohoff2021suppression,suchanek2023entropy} (NRCH).
In the NRCH, the operator $\bb$ is made up of several divergence operators stacked on top of each other (one for each species), \ie $\bb=(\nabla\cdot,\dots,\nabla\cdot)^\top$, while in the Swift-Hohenberg model\footnote{The non-reciprocal Swift-Hohenberg model that appears in~\cite{fruchart2021non} is purely deterministic. However, no field is conserved in this model so the dominant noise term at large scales should be proportional to the identity operator, a case often considered in the literature on the classical (\ie reciprocal) Swift-Hohenberg model~\cite{hohenberg1992effects,vinals1991numerical,mohammed2013modulation,mohammed2015stochastic}.} $\bb$ is usually taken to be the identity (we disregard any multiplicative constant here).
In the various versions of the NRSH and NRCH models cited above, the drift is always of the form $\ba=\bD\bzeta$ where $\bzeta=-\delta\mcF / \delta\bphi - \balpha \bphi$, $\balpha$ being a spatially-uniform matrix that is not symmetric and hence breaks the reciprocity of the interactions. (We here consider that $\balpha$ is antisymmetric, as its symmetric part can always be absorbed in the functional derivative term in $\bzeta$.)
In both cases, $\bomega=\mbd\bzeta=\alpha_{ij}\delta^i\wedge \delta^j$ and $\whomega=2\balpha$ so that they belong to the antisymmetric subspace and are of order $\ell=0$. 
Their respective vorticity operator also belongs to the antisymmetric subspace, with that of the NRSH being of order $\ell=0$ while that of the NRCH is of order $\ell=1$.
Indeed, while the diffusion operator is the identity in the former model, so that $\bW=-\whomega/2=-\balpha$, it is the identity matrix multiplied by minus the Laplacian operator in the latter, so that $\bW=\Delta\whomega/2=\balpha\Delta$. Note that, as announced in the beginning of section~\ref{subsec:basisPhenomeno}, the fact that $\bD\propto \Delta$ in the NRCH model only changes the order of the derivative and not the subspace to which $\bW$ belongs, as compared to $\bomega$ or $\whomega$.
All the references cited above about these two models have reported oscillatory and propagative behaviors as if one component was chasing the other, which is in accordance with the classification proposed in section~\ref{sec:basis} and its corresponding phenomenological picture. This is, surprisingly, in spite of their respective stationary probability not being concentrated in the vicinity of a given static profile as it is the case \eg for the AMB.
Finally, it is worth emphasizing that the matrix $\balpha$ in these models is always considered uniform, so that the possible exotic dependence on the spatial subdomain or spatial scale of which component is ``chasing'' the other (as mentioned at the end of the previous paragraph) should not appear and does not seem to have been observed. It would thus be most interesting to either find or invent such a model.

Another example of a dynamics whose vorticity two-form belongs to the antisymmetric subspace $\mcA$ is given the following reaction-diffusion dynamics between species $u$ and $v$
\begin{eqnarray}
\partial_t u &=& \Delta u + P(u,v) + \sqrt{D_u}\eta_u \ , \label{turing_u}\\
\partial_t v &=& \Delta v + Q(u,v) + \sqrt{D_v}\eta_v \ , \label{turing_v}
\end{eqnarray}
where the strictly-local functionals $P$ and $Q$ are the reaction terms, $D_u$ and $D_v$ are constants, $\boldeta(x,t)\equiv(\eta_u,\eta_v)^\top$ is a random field with the same statistical properties as in eq.~\eqref{EDPS02}, and the space-time units are chosen such that the diffusion constants of $u$ and $v$ are set to unity. Note that the noise should also contain conservative terms arising from diffusion, but these are typically negligible when compared to non-conserved terms on sufficiently large spatial scales.
The vorticity two-form associated with dynamics~\eqref{turing_u}-\eqref{turing_v} is straightforwardly shown to read
\begin{equation}
\bomega = \left(D_u^{-1} \frac{\partial P_x}{\partial v} - D_v^{-1} \frac{\partial Q_x}{\partial u} \right) \delta^{vx}\wedge \delta^{ux} \ ,
\end{equation}
and hence indeed belongs to the antisymmetric subspace $\mcA$, in accordance with the traveling patterns notoriously observed in out-of-equilibrium reaction-diffusion equations~\cite{cross2009pattern}.

When the field $\bphi$ takes values in an Euclidean space of dimension $d_2>2$, the generic phenomenological picture of the antisymmetric subspace depicted above for $d_2=2$ remains very similar when observed in $\mbR^{d_2}$, although its physical manifestation may appear different at first sight.
Indeed, for any antisymmetric matrix $\balpha$ of size $d_2\times d_2$, there exists a basis in which $\balpha$ can be block-diagonalized, with each block being two-by-two and of the form of the right-hand side of eq.~\eqref{Omega_alpha_family} (in odd dimension, there is an additional column of zeros in the matrix)\cite{gantmacher1998thoery}. The phenomenology of the dynamics of each projection of $\delta\bphi$ onto one of the two-dimensional subspaces of $\mbR^{d_2}$ corresponding to a block is hence the same as the generic two-dimensional phenomenology depicted above. The global phenomenology is then a superposition of that of each two-dimensional projection.
This similarity might be hard to perceive when $\bphi$ represents \eg a pair of two-dimensional vector-field stacked together, in which case one directly observes the collective time-evolution of the two vector fields over $\mbR^2$ rather than that of $\bphi$ in $\mbR^4$. This is the case for instance in the field dynamics presented in~\cite{fruchart2021non} which is the coarse-grained hydrodynamics of a mixture of two species of self-propelled particles where one species tries to align with the second, while the latter tries to anti-align with the former. The authors of this article reported several non-equilibrium phases, namely a chiral phase, a swap phase, and phase that is both chiral and swap.
In appendix~\ref{app:NR_flock}, we study a stochastic version of this model and show its vorticity two-form to belong to the antisymmetric subspace. We then show that the coordinates of this vorticity two-form accounts for the chiral, swap, and chiral-swap phases (at least at linear order).

\subsubsection{The symmetric subspace: phase separation, interface growth, and flocking}
\label{subsubsec:oddSubspace}

We now consider a cycle affinity operator of the symmetric subspace that reads $\whomega=-(\bbeta\partial_x^{2\ell+1} + \partial_x^{2\ell+1} \bbeta)$ where $\ell$ is an integer and $\bbeta$ an $x$-dependent symmetric matrix.
Again, since $\bD$ is assumed to be the identity, the associated vorticity operator is
\begin{equation}
\bW = \frac{1}{2}\left[\bbeta\partial_x^{2\ell+1} + \partial_x^{2\ell+1}\bbeta \right] \ .
\label{oddSubspace_vorticityOperator}
\end{equation}
Below, we separate the analysis of self-symmetric and inter-symmetric terms in the matrix $\bbeta$. This may seem redundant at first sight since $\bbeta$, as a real symmetric matrix, can always be diagonalized in an appropriate basis in which, in turn, $\bomega$ would belong to the self-symmetric subspace. And indeed, if the field $\bphi$ is a proper $\mbR^{d_2}$-valued vector field, then changing the basis amounts to adopting an equivalent representation, and the study boils down to the case where $\bbeta$ is diagonal.
However, if $\bphi$ is not a true vector field, but rather a collection of different fields (\eg of mass and polarity) stacked together, then changing basis in $\mbR^{d_2}$ may mix components of different natures and hence does not make sense from a physical standpoint.
 In the latter case, (block-) diagonal and off-(block-)diagonal terms in $\bbeta$ are drastically different, as they respectively correspond to self- and cross-interactions of the various fields that composed $\bphi$.
We still focus on the $d_2=2$ case and further assume that the two components of $\bphi$ represent different (scalar) fields. Consequently, $\bbeta$ is a two-by-two matrix whose self-symmetric (respectively inter-symmetric) components correspond to the diagonal (respectively off-diagonal) ones.

\paragraph{Self-symmetric subspace.}
In this paragraph, we consider the case where $\bbeta$ is diagonal, \ie $\beta_{12}=\beta_{21}=0$.
Just as we previously did for the antisymmetric subspace, we start by considering the limit-regime where $\ell_{\delta\phi}\gg\ell_W$ in which the vorticity operator acts by multiplication by the symmetric matrix $\partial_x^{2\ell+1}\bbeta/2$: 
\begin{equation}
\bW\simeq [\partial^{2\ell+1}_x\bbeta/2] = \begin{bmatrix}
 \partial^{2\ell+1}_x\beta_{11}/2 & 0 \\
0 &  \partial^{2\ell+1}_x\beta_{22}/2\\
\end{bmatrix} \ .
\label{diagonal_long_waveLength_Omega}
\end{equation} 
If $[\partial^{2\ell+1}_x\bbeta/2]\neq 0$ -- which exclude in particular the case where $\bbeta$ is independent of $x$ -- then the $i^{th}$ component of $\bphi$ either increases or decreases exponentially at rate $\partial^{2\ell+1}_x\beta_{ii}/2$, which typically happens until $\ell_{\delta\phi}\sim\ell_W$, a time at which the approximation~\eqref{diagonal_long_waveLength_Omega} is no longer valid.
Note that, as this regime must hence be transient, its nature is quite different from that of the same limit in the antisymmetric subspace. It does not describe the irreversible evolution of fluctuations in the vicinity a given profile, but rather a shift in the latter profile, as we saw for the AMB example in section~\ref{subsec:vorticityOperator}.

In the opposite limit-regime, when $\ell_{\delta\phi}\ll\ell_W$, the vorticity operator approximately reads 
\begin{equation}
\bW\simeq \bbeta\partial^{2\ell+1}_x = \begin{bmatrix}
\beta_{11} & 0 \\
0 & \beta_{22} \\
\end{bmatrix} \partial_x^{2\ell+1} \ .
\end{equation}
Then, if we focus on a generic harmonic mode perturbation $\delta\bphi_k$ that is initially of the form~\eqref{bpsik_initial}, the solution of dynamics~\eqref{BOmega_dynamics_1d} is given by (see appendix~\eqref{app:betaFamily})
\begin{equation}
\delta\bphi_k(t)= \begin{bmatrix} \rho_1 \cos(kx + W_k^1 t + \theta_1) \\ \rho_2 \cos(kx + W_k^2 t + \theta_2) \end{bmatrix} \ ,
\label{solutionBetaFam01}
\end{equation}
with
\begin{equation}
W_k^j = (-1)^{\ell}k^{2\ell+1} \beta_{jj} \ .
\end{equation}
Hence, each component propagates independently along $\mbR$ at speed and direction imposed by $W^j_k/k=(-1)^{\ell}k^{2\ell} \beta_{jj}$ .
In~\cite{o2023nonequilibrium}, the author showed the AMB and the KPZ equation to display such a phenomenology, and proved that their vorticity two-forms are both generated by the self-symmetric subfamily of the basis~\eqref{1dBasis2Forms}. Note that this fact is directly implied by two properties shared by the AMB and the KPZ dynamics: their respective one-forms $\bD^{-1}\ba$ are both local in the field; moreover both these dynamics describes the time evolution of a sole scalar-valued field, \ie $d_2=1$, a case in which the basis~\eqref{1dBasis2Forms} reduces to the self-symmetric family. This confirms the intuition that run-and-chase behaviour cannot emerge from a single species model. 
Finally, the phenomenological predictions made in~\cite{o2023nonequilibrium} were based on a qualitative analysis of the respective (functional) probability currents of AMB and KPZ. It is worth emphasizing that these predictions can be obtained more straightforwardly and specified by computing the corresponding vorticity operators $\bW$ and analyzing their flow, as we have demonstrated in section~\ref{subsubsec:vorticityOp} for AMB. A similar analysis can be carried out for KPZ dynamics, which reads:
\begin{equation}
\partial_t h = \kappa \Delta h + \lambda |\nabla h |^2 + D\eta \ ,
\end{equation}
where $h(x,t)$ is the height field, $D$ the diffusivity constant, $\eta$ a random field with the same statistics as in eq.~\eqref{EDPS02}, and $\kappa(h(x,t))$ and $\lambda(h(x,t))$ are considered as local functions of $h$ for greater generality.
The vorticity two-form of KPZ equation coincides with that of AMB (eq.~\eqref{cycleAffAMB}), while its vorticity operator is readily found  -- using formula~\eqref{Omega_vs_whomega} -- to be:
\begin{equation}
\bW=-\frac{D}{2}\left\{ (2\lambda+\kappa')\nabla h \cdot \nabla\delta h + \nabla\cdot \left[\delta h (2\lambda+\kappa')\nabla h \right]\right\} \ .
\end{equation}

\paragraph{Inter-symmetric subspace.} Let us now turn to the inter-symmetric subspace, which can be embodied by a matrix $\bbeta$ whose diagonal coefficients are both identically zero.
In the first limit regime, where $\ell_{\delta\phi}\gg\ell_W$, $\bW$ acts as a multiplication by the matrix $[\partial^{2\ell+1}_x\bbeta/2]$:
\begin{equation}
 \bW\simeq [\partial^{2\ell+1}_x\bbeta/2] = (\partial^{2\ell+1}_x\beta_{12}/2]) \begin{bmatrix}
0 & 1 \\
1 & 0 \\
\end{bmatrix} \ .
\label{offDiagonal_long_waveLength_Omega}
\end{equation}
The solution of the vortex dynamics~\eqref{BOmega_dynamics_1d} is then
\begin{equation}
\delta\bphi(\br,t) =  \begin{bmatrix}
 \cosh (\nu t) & \sinh(\nu t) \\
 \sinh(\nu t) & \cosh (\nu t) \\
\end{bmatrix} \delta\bphi(\br,0) \ ,
\end{equation}
with $\nu\equiv \partial^{2\ell+1}_x\beta_{12}/2$.
At time $t\gg\nu^{-1}$, this solution asymptotically becomes
\begin{equation}
\delta\bphi(t) \simeq \frac{e^{|\nu| t}}{2}(\delta\phi^1(0)+\textrm{sign}(\nu)\delta\phi^2(0)) \begin{bmatrix}
1 \\ \textrm{sign}(\nu) \\
\end{bmatrix} \ ,
\end{equation}
\textit{i.e.}, at long time, if $\nu>0$ (respectively $\nu<0$), $\delta\phi^1$ and $\delta\phi^2$ becomes equal (respectively opposite) and diverge exponentially fast at rate $|\nu(x)|$, until $\ell_{\delta\phi}\sim\ell_W$, at which point  approximation~\eqref{offDiagonal_long_waveLength_Omega} is not longer valid. Just as in the self-symmetric case, this regime should not be seen as a heuristic description of the fluctuations' dynamics in steady state around a typical profile, but rather a contribution to the creation of this typical profile.

In the opposite $\ell_{\delta\phi} \ll \delta_W$ regime, the vorticity operator~\eqref{oddSubspace_vorticityOperator} is approximately given by
\begin{equation}
\bW \simeq \bbeta \partial_x^{2\ell+1} =  \beta_{12}
\begin{bmatrix}
0 & 1 \\
1 & 0 \\
\end{bmatrix} \partial_x^{2\ell+1} \ .
\end{equation}
The solution of dynamics~\eqref{BOmega_dynamics_1d} starting from a pure harmonic perturbation $\delta\bphi_k(0)$ given by eq.~\eqref{bpsik_initial} reads:
\begin{equation}
\delta\bphi_k(x,t) =  \rho^+ \cos(kx+W_k^1t+\theta^+)  \begin{bmatrix} 1 \\ 1\end{bmatrix} +  \rho^- \cos(kx-W_k^1t+\theta^-) \begin{bmatrix} 1 \\ -1\end{bmatrix} \ ,
\label{solutionBetaFam02}
\end{equation}
where $W_k^1 = (-1)^\ell k^{2\ell+1} \beta^{(\ell)}_{12}$, and $\rho^\pm,\theta^\pm$ are given in appendix~\ref{app:betaFamily}.
We see that, in this case, the `$+$' and `$-$' components, which are respectively defined as the first and second term of the right-hand side of eq.~\eqref{solutionBetaFam02}, have coordinates $\delta\phi_k^1,\delta\phi_k^2$ that are respectively in phase and phase opposition and propagate on $\mbR$ in the directions $\text{sign}(-W_k)$ and $\text{sign}(W_k)$, respectively (see figure~\ref{fig:basis_phenomeno}).
%
%

In figure~\ref{fig:basis_phenomeno}, we summarize the typical behavior of both the antisymmetric and inter-symmetric subspaces in the $\ell_{\delta\phi}\ll\ell_{W}$ regime. Their main difference is that the direction of propagation depends on whether $\Delta\theta$ belongs to $(-\pi,0)$ or to $(0,\pi)$ in the antisymmetric subspace, while it depends on whether $\Delta\theta$ belongs to $(-\pi/2,\pi/2)$ or to $(\pi/2,3\pi/2)$ in the inter-symmetric subspace, where $\Delta\theta$ denotes the phase difference between $\delta\phi^1$ and $\delta\phi^2$. In other words it depends on which component is ``behind the other'' in the former case, and on whether these components are ``together or apart'' in the latter.

Interestingly, to the best of our knowledge, there is no field theory in the literature describing the collective dynamics of the densities of two active species whose vorticity two-form belongs to the inter-symmetric subspace. In such a case, it could describe a kind of cooperation phenomenon between the two species: if they are co-localized, they propagate in one direction, while if they are demixed, they propagate in the other.

On the other hand, if we now think about the components of $\delta\bphi$ as being respectively a density and a polarization in one spatial dimension, then the inter-symmetric phenomenology depicted above strikingly resembles that of flocks of aligning self-propelled particles.
To confirm this correspondence, we now consider a stochastic (coarse-grained) Active Ising Model (AIM)~\cite{solon2015flocking}, which reads:
\begin{eqnarray}
\partial_t\rho &=& D \partial_x^2\rho - v\partial_x m +\partial_x \sqrt{2D} \eta_\rho \ ,\label{AIM_dynamics_rho}
\\
\partial_t m &=& D\partial_x^2 m - v\partial_x\rho + \gamma_1 m + \gamma_2 m^3 + \sqrt{2D} \eta_m \ ,
\label{AIM_dynamics_m}
\end{eqnarray}
where $D$ is a diffusion constant, $v$ is the self-propulsion speed, and $\gamma_{1,2}$ parametrize the aligning interactions. In the original version of the AIM, $\gamma_1$ is constant whereas $\gamma_2$ depends on $\rho(x)$. This latter fact would generate a term in the vorticity two-form that stems from an irreversibility in the alignment interaction (see appendix~\ref{app:active_Ising}). As we want to focus here on the irreversibility that results from the self-propulsion only, we consider a simplified version of the model where $\gamma_2$ is also constant. 
 In eqs.~\eqref{AIM_dynamics_rho}-\eqref{AIM_dynamics_m}, $\eta_{\rho,m}(\br,t)$ are independent random Gaussian fields whose statistics are the same as the components of $\boldeta$ in eq.~\eqref{EDPS02}. These noise terms are absent in the hydrodynamic equation initially derived from the microscopic AIM~\cite{solon2015flocking}. We add them here to account for the macroscopic fluctuations around the hydrodynamic limit in the simplest possible way\footnote{Note that the structure of the coasre-grained noise was explicitly computed in~\cite{scandolo2023active} for several versions of AIM, not including our constant $\gamma_2$ AIM.} that is compatible with the fact that $\rho$ is conserved and $m$ is not.
Consequently, the resulting dynamics~\eqref{AIM_dynamics_rho}-\eqref{AIM_dynamics_m} can be seen as a stochastic diffusion equation and a model A dynamics that are coupled together by self-propulsion.

In appendix~\ref{app:active_Ising}, we show the vorticity two-form of dynamics~\eqref{AIM_dynamics_rho}-\eqref{AIM_dynamics_m} to read
\begin{equation}
\bomega \equiv \mbd \bD^{-1}\ba = -\frac{v}{D}\delta^{\rho x} \wedge \partial\delta^{m x} + \frac{v}{D} G_{xy} \delta^{\rho x} \wedge \partial\delta^{m y} \ ,
\label{AIM_w_2form}
\end{equation}
where $G_{xy}\equiv G(x-y)$ is the Green function of the Laplacian in one dimension, \ie $\partial_x^2 G(x-y)=\delta(x-y)$.
As anticipated, $\bomega$ has a component along the inter-symmetric subfamily: the first one on the right-hand side of eq.~\eqref{AIM_w_2form}.
Interestingly, $\bomega$ also possesses a non-local component, the second term on the right-hand side of eq.~\eqref{AIM_w_2form}, and is thus not entirely generated by the local vorticity basis~\eqref{1dBasis2Forms}. Note that this does not contradict our general result about the fact that the family~\eqref{1dBasis2Forms} is a basis for the space $\Omega^2_{\rm loc}(\mbF)$ precisely because the one-form $\bD^{-1}\ba$ of dynamics~\eqref{AIM_dynamics_rho}-\eqref{AIM_dynamics_m} is non-local.
Despite the presence of this non-local component, we show in appendix~\ref{app:active_Ising} that the vorticity dynamics~\eqref{Vortex_dynamics} associated to the two-form~\eqref{AIM_w_2form} displays a phenomenology very similar to that of the inter-symmetric subfamily, described above and summarized in Fig.~\ref{fig:basis_phenomeno}. More precisely, a harmonic perturbation whose density and polarity components are in phase propagates in the direction of increasing $x$, while it propagates in the opposite direction when the components are in antiphase. 
This confirms the heuristic correspondence drawn above between the phenomenology one could expect from the inter-symmetric subspace and collective motion.


In the next subsection, we discuss the results of subsections~\ref{subsec:basisOf2forms}-\ref{subsec:basisPhenomeno} and their limitations, and elaborate on their possible extensions to broader situations.

\subsection{Discussion}
\label{subsec:discussion_extension}

In section~\ref{sec:basis}, we have constructed a basis of the space $\Omega_{\rm loc}^2(\mbF)$ of local two-forms in $d_1=1$ dimension. This basis first provides explicit reversibility condition(s) for a dynamics~\eqref{EDPS02} whose one-form $\bD^{-1}\ba$ is local in $\bphi$.
Using the vortex dynamics~\eqref{Vortex_dynamics}, we then showed that it could be decomposed in three subfamilies, with the elements of each displaying generic phenomenologies.
Surprisingly, this approach exposes many celebrated out-of-equilibrium phenomena of which it provides a unified description.
This viewpoint could also be useful to ``transport a phenomenology'' from an active field theory to another, as we have seen that \eg the flocking behavior of the symmetric family has not been observed yet in mixture of active species, where it could be interpreted as describing some sort of cooperative behavior.

However, we point out that the correspondence between the subspace of $\Omega_{\rm loc}^2(\mbF)$ and the out-of-equilibrium phenomenology described above strictly applies only when the diffusion operator $\bD$ is proportional to the identity, 
since the typical irreversible phenomenology of an out-of-equilibrium dynamics of the form~\eqref{EDPS02} is given by the vorticity operator $\bW=-\bD\whomega/2$ rather than by the cycle-affinity operator (or equivalently the vorticity two-form $\bomega$).
Interestingly, this correspondence remains qualitatively unchanged on small scales $\ell_{\delta\phi}\ll \ell_W$ when 
\begin{equation}
\bD \underset{\ell_{\delta\phi}\ll \ell_W}{\simeq} \sum_p D_p \partial_x^{2p} 
\label{a_nice_D_operator}
\end{equation}
where each $D_p$ is a constant scalar. This is because an element of $\whomega$, belonging to one of the three identified subspaces, remains in that subspace upon multiplication by a derivative of order two.
But a field theory whose diffusion operator is not of the form~\eqref{a_nice_D_operator} on small scales
could display an irreversible phenomenology which deviates from those predicted above for the same vorticity two-form.
For such a diffusion operator, the general phenomenological study we conducted in~\ref{subsec:basisPhenomeno} should thus be adapted to describe the class of stochastic dynamics~\eqref{EDPS02} sharing this given $\bD$.

Besides, in spatial dimension $d_1>1$, the phenomenology of the family~\eqref{1dBasis2Forms_higherd1} is very similar to its one-dimensional counterpart~\eqref{1dBasis2Forms}.
Nevertheless, we have seen in section~\ref{subsec:higher_dimension_d1} that, contrary to the $d_1 =1$ case, the family~\eqref{1dBasis2Forms_higherd1} is not a basis of the space of local two-forms $\Omega_{\rm loc}^2(\mbF)$. Thus, this family remains to be completed appropriately, and the phenomenology of the resulting additional subspace $\mcC$ of $\Omega_{\rm loc}^2(\mbF)$ explored, in future work. 

We have also seen, with the example of the AIM, that the one-form $\bD^{-1}\ba$ of dynamics~\eqref{EDPS02} can be non-local in $\bphi$, in which case a basis of local vorticities is not sufficient to decompose the corresponding vorticity $\bomega\equiv \mbd \bD^{-1}\ba$. 
We denote by $\Omega^2(\mbF)$ the space of arbitrary two-forms (\ie which are not necessary local) and $\mcN$ its subspace which is made of purely non-local two-forms\footnote{In general $\bomega= \omega_{ijxy}\delta^{ix}\wedge \delta^{jy}$. We say that a two-form $\bomega$ is purely non local iff $\omega_{ijxy}$ does not contain Dirac $\delta(x-y)$ nor any of its derivatives in any direction.}.
The question of whether any general properties of the elements of $\mcN$ could be exhibited remains to be studied. 

Importantly, we proved the family~\eqref{1dBasis2Forms} to be a basis of two-forms only when the space $M$ on which the field $\bphi$ is defined is the real line, $M=\mbR$. This proof can easily be adapted if $M$ is the circle but, if there are boundaries, it fails and the family~\eqref{1dBasis2Forms} must then be completed with two-forms that are \textit{localized at the boundary $\partial M$ of $M$}\footnote{By ``localized at the boundary of $M$'', we mean that these additional two-forms, applied to a pair $(\delta\bphi_1,\delta\bphi_2)$ of fluctuations around a given $\bphi$, only depend on the values of $\bphi,\delta\bphi_1,\delta\bphi_2$ and their derivatives (up to a finite order) at the boundary of $M$.}. Similarly, if the dimension $d_1$ of the space $M$ is higher than one, the generalized decomposition~\eqref{Decomposition_space_local_vorticities_higher_d1} is valid only when $\partial M = \emptyset$. When that is not the case, this decomposition should be completed with the space of vorticity two-forms that are localized at $M$, which we denote by $\mcB$.
Although these vorticities $\bomega\in\mcB$ are beyond the scope of this article, it is worth noting that they might be good candidates to systematically account for the notoriously diverse and surprising behavior of active systems at their boundaries~\cite{souslov2019topological,ben2022disordered,shankar2022topological}.

All in all, the space of two-forms $\Omega^2(\mbF)$, to which the vorticity of all dynamics of the form~\eqref{EDPS02} belongs, can be decomposed as
\begin{equation}
\Omega^2(\mbF) = \mcA \oplus \mcS_{\rm self} \oplus \mcS_{\rm inter}\oplus \mcC \oplus \mcN \oplus \mcB \ .
\label{general_decomposition_vorticity_space}
\end{equation}
Interestingly, we show in section~\ref{sec:entropy} the entropy production rate of dynamics~\eqref{EDPS02} to be a linear functional of the vorticity $\bomega$. Consequently, ~\eqref{general_decomposition_vorticity_space} can be seen as a \textit{decomposition of independent sources of entropy production}. 
In section~\ref{sec:basis}, we studied the first three components of this decomposition, while the exploration of the others remains an exciting challenge for the future.
 
Beyond the perspective of bringing order to the vast zoo of out-of-equilibrium field theories, decomposition~\eqref{general_decomposition_vorticity_space} could also be used to build a nonequilibrium dynamics with desired properties. To this end, one could start from an equilibrium limit of the dynamics that is sought, then construct a vorticity two-form with the appropriate components in each subspace~\eqref{general_decomposition_vorticity_space}, then integrate it\footnote{Note that to be integrated, the contructed two-forms must be exact. To verify that this is the case, one could try to use the coordinates expression~\eqref{eq:expressionAlpha}-\eqref{eq:expressionBeta} of a generic exact two-form. But a more practical criterion would be to check whether the exterior derivative of this two-form vanishes.} and multiply the resulting $\bD^{-1}\ba$ by $\bD$ to obtain a functional vector field whose addition to the equilibrium drift should produce the desired effects, at least qualitatively.
Finally, we will see in section~\ref{sec:MTreversal} that new sources of entropy production can add up to $\bomega$, either when considering a more general notion of time-reversal, or when the space $\mbF$ is not simply connected.

\newpage
\section{Vorticity and entropy production: what is captured and what is missed}
\label{sec:entropy}

\subsection{Entropy production and the vorticity 2-form}
\label{subsec:entropyProd}

To further stress the role or the vorticity two-form $\bomega$ in the possible TRS-breaking of dynamics~\eqref{EDPS02}, we study in this section its relation to entropy production. In particular, we show that $\bomega$ can be interpreted as the entropy production per unit area in the space $\mbF$.

\subsubsection{Path-wise entropy production}
As stated at the beginning of section~\ref{subsec:RevCondFuncExtDer}, all the components of $\bphi$ are supposed to be even under time-reversal (see section~\ref{sec:MTreversal} for a more general situation, involving odd degrees of freedom). In such a situation, the path-wise entropy production~\cite{seifert2005entropy} -- more precisely the path-wise {\em informatic} entropy production~\cite{fodor2022irreversibility} -- of dynamics~\eqref{EDPS02} along a trajectory $(\bphi_t)_{t\in\mbT}$ is defined as:
\begin{equation}
\whSigma[(\bphi_t)_{t\in\mbT}] \equiv \ln \frac{\mcP[(\bphi_t)_{t\in\mbT}]}{\mcP[(\bphi_{\mcT-t})_{t\in\mbT}]} \ . 
\label{pathWiseEntProdDef}
\end{equation}
A direct computation\footnote{We do not reproduce this now classical computation in this article. But thanks to our geometrical conventions of notation, it can be directly and straightforwardly deduced from the same computation in finite dimension (see for instance our companion paper~\cite{o2024geometric}) by replacing the finite dimensional indices $i,j,\dots$ by their continuous analogs $i\br, j\br',\dots$ and using the generalized Einstein convention introduced in this article.},
using \eg the Onsager Machlup functional, shows $\whSigma$ to obey
\begin{equation}
\whSigma[(\bphi_t)_{t\in\mbT}] = \int_0^\mcT [D^{-1}a]_{i\br}\partial_t\phi^{i\br} \rmd t -\ln \Pss[\bphi_\mcT] + \ln\Pss[\bphi_0] \ ,
\label{pathWiseEntProd01}
\end{equation}
where the stochastic integral is defined in the Stratonovich sense.
We directly read from this equation that, just as in finite dimension, $\whSigma$ interestingly does not depend on the time-parametrization of the oriented path $\mcC\equiv (\bphi_t)_{t\in\mbT}$. This fact further underlines the role of geometry in the study of the behavior of dynamics~\eqref{EDPS02} with respect to time-reversal.

Dynamics~\eqref{EDPS02} is reversible iff $\whSigma$ vanishes for every continuous path in $\mbF$. It turns out that this proposition remains true if the set of continuous path is restricted to the set of continuous loops (this can be seen \eg by comparing eqs.~\eqref{loopWiseEntProd02} \&~\eqref{fieldIEPR02} below). Therefore, let us now consider an oriented path $\mcC\equiv (\bphi_t)_{t\in\mbT}$ that is a loop in $\mbF$, \ie such that $\bphi_0=\bphi_{\mcT}$. 
The boundary terms in the expression~\eqref{pathWiseEntProd01} of $\whSigma$ then vanish and the entropy produced along $\mcC$ reads
\begin{equation}
\whSigma[\mcC] = \int_{\mcC} \bD^{-1}\ba \ .
\label{loopWiseEntProd01}
\end{equation}
In eq.~\eqref{loopWiseEntProd01}, the right-hand side stands for the integral of the one-form $\bD^{-1}\ba$ along $\mcC$, which can be computed explicitly \via any parmetrization $(\bphi_t)_{t\in\mbT}$ of $\mcC$, in which case it takes the form of the integral on the right-hand side of eq.~\eqref{pathWiseEntProd01}.
Having assumed that $\mbF$ is simply connected, there exists an oriented surface $\mcS\subset\mbF$ of which $\mcC$ is the boundary, a property that we denote by $\mcC=\partial\mcS$.
Thanks to our functional Stokes' theorem (see section~\ref{subsec:RevCondFuncExtDer} and appendix~\ref{App:subsec:Stokes}) we deduce that
\begin{equation}
\whSigma[\partial\mcS]=\int_{\partial \mcS} \bD^{-1}\ba = \int_\mcS\mbd \bD^{-1}\ba =\int_\mcS \bomega \ ,
\label{loopWiseEntProd02}
\end{equation}
which is a generalization to functional spaces of a result recently obtained in finite dimension~\cite{yang2021bivectorial}.
In order to interpret what this last equation says about the nature of $\bomega$, we can let the surface $\mcS$ become infinitesimal around a given $\bphi\in\mbF$, $\mcS\to \delta \mcS_{\bphi}$, so that the entropy produced along this oriented infinitesimal surface is $\whSigma[\delta\mcS]=\bomega_{\bphi} \cdot \delta \mcS_{\bphi}$. In this last equation $\bomega_{\bphi} \cdot \delta \mcS_{\bphi}$ is an informal notation for $\bomega_{\bphi}(\delta\bphi_1,\bphi_2) = \omega_{i\br j\br'}[\bphi]\delta\phi^{i\br}_1\delta\phi^{j\br'}_2$, where $\delta\bphi_1,\delta\bphi_2\in T_{\bphi}\mbF$ spans $\delta\mcS_{\bphi}$. Therefore, $\bomega\equiv\mbd\bD^{-1}\ba$ can be interpreted as \textit{the entropy production per unit (oriented) area} in $\mbF$, as claimed at the beginning of this section.

\subsubsection{Entropy production rate}
Along the lines of this interpretation, we can even show that the (global) informational entropy production rate (IEPR) ,
which is defined as
\begin{equation}
\sigma\equiv \lim_{\mcT\to\infty} \frac{1}{\mcT}\whSigma[(\bphi_t)_{t\in\mbT}] \label{def_IEPR} \ ,
\end{equation}
 is itself a linear functional of $\bomega$. Indeed, a computation that is similar to what is done in finite dimension (see \eg~\cite{o2024geometric}), first leads to the following expression of the IEPR:
\begin{equation}
\sigma = \llangle [D^{-1}a]_{i\br}\partial_t\phi^{i\br} \rrangle = \int_{\mbF}   [D^{-1}a]_{i\br} \Jss^{i\br} \mcD\bphi \ ,
\label{fieldIEPR01}
\end{equation}
where $\llangle \dots \rrangle$ designates the average in steady state and $\bJss$ is the stationary probability current associated to dynamics~\eqref{EDPS02}.
Let us denote by $\DDiv$ the functional divergence operator (with respect to the Lebesgue measure $\mcD\bphi$), which associates to a vector field $\bv[\bphi]$ over $\mbF$ the scalar-valued functional 
\begin{equation}
\DDiv(\bv)\equiv \frac{\delta v^{i\br}}{\delta \phi^{i\br}} \ .
\label{def_div}
\end{equation}
Since $\bJss$ is stationary, it is divergence-free: $\DDiv(\bJss)=0$. Formally extending Hodge-de Rham theory~\cite{warner1983foundations} to this functional context, because $\mbF$ is simply-connected, we can conclude that there exists a functional, antisymmetric, contravariant tensor of order two, $\bC$, such that $\bJss=-\DDiv(\bC)$, \ie $\Jss^{i\br} = -\frac{\delta C^{j\br' i\br}}{\delta \phi^{j\br'}}$.
%
In turn, we can inject the latter expression into formula~\eqref{fieldIEPR01} and, upon (functionally) integrating by parts and using the antisymmetry of $\bC$, we get
\begin{equation}
\sigma = \frac{1}{2}\int_{\mbF}  \bC\cdot \bomega \; \mcD\bphi\ ,
\label{fieldIEPR02}
\end{equation}
where $\bC\cdot \bomega \equiv C^{j\br' i\br} \omega_{j\br' i\br}$. We see that the IEPR is indeed a linear functional of the vorticity two-form $\bomega$. Once again, this generalizes to field theory a result recently obtained for finite dimensional systems~\cite{yang2021bivectorial}. 
Note that, because there is no stationary probability current when $\bomega=0$, \ie when dynamics~\eqref{EDPS02} is reversible, we can conclude that the presence of a non-vanishing $\bJss$ is a consequence of the irreversibility, which is itself driven by a non-zero $\bomega$. Since $\bC$ is simply a ``primitive'' of $\bJss$, we can interpret $\bomega$ as \textit{the source of entropy production} in~\eqref{fieldIEPR02}, while $\bC$ is only a consequence of it.
Further, following our interpretation of $\bomega$ from eq.~\eqref{loopWiseEntProd02}, we can in turn informally interpret $\bC[\bphi]$ in eq.~\eqref{fieldIEPR02} as a weight over all the infinitesimal loops that surrounds point $\bphi\in\mbF$.


\subsubsection{Loop-wise entropy production and phenomenology of AMB}
\label{subsubsec:loopWiseEP_AMB}

In this section, we show how the loop-wise entropy production formula~\eqref{loopWiseEntProd02} can be used as an alternative way of describing the irreversible phenomenology of AMB, as was first uncovered in~\cite{o2023nonequilibrium} and revisited in section~\ref{subsec:vorticityOperator}. 
To do this, let us denote by $\rhoss$ the average stationary profile of a phase-separated droplet of liquid surrounded by a gaseous phase. We can then consider a certain oriented loop $\mcC$ around $\rhoss\in\mbF$ and compute the corresponding entropy production $\whSigma[\mcC]$. If the latter is positive, it means that the ``activity'' of the field dynamics~\eqref{AMB} favours the time evolution of the field in the direction corresponding to the orientation of $\mcC$. Conversely, if $\whSigma[\mcC]<0$, it means that the activity favours the opposite evolution.
Since we are going to use formula~\eqref{loopWiseEntProd02}, we directly start from a parametrized surface $\mcS\in\mbF$ such that $\mcC=\partial\mcS$.
To simplify the computations, we choose $d_1=1$, \ie $\rhoss$ is a 1d profile, centered at $x=0$, and we focus on the left boundary of the droplet, for which $x\in \mbR_{<0}$. Note that the cycle-affinity two-form of AMB~\eqref{cycleAffAMB} being spatially additive (it is a spatial integral), if we consider the symmetric evolution on the other side of the droplet, the resulting entropy product is simply doubled.

We consider a surface $\mcS\subset\mbF$ that is parametrized by the map $R:[0,1]\times[0,2\pi/|w|]\to \mbF$ such that
\begin{equation}
R(\tau,t): x\in\mbR_{<0} \mapsto \rho(\tau,t,x)\equiv \rhoss(x)+\tau A(x) \cos(kx-wt) \ ,
\label{AMBsurfaceParam}
\end{equation}
where $k,w\in\mbR$, with $k\geq 0$ and $w\neq 0$, and $A(x)$ is a spatially-varying amplitude that is for now arbitrary.
The boundary of $\mcS\equiv R([0,1]\times[0,2\pi/|w|])$ is then the loop $\partial\mcS = R(1,[0,2\pi/|w|])$ (see figure~\ref{fig:functionalLoop}).
This oriented loop corresponds to an almost harmonic perturbation of $\rhoss$ which propagates over one wavelength $2\pi/k$ (hence eventually closing the loop) either leftward or rightward, depending on the sign of $w$.

We can then explicitly compute the entropy production along $\partial\mcS$ (see appendix~\ref{App:sec:loopwiseEnt}). Assuming for simplicity that $2\lambda+\kappa'$ is a constant (a simplification already made in~\cite{o2023nonequilibrium}), we get
\begin{equation}
\whSigma[\partial \mcS] = \mathrm{sign}(w)(2\lambda+\kappa') \pi  k \int_{\mbR_{<0}} A^2\partial_x\rhoss \, \rmd x \ .
\label{loopwiseAMB}
\end{equation}
We see that, as long as $A(x)$ does not vanish at the interface, since $\partial_x\rhoss\geq 0$ on the left-hand side of the droplet, the entropy production is positive iff $w$ and $2\lambda+\kappa'$ have the same sign. This means that, for instance if $2\lambda+\kappa'>0$, the activity of dynamics~\eqref{AMB} favours a propagation of the perturbation $\rho-\rhoss$ from the gas to the liquid, and vice versa if $2\lambda+\kappa'<0$, as predicted in~\cite{o2023nonequilibrium}. Another interesting point is, because $\partial_x\rhoss$ vanishes in the liquid and gas bulks, the activity only influences perturbations at the interface. In particular, if $A(x)$ vanishes at the interface, then $\whSigma=0$, which means that the activity does not favour any propagation.
On the contrary, for instance if $A$ is constant at the interface but vanishes deep in the bulks, we get
\begin{equation}
\whSigma[\partial \mcS] = \mathrm{sign}(w) \pi A^2 k(2\lambda+\kappa')(\rho_L - \rho_G) \ .
\end{equation}

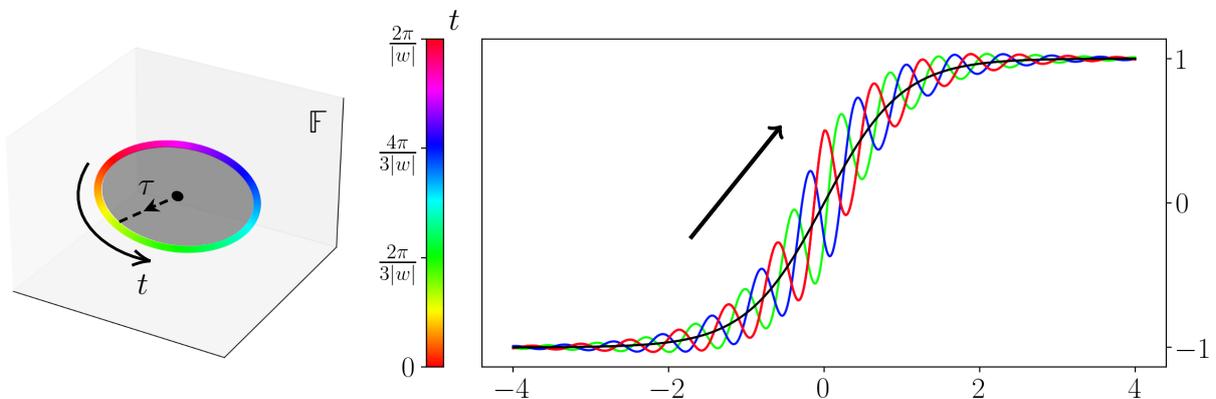
\begin{figure}[h!]
\hspace{-1cm}
\begin{tikzpicture}[scale=1]
\draw(0,0)node{\includegraphics[width=0.35\textwidth]{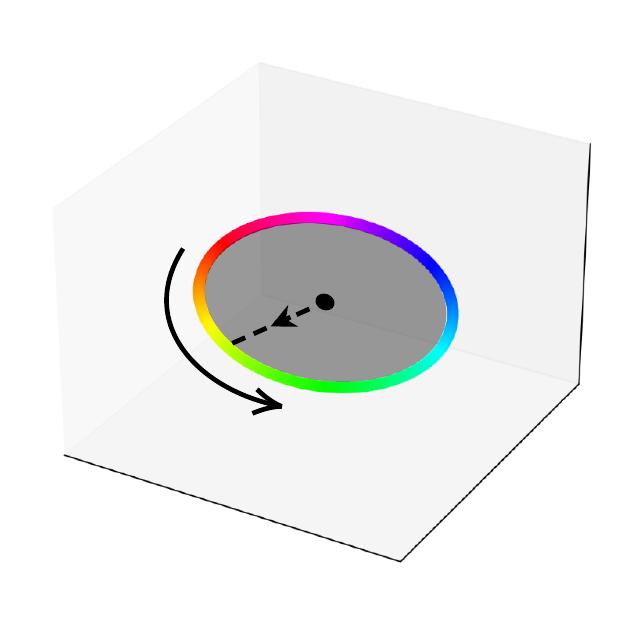}}; 
\draw(8,0)node{\includegraphics[width=0.75\textwidth]{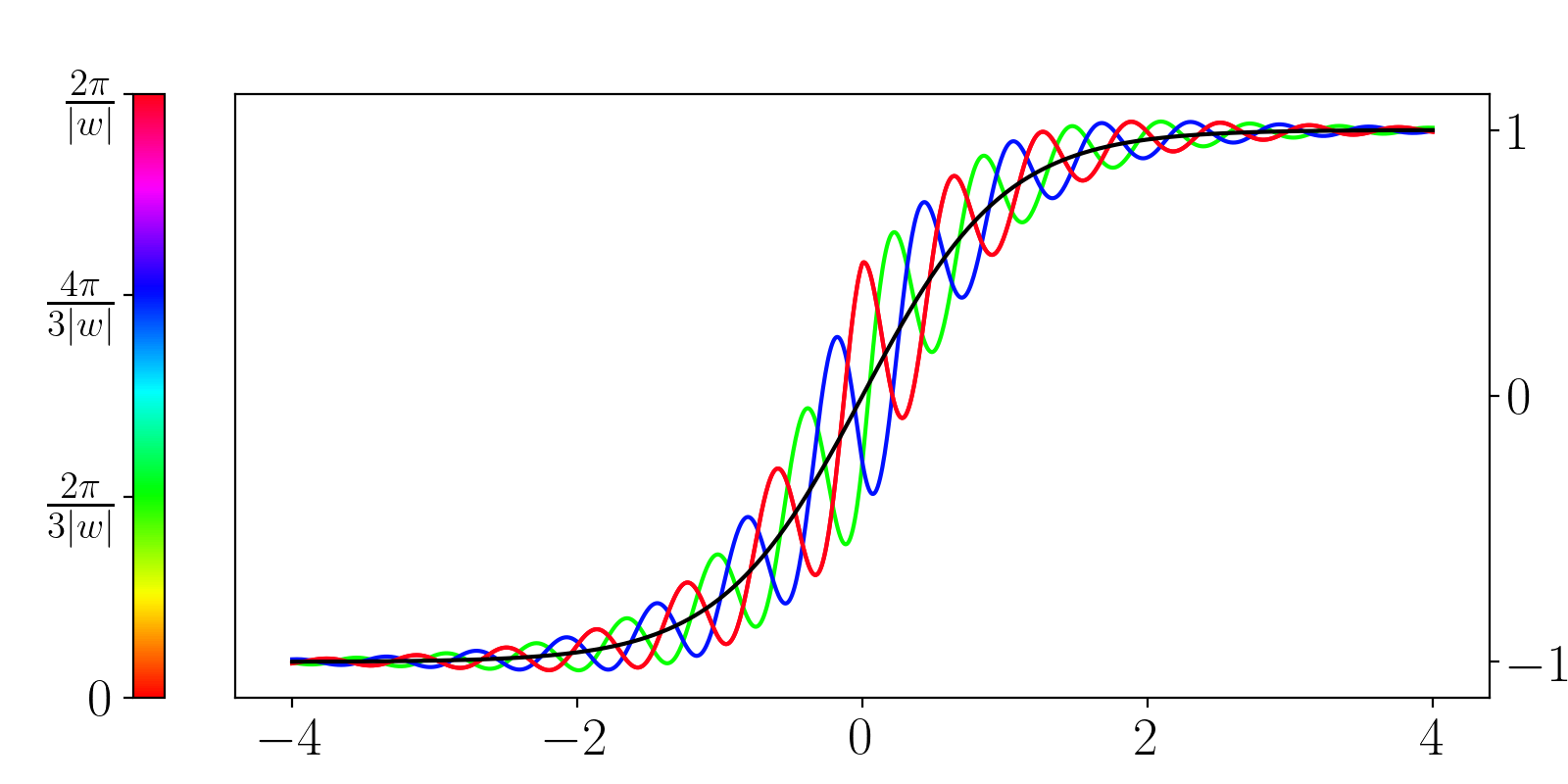}};

\draw(1.9,1.05)node{\large $\mbF$};
\draw(-0.35,0.18)node{\large $\tau$};
\draw(-0.4,-1.1)node{\large $t$};
\draw(3.7,2.4)node{\large $t$};

\draw[->,ultra thick](6.8,-0.5)--(8,1);
\end{tikzpicture}
\caption{\label{fig:functionalLoop} 
A schematic representation of a functional loop. The grey box on the left-hand side represents the functional space $\mbF$. The dark disk in it stands for the surface $\mcS$ and the multicolored loop for its boundary $\partial\mcS$. The color encodes the variable $t\in[0,2\pi]$, with $R(1,t)$ parameterizing $\partial\mcS$, while the whole disk is parameterized by $R(\tau,t)$. On the right-hand side, the black curve is a schematic of the left liquid-gas boundary of the average stationary profile $\rhoss$ (the latter being represented by the black dot at the center of the disk on the left-hand side). The colored curves on top of it  correspond to $R(1,t)$ for $|w|t=0,2\pi/3,4\pi/3,2\pi$, with $R(1,0)=R(1,2\pi/|w|)$. The black arrow above the graphs represents the propagation of the wave as $t$ increases. Note that, as in Fig.~\ref{fig:steadyState_and_derivatives}, we pictorially represent $\rhoss$ by the graph of the function $\tanh(x)$.}
\end{figure}

\subsection{Hidden current and hidden entropy production}
\label{subsec:hiddenIrrev}
Many fields of interest in hydrodynamics and continuum mechanics describe conserved quantities, as \eg mass, momentum or energy. 
When studying the stochastic dynamics of such a field, there exists a notion of entropy production rate that is more discriminating than~\eqref{def_IEPR}. It consists in observing not only the field $\bphi$, but also its associated (real space) current $\bj$, and comparing their forward and backward joint statistics. We denote this \textit{resolved} IEPR by $\sigma_{\bphi, \bj}$, while the one defined in~\eqref{def_IEPR} is from now on denoted by $\sigma_{\bphi}$ and called the \textit{bare} IEPR. 
It follows from the properties of the Kullback-Leibler divergence that we always have $\sigma_{\bphi,\bj} \geq \sigma_{\bphi}$. 
The difference between the two, that we call the \textit{hidden entropy production rate}, is generated by a divergence-free component of the current that we call the \textit{hidden current}. The choice of the names ``hidden current'' and ``hidden IEPR'' is motivated by the fact that they cannot be inferred from observations solely of $\bphi$ and its variations, \ie they are hidden from the space $\mbF$ to which the field $\bphi$ belongs~\cite{li2021steady}.

As can be directly deduced from a geometrical argument, these hidden current and IEPR are also invisible from the vorticity two-form, hence displaying a possible blind spot of the latter in the study of the TRS-status of dynamics~\eqref{EDPS02}. Indeed, the space of current at a given $\bphi$ can be seen as a resolved version of the tangent space $T_{\bphi}\mbF$ that we project down on the latter when taking the divergence, erasing the invisible current in doing so. But $\bomega$ being a two-form over $\mbF$, it can only access tangential data and hence will necessarily be blind to the hidden current and, consequently, to the hidden IEPR as well. 

We start by presenting our general results in section~\ref{subsubsec:hidden_general_case} before illustrating them on the example of a conserved scalar field in section~\ref{subsubsec:ex_conserved_scalar}. Finally, in section~\ref{subsubsec:micro_vs_hidden_IEPR}, we discuss the relation between the resolved and microscopic IEPR.

\subsubsection{General case}
\label{subsubsec:hidden_general_case}

We expect the existence of the hidden current and IEPR to be allowed whenever the operator $\bb_{\bphi}$ of dynamics~\eqref{EDPS02} has a non-trivial kernel. However, we here restrict our analysis to the case where this operator can be decomposed as
\begin{equation}
\bb_{\bphi}=\bB\bbb_{\bphi} \ ,
\label{b_decomp}
\end{equation}
where $\bB$ and $\bbb_{\bphi}$ are linear operators, the latter being such that $\bbb_{\bphi}\bbb_{\bphi}^\dagger$ is invertible and the former -- which is typically minus a divergence operator or several divergence operators stacked one on top of the other -- being independent of $\bphi$ and having a non-empty kernel.
The operator $\bbb$ turns a $\mbR^{d_3}$-valued field over $\mbR^{d_1}$ into a $\mbR^{d_4}$-valued one over $\mbR^{d_1}$ and $\bB$ a $\mbR^{d_4}$-valued field into a $\mbR^{d_2}$-valued field, both also over $\mbR^{d_1}$. We denote by $\mbJ_{\bphi}$ the image of $\bbb_{\bphi}$.

As the drift $\ba_{\bphi}$ of dynamics~\eqref{EDPS02} belongs to the image of $\bb$ (this is necessary for $\bD^{-1}\ba$ to be well defined), for all $\bphi\in\mbF$ there exists $\ant_{\bphi}$, a $\mbR^{d_3}$-valued field over $\mbR^{d_1}$,  such that $\ba_{\bphi}=\bb_{\bphi}\ant_{\bphi}$.
We now define the \textit{deterministic current}: $\bJ_{\bphi}\equiv \bbb_{\bphi} \ant_{\bphi}\in\mbJ_{\bphi}$, as well as the $\bphi$-dependent fields
\begin{eqnarray}
\mks^{i_1 \br_1} &\equiv & \bbb^{i_1\br_1}_{i_3\br_3}\frac{\delta}{\delta \phi^{i_2\br_2}}  b^{i_2\br_2}_{i_3\br_3}, \label{spurious_current}\\
\mkh_{\blambda}^{i_1\br_1} &\equiv & \bbb^{i_1\br_1}_{i_3\br_3} b^{i_2\br_2}_{i_3\br_3}\frac{\delta}{\delta \phi^{i_2\br_2}} \lambda \ ,
\end{eqnarray}
where $\blambda$ is the reference measure.
 It follows from these definitions that $\ba=\bB\bJ$, $\bh_{\blambda}=\bB \mkh_{\blambda}$, and $\bs_{(1/2)}=\bB\mks$ so that Stratonovitch prescription of dynamics~\eqref{EDPS02} can be rewritten as
\begin{eqnarray}
\partial_t\bphi^{\br} &=& \bB^{\br}_{\br'} \bj^{\br'}[\bphi] \label{EDPS_current01}\\
\bj^{\br}[\bphi] &=& \bJ^{\br}[\bphi] + \mks^{\br}[\bphi] +\mkh^{\br}_{\blambda}[\bphi] + \bbb^{\br}_{\br'}[\bphi]\cdot \boldeta^{\br'} \label{EDPS_current02} \ ,
\end{eqnarray}
where $\bj$ is the \textit{Stratonovitch random current}.
The resolved entropy production rate is now defined as:
\begin{equation}
\sigma_{\bphi,\bj} \equiv \lim_{\mcT\to\infty} \frac{1}{\mcT}\ln \frac{\mcP[(\bphi_t,\bj_t)_{t\in\mbT}]}{\mcP[(\bphi_{\mcT-t},-\bj_{\mcT-t})_{t\in\mbT}]} \ . 
\end{equation}
We then show this resolved IEPR, that takes into account the statistics of both the field and the Stratonovitch current, to be given by (see appendix~\ref{appsubsec:jointIEPR}):
\begin{equation}
\sigma_{\bphi,\bj}= \llangle \int  \bj \left[\bbb \bbb^\dagger\right]^{-1} \bJ \rmd \br \rrangle \ .
\end{equation}

On the other hand, we already showed (see eq.~\eqref{fieldIEPR01}) the bare IEPR to read
\begin{equation}
\sigma_{\bphi} \equiv \lim_{\mcT\to\infty} \frac{1}{\mcT}\ln \frac{\mcP[(\bphi_t)_{t\in\mbT}]}{\mcP[(\bphi_{\mcT-t})_{t\in\mbT}]} = \llangle \int \dot{\bphi} [\bb\bb^\dagger]^{-1}\ba \rrangle \ ,
\end{equation}
where the integral runs over $\mbR^{d_1}$.

To connect the resolved and bare IEPRs, we now assume that the space of currents (at a given $\bphi$) admits the following decomposition (see appendix~\ref{appsubsec:hiddenCurrent}):
\begin{equation}
\mbJ_{\bphi} = \IM ( \bbb_{\bphi}\bbb^\dagger_{\bphi} \bB^\dagger ) \oplus \ker(\bB) \ .
\label{target_split}
\end{equation}
It follows that the deterministic current splits as
\begin{equation}
\bJ_{\bphi}=-\bbb_{\bphi}\bbb_{\bphi}^\dagger \bB^\dagger \bmu + \bJ_{\bH} \ ,
\end{equation}
where $\bJ_{\bH}$ is the \textit{hidden current}.
Using eqs.~\eqref{b_decomp} and~\eqref{EDPS_current01} together with $\ba=\bB\bJ = -\bB\bbb\bbb^\dagger\bB^\dagger\bmu$ the bare IEPR reads:
\begin{eqnarray}
\sigma_{\bphi} = -\llangle \int \dot{\bphi} \bmu \rrangle =  \llangle \int\bj \left[\bbb\bbb^\dagger\right]^{-1} \left[-\bbb\bbb^\dagger \bB^\dagger \bmu\right] \rrangle \ .
\end{eqnarray}
We can finally conclude that the resolved IEPR linearly splits as
\begin{equation}
\sigma_{\bphi, \bj} = \sigma_{\bphi}+\sigma_{\bH} \ ,
\end{equation}
where the \textit{hidden entropy production rate} $\sigma_{\bH}$ is given by
\begin{equation}
\sigma_{\bH}\equiv \llangle \int \bj  \left[\bbb\bbb^\dagger\right]^{-1} \bJ_{\bH} \rrangle \ .
\end{equation}
Since the Kullback-Leibler divergence is non-decreasing upon marginalization, the hidden IEPR is always non-negative, $\sigma_{\bH}\geq 0$.
In some situations, the hidden current $\bJ_{\bH}$ and the corresponding hidden IEPR $\sigma_{\bH}$ can vanish, as is the case \eg for the AMB dynamics~\eqref{AMB}-\eqref{AMB_chemical_potential}, so that $\sigma_{\bphi,\, \bj}=\sigma_{\bphi}$.
Another possibility is that $\sigma_{\bphi}=0$ while $\sigma_{\bphi,\bj}\neq 0$, in which case dynamics~\eqref{EDPS_current01}-\eqref{EDPS_current02} appears time-reversible from the statistical properties of the sole field $\bphi$, while it is not as soon as the current $\bj$ is also taken into account~\cite{li2021steady}.

\subsubsection{The example of a conserved scalar field}
\label{subsubsec:ex_conserved_scalar}
To illustrate the notions of hidden current and IEPR on a more concrete example, let us consider a conserved scalar field $\rho(\br,t)$ whose dynamics is given by:
\begin{eqnarray}
\partial_t \rho (\br,t) &=& -\nabla\cdot \bj(\br,t)  \ , \label{Ex_conserved_rho} \\
\bj(\br,t)&=&\bJ(\br,[\rho_t]) +\mks(\br,[\rho_t])+ \bM^{1/2}(\br,[\rho_t])\cdot \boldeta(\br,t) \ , \label{Ex_conserved_j}
\end{eqnarray}
where $\boldeta$ is a noise with the same statics as in eq.~\eqref{EDPS01}, and $\bM$ a (positive definite) mobility matrix.
In the notations of section~\ref{subsubsec:hidden_general_case}, $\bB=-\nabla\cdot$, $\bbb=\bM^{1/2}$, and $\mks$ is given by eq.~\eqref{spurious_current}.
On the other hand, in the notations of eq.~\eqref{EDPS02}, the drift is $\ba=-\nabla\cdot \bJ$, the operator $\bb=\nabla\cdot \bM^{1/2}$, and the diffusion operator $\bD\equiv\bb\bb^\dagger$ reads $\bD = -\nabla\cdot \bM\nabla$.
Finally, we choose the ``functional Lebesgue measure'' as the gauge $\blambda$ so that, as detailed in section~\ref{sec:framework}, $h_{\blambda}=\mkh_{\blambda}=0$.

The bare IEPR of dynamics~\eqref{Ex_conserved_rho}-\eqref{Ex_conserved_j} reads
\begin{equation}
\sigma_{\rho}\equiv \lim_{\mcT\to\infty}\frac{1}{\mcT}\ln \frac{P\left[ (\rho_t)_{t\in[0,\mcT]} \right]}{P\left[ (\rho_{\mcT-t})_{t\in[0,\mcT]} \right]} = \llangle \dot{\rho}^{\br}[D^{-1}a]_{\br} \rrangle  = \llangle \int \rmd \br \,\dot{\rho} [\nabla\cdot \bM \nabla]^{-1} \nabla\cdot \bJ \rrangle \ .
\label{IEPR_rhoSpace}
\end{equation}
while the resolved IEPR is
\begin{equation}
\sigma_{\rho,\, \bj}\equiv \lim_{\mcT\to\infty}\frac{1}{\mcT}\ln \frac{P\left[ (\rho_t,\bj_t)_{t\in[0,\mcT]} \right]}{P\left[ (\rho_{\mcT-t},-\bj_{\mcT-t})_{t\in[0,\mcT]} \right]} = \llangle \int \rmd \br \, \bj \cdot \bM^{-1} \bJ \rrangle  \ . \label{Ex_rhoJ_IEPR}
\end{equation}

In this particular case, the decomposition~\eqref{target_split} assumed in the general case, can be proven, as shown in appendix~\ref{appsubsec:hiddenCurrent}.
It leads to the following decomposition of the Stratonovitch deterministic current:
\begin{equation}
\bJ(\br,[\rho]) = -\bM(\br,[\rho])\nabla\mu(\br,[\rho]) + \bJ_{\bH}(\br,[\rho]) \ , \label{Ex_decomp_current}
\end{equation}
where, for a given $\rho$, $\bJ_{\bH}$ is a divergence free vector-field over $\mbR^{d_1}$. This is the hidden current, which cannot be estimated from the time-varying density field $\rho(\br,t)$ alone, hence being hidden from the point of view of the sole space $\mbF$. 
Using eqs.~\eqref{IEPR_rhoSpace}, \eqref{Ex_decomp_current}, and~\eqref{Ex_conserved_rho} and integrating by parts gives
\begin{equation}
\sigma_\rho = \llangle \int \rmd \br \, \bj \cdot \bM^{-1} \left[ -\bM\nabla\mu\right] \rrangle \ . \label{Ex_rho_IEPR}
\end{equation}
Comparing eqs.~\eqref{Ex_rhoJ_IEPR} and~\eqref{Ex_rho_IEPR} while taking into account the decomposition of the current~\eqref{Ex_decomp_current}, we conclude that the resolved IEPR can indeed be decomposed as the superposition of the bare and the hidden IEPRs, $\sigma_{\rho, \,\bj} = \sigma_\rho + \sigma_{\bH}$, where the latter reads
\begin{equation}
\sigma_{\bH} \equiv \llangle \int \rmd \br \, \bj \cdot \bM^{-1} \bJ_{\bH} \rrangle \ .
\end{equation}
Note that, while the hidden current is always divergence free, it depends on the matrix $\bM$: if the latter is varied with a fixed $\bJ$, then $\bJ_{\bH}$ changes as well.

\subsubsection{Relation to microscopic entropy production}
\label{subsubsec:micro_vs_hidden_IEPR}
Since we have been considering two distinct notions of IEPR so far, one could wonder what is the relation between them and the ``microscopic'' entropy production, in the case where the field dynamics~\eqref{EDPS_current01}-\eqref{EDPS_current02} is explicitly derived from a microscopic one.
To answer this question, we consider an assembly of $N$ interacting particles whose stochastic dynamics reads
\begin{equation}
\dot{\bx}_i = \tba(\bx_i,[\hrho])+\tbh_{\tblambda}+\tbs_{(\varepsilon)}(\bx_i,[\hrho]) + \tbb_\alpha(\bx_i,[\hrho])\teta_i^\alpha \ ,
\label{EDS_collective}
\end{equation}
where $\bx_i\in\mbR^{d_1}$ represents the state of the $i^{th}$ particle, $\hrho(\bx)\equiv \sum_{i=1}^N \delta(\bx-\bx_i)$ is the empirical measure on the state space $\mbR^{d_1}$, the $\tilde{\boldeta}_i$'s are independent Gaussian white noises of zero mean and variance $\langle\teta_i^\alpha(t)\teta_j^\beta(t')\rangle=2\delta_{ij}\delta^{\alpha\beta}\delta(t-t')$, $\tba$ and the $\tbb_\alpha$'s are vector fields
 over $\mbR^{d_1}$, and $\tbh_{\tblambda}$ and $\tbs_{(\varepsilon)}$ the corresponding $\tblambda$-gauge drift and spurious drift. 
(A tilde has been added to every microscopic quantity to avoid any confusion with their field theoretic counterparts.) Note that a function of $N+1$ $\mbR^{d_1}$-variables, $f(\bx_i,\bx_1,\dots,\bx_N)$, can be re-written as $\bar{f}(\bx_i,[\hrho])$ iff it is invariant under any permutation of its last $N$ variables. Physically, this means that the particles in dynamics~\eqref{EDS_collective} are identical and that each of them interact with all the others -- itself included -- independently of their label $i=1,\dots, N$.
The microscopic IEPR of dynamics~\eqref{EDS_collective} reads~\cite{o2024geometric}:
\begin{equation}
\tilde{\sigma}=\llangle \sum_{i=1}^N \dot{\bx}_i\cdot \tbD^{-1}\tba(\bx_i,[\hrho]) \rrangle  \ ,
\label{micro_IEPR}
\end{equation}
where $\tbD\equiv \tbb\tbb^\dagger$ is the microscopic diffusion tensor.

Following standard methods~\cite{dean1996langevin,solon2015active}, dynamics~\eqref{EDS_collective} can be shown to give rise to the stochastic field dynamics~\eqref{Ex_conserved_rho}-\eqref{Ex_conserved_j}, where the mobility reads $\bM=\tbD \hrho$
while the Stratonovitch deterministic current is given by $\bJ=\tba \hrho - \tbD \nabla\hrho$.
Injecting these last two equalities into the expression of the resolved IEPR~\eqref{Ex_rhoJ_IEPR} gives
\begin{equation}
\sigma_{\hrho,\hat{\bj}}= \llangle\int \hat{\bj} \cdot \left[ \tbD^{-1}\tba-\nabla\ln \hrho\right] \rrangle \ .
\label{emp_entropy_interm}
\end{equation}
Using $\partial_t\hrho = -\nabla\cdot \hbj$, the second term on the right-hand side of this last equation reads:
\begin{equation}
-\llangle \int \hbj \cdot \nabla\ln \hrho \rrangle =-\frac{d}{d t} \llangle \int \hrho\ln\hrho \rrangle \ ,
\end{equation}
and hence vanishes.
Finally, as the Stratonovitch empirical current is $\hat{\bj}\equiv \sum_{i=1}^N \dot{\bx}_i \delta(\bx-\bx_i)$, we have the equality
\begin{equation}
\sigma_{\hrho,\hat{\bj}}=\tilde{\sigma} \ .
\end{equation}
Thus, for the case considered here, the resolved IEPR~\eqref{emp_entropy_interm} of the dynamics of the state-space empirical measure coincides with the IEPR~\eqref{micro_IEPR} of the corresponding microscopic dynamics.

\newpage
\section{Beyond T-reversal: EMT-reversal}
\label{sec:MTreversal}
In sections~\ref{sec:quantifyingTRSbreaking}-\ref{sec:entropy}, we only focused on T-reversal, \ie on a time reversal that solely amounts to reversing the time variable. 
But as explained in the introduction, time reversal of the dynamics of certain physical systems also requires to flip some degrees of freedom (MT-reversal), or even directly transform the force-field and the diffusion operator (EMT-reversal).
In~\cite{o2024geometric}, we gave explicit conditions for a finite-dimensional dynamics to be EMT-reversible. In this section, we show how the geometric notations introduced in section~\eqref{sec:framework} allow us to generalize these conditions to stochastic field dynamics of the form~\eqref{EDPS02}.

To this end, let us first introduce a ``mirror map'' $\mkm[\bphi]$ that flips some degrees of freedom. It must be an involution, \ie such that $\mkm[\mkm[\bphi]]=\bphi$. We focus here on maps $\mkm$ that leave the functional Lebesgue measure $\blambda$ invariant. 
When that is not the case, the reversibility conditions given below can be extended and we let the curious reader adapt for themselves corresponding EMT-reversibility conditions given in~\cite{o2024geometric} for finite-dimensional dynamics to the current field theoretic setting.
We also introduce an ``extension map'' $\mke[\ba,\bD]\equiv(\mke_{\ba}[\ba],\mke_{\bD}[\bD])$ -- which is also an involution -- that directly transforms the drift and diffusion operator.

Denoting by $\mcP_{(\ba,\bD)}$ the stationary path-probability of dynamics~\eqref{EDPS02}, we can now define the \textit{$\mke\mkm\mkt$-reversed process} of~\eqref{EDPS02}: it is the stochastic partial differential equation whose path-probability measure in steady state, $\bar{\mcP}_{(\ba,\bD)}$, is given by
\begin{equation}
\bar{\mcP}_{(\ba,\bD)}\left[(\bphi_t)_{t\in\mbT}\right] \equiv \mcP_{\mke[\ba,\bD]}\left[(\bphi_{\mcT-t})_{t\in\mbT}\right] \ .
\end{equation}
Dynamics~\eqref{EDPS02} is then said to be \textit{$\mke\mkm\mkt$-reversible} whenever $\bar{\mcP}_{(\ba,\bD)}=\mcP_{(\ba,\bD)}$.

We now introduce notations that will be useful below.
As in finite dimension, the mirror map $\mkm$ acts by ``pushforward'' on the drift $\ba$ and diffusion operator $\bD$ respectively as
\begin{equation}
(\mkm_*\ba)^{i\br}[\bphi] \equiv \left(\frac{\delta \mkm^{i\br}}{\delta \phi^{j\br'}}a^{j\br'}\right)[\mkm(\bphi)]
\end{equation}
and 
\begin{equation}
(\mkm_*\bD)^{i_1\br_1i_2\br_2}[\bphi] \equiv \left(\frac{\delta \mkm^{i_1\br_1}}{\delta \phi^{i_3\br_3}}\frac{\delta \mkm^{i_2\br_2}}{\delta \phi^{i_4\br_4}}D^{i_3\br_3i_4\br_4}\right)[\mkm(\bphi)] \ .
\end{equation}
We can then define the symmetric and antisymmetric parts of $\ba$ and $\bD$ under the map $\mkm_*\mke$: 
\begin{equation}
\ba^{S/A}\equiv \frac{1}{2}\left(\ba\pm \mkm_*\mke_{\ba}[\ba]\right)
\end{equation}
and 
\begin{equation}
\bD^{S/A}\equiv \frac{1}{2}\left(\bD\pm \mkm_*\mke_{\bD}[\bD]\right) \ .
\end{equation}
In the next section, we give three sets of necessary and sufficient conditions for dynamics~\eqref{EDPS02} to be $\mke\mkm\mkt$-reversible\footnote{As a reminder, the map $\mkt$ simply transforms a trajectory as $\mkm :(\bphi_t)_{t\in\mbT}\mapsto(\bphi_{\mcT-t})_{t\in\mbT}$.}. 
These sets of conditions become progressively less general (for instance the last set of conditions requires $\bD$ to be invertible while the first does not), but easier to verify in practice.

\subsection{EMT-reversibility conditions}
Let us thus start by translating to field-level the most general of the sets of reversibility conditions found for finite dimensions in~\cite{o2024geometric}.
Dynamics~\eqref{EDPS02} is $\mke\mkm\mkt$-reversible iff there exist a positive functional $\tPss:\mbF\to\mbR$ such that 
\begin{gather}
\bD^A = 0 \ , \label{EMT_cond_D_01} \\
\ba^S = \bD  \mbd \ln \tPss \ , \label{EMT_cond_a_01} \\
\DDiv[\tPss(\ba  - \bD \mbd \ln \tPss)] = 0 \ , \label{EMT_cond_div_01}
\end{gather}
where $\bD \mbd \ln \tPss$ is the vector field over $\mbF$ with coordinates $D^{i\br j\br'} \frac{\delta \ln \tPss}{\delta \phi^{j\br'}}$ and $\DDiv$ stands for the functional divergence operator defined in~\eqref{def_div}.
When conditions~\eqref{EMT_cond_D_01}-\eqref{EMT_cond_div_01} are satisfied, then the functional $\tPss$ is a\footnote{
In section~\ref{sec:framework}, we assumed existence and uniqueness of the stationary probability measure of dynamics~\eqref{EDPS02}. All the reversibility conditions given in this section turn out to be valid even under the weaker assumption that the kernel of the Fokker-Planck operator, $P\mapsto -\DDiv[\ba P - \bD \mbd P ]$, is one-dimensional. Thus, it also include cases where all the stationary measures are proportional but none of them is normalized. Note that conditions~\eqref{EMT_cond_D_01}-\eqref{EMT_cond_div_01} are independent of this (positive) proportionality constant between the various stationary measures.} stationary density (with respect to $\blambda\equiv \mcD \bphi$) of dynamics~\eqref{EDPS02}.

Let us now consider the case where the space $\mbF$ can be decomposed as $\mbF=\mbF_1\times\mbF_2$ in such a way that the diffusion operator reads
\begin{equation}
\bD = \begin{bmatrix}
0 & 0 \\
0 & \bD_2 
\end{bmatrix} \ ,
\end{equation}
where $\bD_2$ is invertible on $\mbF_2$.
This can apply for instance when the components of $\bphi=(\bphi_1,\bphi_2)$ correspond to mass density and velocity, and the mass is conserved.
In this case, the conditions~\eqref{EMT_cond_D_01}-\eqref{EMT_cond_div_01} for dynamics~\eqref{EDPS02} to be $\mke\mkm\mkt$-reversible can be reformulated as follows: there exists a positive functional $\tPss$on $\mbF$ such that
\begin{gather}
\bD^A_2 = 0 \ , \label{EMT_cond_D_02} \\
\ba^S_1 = 0 \ , \label{EMT_cond_a_02} \\
\mbd_2\bD_2^{-1}\ba^S_2 = 0 \ , \label{EMT_cond_a_02bis} \\
\ba_1\cdot \mbd_1 \ln \tPss + \DDiv_1 (\ba_1) + \DDiv_2 (\ba_2^A) + \ba^A_2\cdot \bD_2^{-1}\ba_2^S  = 0 \ , \label{EMT_cond_div_02}
\end{gather}
where $\ba_i$ is the component of $\ba$ along the factor $\mbF_i$, $\mbd_i$ and $\DDiv_i$ respectively the exterior derivative and the divergence operator on $\mbF_i$.
In eq.~\eqref{EMT_cond_div_02}, `$\cdot$' stands for the contraction between functional contravariant and covariant tensors over $\mbF$, \ie $\bu\cdot\bv \equiv u_{i\br} v^{i\br}$, where $\bu$ and $\bv$ are covariant and contravariant tensors, respectively\footnote{We did not write the `$\cdot$' explicitly between $\bD_2^{-1}$ and $\ba^S_2$ to lighten notations, having also in mind that this abuse of notation draws an analogy with matrix multiplication.}. 

Finally, if we consider the case where $\bD$ is invertible on the whole $\mbF$, which is equivalent to the first factor $\mbF_1$ above being  reduced to zero, $\mbF_1=\{0\}$, then the $\mke\mkm\mkt$-reversibility conditions~\eqref{EMT_cond_D_02}-\eqref{EMT_cond_div_02} boil down to
\begin{gather}
\bD^A = 0 \ , \label{EMT_cond_D_03} \\
\mbd\bD^{-1}\ba^S = 0 \ , \label{EMT_cond_a_03} \\
\DDiv (\ba^A) + \ba^A\cdot \bD^{-1}\ba^S  = 0 \ . \label{EMT_cond_div_03}
\end{gather}
As in finite dimension, by analogy between the probability and thermodynamical currents~\cite{graham1971fluctuations}, condition~\eqref{EMT_cond_D_03}, which can be reformulated as $\mkm_*\mke[\bD]=\bD$, can be seen as an \textit{Onsager-Casimir symmetry}.
Then, upon adapting the definition~\eqref{vorticity_definition} of the ``$\mkt$-vorticity'' two-form to define the ``$\mke\mkm\mkt$-vorticity'' two-form as 
\begin{equation}
\bomega\equiv \mbd \bD^{-1}\ba^S \ ,
\label{vorticity_2form_EMTrev}
\end{equation}
condition~\eqref{EMT_cond_a_03} requires the cancellation of $\bomega$, similarly to the T-reversal case. Note that, although they both belong to $\Omega^2(\mbF)$, these $\mkt$- and $\mke\mkm\mkt$-vorticities are distinct in general.
Finally, in the case of condition~\eqref{EMT_cond_div_03}, a satisfying physical interpretation remains elusive, as is also the case for its finite-dimensional counterpart~\cite{o2024geometric}.

As compared to the two others, the third set of reversibility conditions~\eqref{EMT_cond_D_03}-\eqref{EMT_cond_div_03} requires the strongest assumption: the invertibility of $\bD$ on the whole $\mbF$. But, as in finite dimension, it remarkably does not refer to any unknown functional $\tPss$, as opposed to conditions~\eqref{EMT_cond_D_01}-\eqref{EMT_cond_div_01} or~\eqref{EMT_cond_D_02}-\eqref{EMT_cond_div_02}.	On the other hand, contrary to the finite-dimensional setting, it does not imply that conditions~\eqref{EMT_cond_D_03}-\eqref{EMT_cond_div_03} can always be easily checked in practice, as they require either the explicit knowledge of the inverse of the operator $\bD$, or the symmetric drift to be decomposed as $\ba^S=\bD\balpha^S$ with $\balpha^S$ explicitly known. When it is not the case, $\bD^{-1}\ba^S$ is not explicitly known and, in particular, a major part of the agenda established in Part I is consequently blocked.
Finally, an additional difficulty of these three sets of reversibility conditions, inherent to our functional setting, is that they involve the functional divergence operator $\DDiv$ that often produces singularities.

\subsection{New sources of entropy production and topology}
In this section, we focus on the case where $\bD$ is invertible on $\mbF$.
For given maps $\mke$ \& $\mkm$, the (bare) $\mke\mkm\mkt$-IEPR, \ie the IEPR that quantifies the breakdown or otherwise of $\mke\mkm\mkt$-reversal of dynamics~\eqref{EDPS02}, is defined as
\begin{equation}
\sigma \equiv \lim_{\mcT\to\infty} \frac{1}{\mcT}\ln \frac{\mcP_{(\ba,\bD)}\left[(\bphi_t)_{t\in\mbT}\right]}{\bar{\mcP}_{(\ba,\bD)}\left[(\bphi_t)_{t\in\mbT}\right]} \ .
\end{equation}
When the Onsager-Casimir condition~\eqref{EMT_cond_D_03} is satisfied, directly generalizing a result of~\cite{o2024geometric} (and eq.~\eqref{fieldIEPR02}), we conclude that this ($\mke\mkm\mkt$-) IEPR reads:
\begin{equation}
\sigma = \frac{1}{2}\int_{\mbF} \bC\cdot \bomega  \;\mcD \bphi - \int_{\mbF}  \left[\DDiv (\ba^A) + \ba^A\cdot \bD^{-1}\ba^S \right] \Pss \mcD\bphi \ ,
\label{IEPR_EMT-reversal}
\end{equation}
where $\bomega$ is the $\mke\mkm\mkt$-vorticity two-form~\eqref{vorticity_2form_EMTrev}, $\Pss\mcD\bphi$ the stationary measure, and $\bC\cdot \bomega \equiv C^{i\br j\br'}\omega_{i\br j\br'} = C^{i\br j\br'} [\mbd \bD^{-1}\ba^S]_{i\br j\br'}$, with $\bC$ a functional, contravariant, antisymmetric tensor of order two whose divergence equals the negative of the stationary probability current, $\DDiv(\bC)=-\bJss$.
Because EMT-reversibility is both equivalent to conditions~\eqref{EMT_cond_D_03}-\eqref{EMT_cond_div_03} and to $\sigma=0$, we conclude that dynamics~\eqref{EDPS02} is reversible iff the two integrals in~\eqref{IEPR_EMT-reversal} vanish independently. 
Consequently, the left-hand side of eq.~\eqref{EMT_cond_div_03} and that of eq.~\eqref{EMT_cond_a_03} (the $\mke\mkm\mkt$-vorticity two-form) can be seen as \textit{two independent and additive sources of entropy production}. Further, these sources are different in their mathematical nature are they are respectively a scalar field and a two-form over $\mbF$. Note that this differs from the simplest case of T-reversal, where there can be an entropy source of only one type: the (T-) vorticity two-form~\eqref{vorticity_definition} -- the latter being generically distinct from the $\mke\mkm\mkt$-vorticity two-form~\eqref{vorticity_2form_EMTrev} for non trivial mirror ($\mkm$) and extension ($\mke$) maps. 
It is also worth noting that going from a pair of maps $(\mkm, \mke)$ to another one changes the decomposition $\ba=\ba^S+\ba^A$ by transferring terms between symmetric and antisymmetric components, and hence modifies the IEPR by altering the two integrals in eq.~\eqref{IEPR_EMT-reversal}.
%

Finally note that condition~\eqref{EMT_cond_a_03} is equivalent to the functional one-form being closed and therefore exact since we have been assuming the space $\mbF$ to be simply connected throughout this article. When this topological assumption does not hold,
condition~\eqref{EMT_cond_a_03} (and similarly condition~\eqref{EMT_cond_a_02bis}) must be augmented with topological constraints, as in finite dimension~\cite{o2024geometric}. More precisely, when $\mbF$ is not simply connected, condition~\eqref{EMT_cond_a_03} should be replaced by:
\begin{equation}
\mbd\bD^{-1}\ba^S = 0 \quad \text{and,} \ \forall i=1,\dots,\dim( H_1(\mbF)), \ \int_{\mcC_i} \bD^{-1}\ba^S = 0 \ ,
\label{EMT_cond_a_04}
\end{equation}
where $\{\mcC_i\}_i$ is a set of loops in $\mbF$, generating its first homology group $H_1(\mbF)$, along which we integrate the functional one-form $\bD^{-1}\ba^S$. Intuitively, the first part of condition~\eqref{EMT_cond_a_04} forbids the one-form $\bD^{-1}\ba^S$ to rotate around any point in $\mbF$, while its second part prevents rotations around any holes of $\mbF$.
When $H_1(\mbF)$ is not trivial, \ie when $\mbF$ is not simply-connected, additional entropy production terms, that accounts for this topological obstruction to EMT-reversibility, must be added to the right-hand side of eq.~\eqref{IEPR_EMT-reversal}. 
Assuming that Hodge-de Rham theory~\cite{warner1983foundations} can be extended to infinite dimensional manifolds (which as far as we know remains unproven), the stationary probability current, which is divergence-free, can be written $\bJss=-\DDiv \bC + \sum_{\alpha=1}^{\dim( H_1(\mbF))}j^\alpha \bgamma_\alpha$, where $(\bgamma_{\alpha})_{\alpha=1,\dots,\dim( H_1(\mbF))}$ is a basis of harmonic vector fields on $\mbF$. It follows (see~\cite{o2024geometric} for details in the finite-dimensional counterpart) that the IEPR reads:
\begin{equation}
\sigma = \frac{1}{2}\int_{\mbF} \bC\cdot \bomega  \;\mcD \bphi + \sum_{\alpha=1}^{\dim( H_1(\mbF))}j^{\alpha}\int\bgamma_{\alpha}\cdot \bD^{-1}\ba^S \mcD\bphi  - \int_{\mbF}  \left[\DDiv (\ba^A) + \ba^A\cdot \bD^{-1}\ba^S \right] \Pss \mcD\bphi \ ,
\label{IEPR_EMT-reversal_topol}
\end{equation}
with $\bgamma_{\alpha}\cdot \bD^{-1}\ba^S\equiv \gamma_\alpha^{i\br} [D^{-1}a^S]_{i\br}$. Once again, and for the same reason as above, these topological terms in the expression~\eqref{IEPR_EMT-reversal_topol} of $\sigma$ can be interpreted as additional independent sources of entropy production.

%

If the target space $\mcN$ of $\bphi$ is a vector space, then $\mbF$ is also a vector space. In this case, $\mbF$ is simply connected, and there cannot be any topological obstruction to the reversibility of~\eqref{EDPS02}. However, as soon as $H_1(\mcN)$ is non-trivial -- which is often the case when the field $\bphi$ is a Goldstone mode associated with a continuous symmetry spontaneously broken, for instance -- one can expect $H_1(\mbF)$ to also be non-trivial. In such a case, if the system is out of equilibrium, there may exist a “force” that tends to continuously drive the field along loops in $H_1(\mcN)$, thereby creating a topological source of irreversibility.

Examples of systems where such a topological obstruction to reversibility can arise are abundant in statistical physics, ranging from continuous versions of the Kuramoto model~\cite{kuramoto1975international,acebron2005kuramoto,ermentrout1980large,hutt2007generalization}, a paradigmatic model of synchronization between coupled oscillators, to chiral active systems~\cite{liebchen2017collective,liebchen2022chiral,markovich2019chiral,workamp2018symmetry,kole2021layered}, as well as pulsatile active matter~\cite{zhang2023pulsating,banerjee2024hydrodynamics,liu2024collective}.

\newpage
\section{Conclusion}
In this article we introduced a functional geometric formalism that allowed us to generalize to stochastic field dynamics various results that were previously only known for finite dimensional systems (notably in sections~\ref{sec:quantifyingTRSbreaking}, \ref{subsec:entropyProd}, and~\ref{sec:MTreversal}).
In particular, for stochastic field theories governed by~\eqref{EDPS02} we gave fully general reversibility conditions for both standard time reversal symmetry (T-reversibility) and its extended counterpart (EMT-reversibility). In each case we showed how these yield simpler conditions under the assumption of partial or complete invertibility of the diffusion operator $\bD$. These results generalise those we presented in Part I~\cite{o2024geometric} for finite dimensional dynamics, and, to the best of our knowledge, were unknown previously for fields even in the simplest case of T-reversibility.


Most of the work presented in this article focused on the latter, simplest notion of time reversal. Building on the insights first presented in~\cite{o2023nonequilibrium}, we showed T-reversibility to be equivalent to the cancellation of the vorticity two-form $\bomega$. 
Contrary to its finite dimensional counterpart, this T-reversibility criterion is distributional (\ie it amounts to the cancellation of a distribution) and can thus be hard to make use of. We fully solved this problem in the local, one-spatial-dimension case by constructing a basis of two-forms that turns the distributional condition into an easily applicable algebraic one.

Further, when the dynamics is irreversible, we showed how one can gain insight into the corresponding nonequilibrium phenomenology through the flow~\eqref{Vortex_dynamics} of the vorticity operator, the latter being dual to the vorticity two-form.
By examining in this way the phenomenology associated to the elements of the aforementioned basis, we found that they can be divided into three subfamilies, with the elements of each subfamily exhibiting closely related, characteristic phenomenologies. We demonstrated that many well-known, typically out-of-equilibrium phenomena, such as flocking and those induced by non-reciprocal interactions, can be classified in this manner, their respective vorticity two-forms each belonging to a specified subspace of the decomposition~\eqref{Decomposition_space_local_vorticities} of the space of two-forms $\Omega^2(\mbF)$.

At the end of section~\ref{sec:basis}, we discussed in detail how this decomposition of $\Omega^2(\mbF)$ extends to more general contexts and take the form~\eqref{general_decomposition_vorticity_space}, the additional factors of the decomposition presumably accounting for \eg the notoriously rich boundary effects of active systems. We also discussed how this decomposition might be used to build field theories with desired properties.

In section~\ref{sec:entropy}, we studied the relation between entropy production and vorticity. 
In particular we showed the loop-wise entropy production and informatic entropy production rate (IEPR) to both be linear functionals of $\bomega$. While the relation with the loop-wise entropy production gives another connection between $\bomega$ and the generic nonequilibrium phenomenology of the dynamics, as demonstrated in section~\ref{subsubsec:loopWiseEP_AMB}, the link we establish to the IEPR allows the elements of the decomposition~\eqref{general_decomposition_vorticity_space} of $\Omega^2(\mbF)$ to be interpreted as independent sources of entropy production. We also showed that hidden currents and a resulting contribution to IEPR can arise, typically when the field $\bphi$ is conserved. They are hidden from the space $\mbF$ in the sense that there presence is only revealed by considering the joint statistics of the field and its corresponding current. 

Finally, in section~\ref{sec:MTreversal}, we gave reversibility conditions for arbitrary EMT-reversal rather than the simple T-reversal considered up to that point. We showed that considering nontrivial mirror ($\mkm$) and extension ($\mke$) maps -- which respectively flip the odd parts of the degrees of freedom, drift, and diffusion operator under time-reversal -- leads to a corresponding $\mke\mkm\mkt$-IEPR that is the superposition of several independent sources of entropy production, including the ($\mke\mkm\mkt$-)vorticity two-form, while the (T-)voticity two-form is the unique source for T-irreversibility. 
Similarly, additional sources arise if the space $\mbF$ to which $\bphi$ belongs is not simply connected.
Studying all these additional irreversibility sources in a way that is similar to what we did for the vorticity space in this article is an exciting challenge for the future that could move towards a more complete classification of out-of-equilibrium field theories 
by precisely characterizing their irreversibility and showing how it links to qualitatively as well as quantitatively observable features.
%

\section*{Acknowledgments}
The authors thank Julien Tailleur for useful discussions and Balázs Németh for a careful reading of the manuscript and for greatly simplifying the argument that appears in Appendix~\ref{app:factorisation_polynomials}.
This work was partially funded by the European Research Council under the Horizon 2020 Programme (ERC Grant Agreement No. 740269) and the Horizon Europe Programme (ERC Synergy Grant \textit{Shapincellfate}).

\appendix

\newpage

\section{An ill--defined spurious drift}
\label{AppIllDefSpurious}
In this appendix, we compute the spurious drift
\begin{equation}
s^{\br_1}_{(\varepsilon)} \equiv \frac{\delta}{\delta \phi^{\br_3}} b^{\br_1}_{i\br_2} b^{\br_3}_{i\br_2} - 2\varepsilon b^{\br_3}_{i\br_2} \frac{\delta}{\delta \phi^{\br_3}} b^{\br_1}_{i\br_2} \ ,
\end{equation}
associated to the operator 
\begin{equation}
b^{\br}_{j\br'} =\sqrt{M(\bphi(\br'))}\partial_{r'_j}\delta(\br-\br') \ .
\end{equation}

The first term of the above expression of $s^{\br_1}_{(\varepsilon)}$ reads
\begin{eqnarray*}
\frac{\delta}{\delta \phi^{\br_3}} b^{\br_1}_{i\br_2} b^{\br_3}_{i\br_2} &=& \sum_i \int \rmd\br_2\rmd\br_3  \left[ \partial_{r_2^i} \delta(\br_1-\br_2) \right]  \left[ \partial_{r_2^i} \delta(\br_3-\br_2) \right]  \frac{\delta}{\delta\phi^{\br_3}} M(\phi^{\br_2}) \\
&=& \sum_i \int \rmd\br_2\rmd\br_3  \left[ \partial_{r_2^i} \delta(\br_1-\br_2) \right]  \left[ -\partial_{r_3^i} \delta(\br_3-\br_2) \right]  \frac{\delta}{\delta\phi^{\br_3}} M(\phi^{\br_2}) \\
&=& \sum_i \int \rmd\br_2 \left[ \partial_{r_2^i} \delta(\br_1-\br_2) \right]    \left.\left[\partial_{r_3^i}\frac{\delta M(\phi^{\br_2})}{\delta\phi^{\br_3}}  \right] \right|_{\br_3=\br_2} \ ,
\end{eqnarray*}
where we have used that $\partial_{r_2^i} \delta(\br_3-\br_2) = -\partial_{r_3^i} \delta(\br_3-\br_2)$ to go from the first to the second line, and did an integration by parts to get to the last line.

Similarly, the second term in the expression of $s_{(\varepsilon)}$ reads
\begin{eqnarray*}
-2\varepsilon b^{\br_3}_{i\br_2} \frac{\delta}{\delta \phi^{\br_3}} b^{\br_1}_{i\br_2} &=& - 2\varepsilon \sum_i \int \rmd\br_2\rmd\br_3  \left[ \partial_{r_2^i} \delta(\br_1-\br_2) \right]  \left[ \partial_{r_2^i} \delta(\br_3-\br_2) \right]  \sqrt{M(\phi^{\br_2})}	\frac{\delta}{\delta\phi^{\br_3}} \sqrt{M(\phi^{\br_2})} \\
&=& -\varepsilon \sum_i \int \rmd\br_2\rmd\br_3  \left[ \partial_{r_2^i} \delta(\br_1-\br_2) \right]  \left[ \partial_{r_2^i} \delta(\br_3-\br_2) \right]  \frac{\delta}{\delta\phi^{\br_3}} M(\phi^{\br_2}) \ ,
\end{eqnarray*}
which is $-\varepsilon$ times the first term.
The total spurious drift is thus
\begin{equation}
s^{\br_1}_{(\varepsilon)}  =  (1-\varepsilon)\sum_i \int \rmd\br_2 \left[ \partial_{r_2^i} \delta(\br_1-\br_2) \right]    \left.\left[\partial_{r_3^i}\frac{\delta M(\phi^{\br_2})}{\delta\phi^{\br_3}}  \right] \right|_{\br_3=\br_2} \ .
\end{equation}
This is ill-defined as described in the main text (eq.~\eqref{spurious_temp}).

\section{Functional Stokes theorem}
\label{App:subsec:Stokes}

Let $S\subset\mbF$ be a smooth, connected, and oriented surface in the function space $\mbF$, whose (smooth oriented) boundary is denoted by $\partial S$. Let $\bzeta$ be a functional one--form over $\mbF$.
The objective of this appendix is to show that Stokes' theorem\footnote{We here only give a proof for one-forms and their functional exterior derivatives, but it can be easily adapted to show that the functional Stokes' theorem actually holds for functional differential forms of any order.} still holds in the functional space $\mbF$, namely that we have
\begin{equation}
\int_{\partial S} \bzeta = \int_S \mbd \bzeta \ . 
\label{App:functionalStokes}
\end{equation}
To prove this, we can show that the functional exterior derivative $\mbd$ is well behaved with respect to the pullback. We will then only need to use the finite--dimensional Stokes' theorem.
\\

We start by assuming that we can parametrize the suface $S$ by another smooth, connected, and oriented surface $s\subset\mbR^2$, through a diffeomorphism $\Phi:\mbR^2\to\mbF$, \ie $\Phi(s)=S$.

On the one hand, the change of variable formula\footnote{Note that formula~\eqref{App:pullbackOneForm} could actually be used as the definition of the integral of a functional one--form $\bzeta$ on a smooth line in $\mbF$.} first gives
\begin{equation}
\int_{\partial S} \bzeta = \int_{\partial s} \Phi^*\bzeta \ ,
\label{App:pullbackOneForm}
\end{equation}
where $\Phi^*$ stands for the pullback with respect to the map $\Phi$ (whose definition is recalled below). 
But $\Phi^*\bzeta$ is then nothing but a usual one--form on $\mbR^2$, to which we can apply the (finite dimensional) Stokes' theorem, namely
\begin{equation}
\int_{\partial s} \Phi^*\bzeta = \int_s \rmd (\Phi^*\bzeta) \ ,
\end{equation}
where $\rmd$ stands for the finite dimensional exterior derivative.

On the other hand, the change of variables formula similarly gives for two--forms:
\begin{equation}
\int_{S} \mbd \bzeta = \int_s \Phi^*(\mbd\bzeta) \ .
\end{equation}
Hence, if we can show that the pullback commutes with the exterior derivatives (the usual one, and the functional one), \ie that
\begin{equation}
\int_s \rmd (\Phi^*\bzeta) = \int_s \Phi^*(\mbd\bzeta) \ ,
\label{App:PullExtDerCommut}
\end{equation}
then we will be able to conclude from eqs.~\eqref{App:pullbackOneForm}--\eqref{App:PullExtDerCommut} that the functional Stokes' theorem~\eqref{App:functionalStokes} does hold.
\\

We thus turn to the proof of eq.~\eqref{App:PullExtDerCommut}.
In order to avoid any confusion with the spatial variable $\br$ of the fields $\phi(\br)$ that belong to $\mbF$, let us denote by $(t,\tau)$ the canonical coordinates of $\mbR^2$, rather than by $(x,y)$.

We first evaluate the pullback by $\Phi$ of the functional exterior derivative $\mbd\bzeta_\phi = \frac{\delta\zeta_{j\br'}}{\delta \phi^{i\br}}[\phi] \delta^{i\br}\wedge \delta^{j\br'}$, of $\bzeta_\phi=\zeta_{i\br}[\phi]\delta^{i\br}$ on a pair of vectors $\bu=u^t\partial_t + u^\tau\partial_\tau$ and $\bv=v^t\partial_t + v^\tau\partial_\tau$ at an arbitrary point $(t,\tau)\in s$:
\begin{eqnarray}
[\Phi^*\mbd\bzeta]_{(t,\tau)}(\bu,\bv) &=& \frac{\delta\zeta_{j\br'}}{\delta \phi^{i\br}}[\Phi(t,\tau)]\delta^{i\br}\wedge\delta^{j\br'} \left( \rmd \Phi(\bu), \rmd \Phi(\bv) \right) \ .
\end{eqnarray}
Using the expression of the differential $\rmd\Phi = \partial_t\Phi \rmd t + \partial_\tau \Phi \rmd\tau$ of $\Phi$ together with the bilinearity of the skew--symmetry map $\delta^{i\br}\wedge\delta^{j\br'}$, we get
\begin{eqnarray}
[\Phi^*\mbd\bzeta]_{(t,\tau)}(\bu,\bv) &=& \frac{\delta\zeta_{j\br'}}{\delta \phi^{i\br}}[\Phi(t,\tau)]\Bigg\{ 
\Big[\partial_t\Phi^{i\br} u^t + \partial_\tau \Phi^{i\br} u^\tau \Big] \left[\partial_t\Phi^{j\br'} v^t + \partial_\tau \Phi^{j\br'} v^\tau \right] \\
& & - \left[\partial_t\Phi^{j\br'} u^t + \partial_\tau \Phi^{j\br'} u^\tau \right]\Big[\partial_t\Phi^{i\br} v^t + \partial_\tau \Phi^{i\br} v^\tau \Big] \Bigg\} \\
&=& \frac{\delta\zeta_{j\br'}}{\delta \phi^{i\br}}[\Phi(t,\tau)] \Big[ (\partial_t\Phi^{i\br}) (\partial_\tau \Phi^{j\br'}) - (\partial_t\Phi^{j\br'}) (\partial_\tau \Phi^{i\br})\Big] \left[ u^tv^\tau - u^\tau v^t \right] \ .
\end{eqnarray}
In terms of two--forms, this also reads
\begin{equation}
[\Phi^*\mbd\bzeta]_{(t,\tau)} = \frac{\delta\zeta_{j\br'}}{\delta \phi^{i\br}}[\Phi(t,\tau)] \Big[ (\partial_t\Phi^{i\br}) (\partial_\tau \Phi^{j\br'}) - (\partial_t\Phi^{j\br'}) (\partial_\tau \Phi^{i\br})\Big] \rmd t \wedge \rmd \tau \ .
\label{App:pullbackFunctExtDer}
\end{equation}

Let us now turn to the evaluation on $(\bu,\bv)$ of the (finite--dimensional) exterior derivative of $\Phi^*\bzeta$.
We first have that the pullback $\Phi^*\bzeta$ applied to $\bu$ reads
\begin{equation}
[\Phi^*\bzeta]_{(t,\tau)} (\bu) = \zeta_{i\br} [\Phi(t,\tau)] \delta^{i\br} \left[ \rmd \Phi(\bu) \right] =\zeta_{i\br} [\Phi(t,\tau)] \left[ u^t\partial_t\Phi^{i\br} + u^\tau \partial_\tau \Phi^{i\br} \right] \ , 
\end{equation}
where we just used the linearity of $\delta^{i\br}$.
In terms of one--forms, it reads
\begin{equation}
[\Phi^*\bzeta]_{(t,\tau)} =\zeta_{i\br} [\Phi(t,\tau)] \left[ \partial_t\Phi^{i\br}\rmd t + \partial_\tau \Phi^{i\br}\rmd \tau \right] \ .
\end{equation}
We can then easily compute the exterior derivative
\begin{equation}
[\rmd (\Phi^*\bzeta)]_{(t,\tau)} =  \partial_\tau \left\{\zeta_{i\br} [\Phi(t,\tau)] \partial_t\Phi^{i\br} \right\} \rmd\tau \wedge \rmd t + \partial_t \left\{\zeta_{i\br} [\Phi(t,\tau)] \partial_\tau \Phi^{i\br} \right\} \rmd t\wedge \rmd \tau \ .
\end{equation}
Using the chain rule together with the skew--symmetry property $\rmd \tau \wedge \rmd t = -\rmd t\wedge \rmd \tau$, we get
\begin{equation}
[\rmd (\Phi^*\bzeta)]_{(t,\tau)} = \frac{\delta \zeta_{i\br}}{\delta \phi^{j\br'}} \left[\partial_t \Phi^{j\br'} \partial_\tau\Phi^{i\br} -  \partial_\tau \Phi^{j\br'} \partial_t\Phi^{i\br} \right] \rmd t\wedge \rmd \tau \ ,
\end{equation}
which coincides with the right--hand side of eq.~\eqref{App:pullbackFunctExtDer}. 
Thus, we showed that $\Phi^*\mbd\bzeta = \rmd (\Phi^*\bzeta)$, consequently that eq.~\eqref{App:PullExtDerCommut} holds, and finally that the functional Stokes' theorem~\eqref{App:functionalStokes} follows.

\section{Vorticity of the Active Model B}
\label{app:AMB_vorticity}

The vorticity of AMB~\eqref{AMB}-\eqref{AMB_chemical_potential} was previously computed by one of us in~\cite{o2023nonequilibrium} using a less general version of definition~\eqref{functExtder01}.
In this appendix, for illustrative purposes, we rederive expression~\eqref{cycleAffAMB} for $\bomega$ using a slightly different approach -- namely, formula~\eqref{extDer_convenient} -- which allows us to compute the exterior derivative directly at the level of functional forms, without introducing any pair of perturbations $(\delta\rho_1,\delta\rho_2)$. For completeness, we also derive expression~\eqref{full_vorticity_op_AMB} for the corresponding vorticity operator.

The following calculation involves two integration variables, denoted $\br$ and $\br'$, and the integrals are  written out explicitly for clarity -- that is, without using the generalized Einstein summation convention introduced  in section~\ref{subsec:SPDEframework}. 
Although not written explicitly to keep the notation light, the field $\rho$ and its derivatives are always evaluated at point $\br$, not at $\br'$, and the functions $\lambda$ and $\kappa$ are evaluated at $\rho(\br)$.

As shown at the end of section~\ref{subsec:SPDEframework} and further discussed in section~\ref{subsec:RevCondFuncExtDer}, in the case of AMB: 
\begin{equation}
\bD^{-1}\ba=-\int \rmd \br \mu_\br\delta^\br \ .
\end{equation}
Hence, using eq.~\eqref{extDer_convenient}, we get
\begin{eqnarray}
\bomega \equiv \mbd \bD^{-1}\ba = -\mbd \left[ \int \rmd \br \mu_\br \delta^\br \right] = - \int \rmd \br \; \mbd[\mu_\br]\wedge \delta^\br = - \int \rmd \br \rmd \br' \frac{\delta\mu_\br}{\delta\rho^{\br'}}\delta^{\br'}\wedge \delta^\br 
\label{AMB_vort_calc_01}
\end{eqnarray}
Now, the functional derivative of $\mu$ reads
\begin{eqnarray*}
\frac{\delta \mu_\br}{\delta\rho^{\br'}} &=& \left[ a+2 b \rho + \lambda' |\nabla\rho|^2 -\kappa'\Delta\rho \right] \delta(\br-\br') + 2\lambda\nabla\rho \cdot \nabla_\br\delta(\br-\br') - \kappa \Delta_\br\delta(\br-\br') \\
&=& \left[ a+2 b \rho + \lambda' |\nabla\rho|^2 -\kappa'\Delta\rho \right] \delta(\br-\br') - 2\lambda\nabla\rho \cdot \nabla_{\br'}\delta(\br-\br') - \kappa \Delta_{\br'}\delta(\br-\br') \ ,
\end{eqnarray*}
where we used $\nabla_\br\delta(\br-\br')=-\nabla_{\br'}\delta(\br-\br')$ to go from the first line to the second.
Injecting this in the expression~\eqref{AMB_vort_calc_01} of the vorticity yields
\begin{eqnarray}
\bomega &=& -\int \rmd \br \left[a+2 b \rho + \lambda' |\nabla\rho|^2 -\kappa'\Delta\rho \right] \delta^\br\wedge \delta^\br \notag\\
 & & + \int \rmd \br \rmd \br' \left[ 2\lambda \nabla\rho \cdot \nabla_{\br'}\delta(\br-\br') + \kappa \Delta_{\br'} \delta(\br-\br') \right] \delta^{\br'}\wedge \delta^\br \ .
\label{AMB_vort_calc_02}
\end{eqnarray}
Since the wedge product is antisymmetric, $\delta^\br\wedge\delta^\br=0$, and the first line on the right-hand side of eq.~\eqref{AMB_vort_calc_02} vanishes.
Integrating by parts 
then gives 
\begin{eqnarray*}
\bomega &=&  - \int \rmd \br \rmd \br' \delta(\br-\br')\left[ 2\lambda \nabla\rho \cdot \nabla_{\br'}\left(\delta^{\br'}\wedge \delta^\br\right) - \kappa \Delta_{\br'} \left( \delta^{\br'}\wedge \delta^\br\right) \right] \\
&=&  - \int \rmd \br \rmd \br' \delta(\br-\br')\left[ 2\lambda \nabla\rho \cdot \left(-\nabla\delta^{\br'}\right) \wedge \delta^\br- \kappa \left(\Delta  \delta^{\br'}\right)\wedge \delta^\br \right]  \ ,
\end{eqnarray*}
where the second line is due to the convention $\nabla_{\br'}\left[\delta^{\br'}\wedge \delta^\br\right](\delta\rho_1,\delta\rho_2)\equiv \nabla_{\br'} \left[\delta\rho_1^{\br'}\delta\rho_2^\br-\delta\rho_1^{\br}\delta\rho_2^{\br'}\right]=\delta\rho_2^\br \nabla\delta\rho_1^{\br'}-\delta\rho_1^{\br}\nabla\delta\rho_2^{\br'}$, while $\left[\nabla\delta^{\br'}\right]\wedge\delta^\br(\delta\rho_1,\delta\rho_2)\equiv \left[\nabla\delta^{\br'}(\delta\rho_1)\right] \delta^{\br}(\delta\rho_2)-\left[\nabla\delta^{\br'}(\delta\rho_2)\right] \delta^{\br}(\delta\rho_1)=\left[-\nabla\delta\rho_1^{\br'}\right]\delta\rho_2^\br -\left[-\nabla\delta\rho_2^{\br'}\right]\delta\rho_1^\br= \delta\rho_1^\br\nabla\delta\rho_2^{\br'} - \delta\rho_2^\br\nabla\delta\rho_1^{\br'}$, the sign difference arising from the definition of the Dirac-delta gradient.
Integrating out $\br'$, and using the skew-symmetry of the wedge product, the vorticity reads
\begin{eqnarray*}
\bomega &=& - \int \rmd \br \left[ 2\lambda \nabla\rho \cdot \left(-\nabla\delta^{\br}\right) \wedge \delta^\br- \kappa \left(\Delta  \delta^{\br}\right)\wedge \delta^\br \right] \\
&=& - \int \rmd \br \left[ 2\lambda \nabla\rho \cdot \delta^\br \wedge \nabla\delta^{\br} + \kappa \delta^\br \wedge \Delta  \delta^{\br} \right] \ .
\end{eqnarray*}
Finally, using eq.~\eqref{relation_a_generaliser}, and integrating by parts yields the formula we wanted to prove:
\begin{equation}
\bomega = - \int \rmd \br \left[ 2\lambda+\kappa'\right]\nabla\rho \cdot \delta^\br \wedge \nabla\delta^\br \ .
\label{append_vort_form_AMB}
\end{equation}

Let us now turn to the vorticity operator $\bW$ and show that it is given by eq.~\eqref{full_vorticity_op_AMB} in the case of AMB.
We start by deriving the expression for the corresponding cycle affinity operator $\whomega$. To do so, we first apply the vorticity 2-form~\eqref{append_vort_form_AMB} to a pair of fluctuations $(\delta\rho_1,\delta\rho_2)$ and integrate by parts to factorize $\delta\rho_1$ on the left-hand side of the integral:
\begin{eqnarray*}
\bomega_{\rho}(\delta\rho_1,\delta\rho_2) &=& \int \rmd \br \, (2\lambda+\kappa')\nabla\rho \cdot \left[\delta\rho_1\nabla\delta\rho_2  -\delta\rho_2\nabla\delta\rho_1 \right] \\
&=& \int \rmd \br \, \delta\rho_1 \Big\{ (2\lambda+\kappa')\nabla\rho \cdot \nabla\delta\rho_2 + \nabla\cdot \left[\delta\rho_2(2\lambda+\kappa')\nabla\rho\right] \Big\} \ .
\end{eqnarray*}  
Using definition~\eqref{def_affinity_operator}, we then identify the curly bracket in this last integral as $\whomega_\rho \delta\rho_2$. Applying formula~\eqref{Omega_vs_whomega}, with $\bD=-M\Delta$ for AMB, finally yield the desired expression for the vorticity operator:
\begin{equation}
\bW_{\rho} \delta\rho =\frac{M}{2} \Delta\Big\{ (2\lambda+\kappa')\nabla\rho \cdot \nabla\delta\rho + \nabla\cdot \left[\delta\rho(2\lambda+\kappa')\nabla\rho\right] \Big\} \ .
\end{equation}

\newpage

\section{Intermediate results about polynomials and antisymmetric differential operators}
\label{app:factorisation_operators}
The objective of this appendix is to show that the operators
\begin{center}
$\begin{array}{ccccccccccc}
& \mcD^A_k  : & \mbH\times \mbH & \longrightarrow & \mbH & \quad \text{and} \quad & \mcD^S_k : & \mbH\times \mbH & \longrightarrow & \mbH \\\
& & (\psi_1,\psi_2) & \longmapsto & \psi_2\partial^k\psi_1 - \psi_1\partial^k\psi_2 & \quad \quad &   & (\psi_1,\psi_2) & \longmapsto & \psi_2\partial^k\psi_1 + \psi_1\partial^k\psi_2 \\
\end{array}$
\end{center}
which are, up to a factor 2, respectively the skewsymmetric and symmetric part of the operators
\begin{center}
$\begin{array}{ccccc}
& \mcD_k  : & \mbH\times \mbH & \longrightarrow & \mbH \\
& & (\psi_1,\psi_2) & \longmapsto & \psi_2(x)\partial^k\psi_1(x) \\
\end{array}$
\end{center}
where $\mbH$ is the space of smooth functions from $\mbR$ to itself, can be decomposed in the following way:
\begin{equation}
\forall j \hspace{0,3cm} \mcD^A_{2j}=\sum_{i=0}^{j-1} b^{2j}_{2i+1}\partial^{(2j-2i-1)}\mcD^A_{2i+1} \ ,
\end{equation}
\begin{equation}
\forall j \hspace{0,3cm} \mcD^S_{2j+1}=\sum_{i=0}^jc^{2j+1}_{2i}\partial^{(2j+1-2i)}\mcD^S_{2i} \ ,
\end{equation}
where the coefficients $b^{2j}_{2i+1}$ and $c^{2j+1}_{2i}$ and recursively defined through:
\begin{equation}
\forall j\in \mathbb{N}^*, \forall i<j, \hspace{0,3cm}  b^{2j}_{2i+1}=a^{2j}_{2i+1} + \sum_{k=i+1}^{j-1} a^{2j}_{2k}b^{2k}_{2i+1} \ ,
\end{equation}
\begin{equation}
\text{and}\hspace{0,3cm}\forall j\in \mathbb{N}, \forall i\leq j, \hspace{0,3cm}  c^{2j+1}_{2i}=a^{2j+1}_{2i} + \sum_{k=i}^{j-1} \ , a^{2j+1}_{2k+1}c^{2k+1}_{2i}
\end{equation}
\begin{equation}
\text{with}\hspace{0,3cm}a^{n}_k = \frac{(-1)^{n-1-k}}{2}\binom{n}{k} \ .
\end{equation}
To this end, we first show in section~\ref{app:factorisation_polynomials} a similar factorisation for some polynomials in $\mbR[X,Y]$. Then, thanks to an algebra isomorphism between $\mbR[X,Y]$ and an appropriate algebra of operators, we deduce in section~\ref{app:algebra_isomorphism} the above formula.

\subsection{Factorisation of polynomials}
\label{app:factorisation_polynomials}
The objective of this section is to prove lemma~\ref{lemma3} below. We will first prove the preliminary lemma~\ref{lemma2}.

\begin{Lemme}
\label{lemma2}
For any integer $n$,

\noindent -- if $n$ is even: 
 \begin{equation}
 X^n-Y^n = \sum_{k=0}^{n-1}a^{n}_k(X+Y)^{n-k}(X^{k}-Y^{k}) \ ,
 \end{equation}
-- if $n$ is odd:
  \begin{equation}
  X^n+Y^n = \sum_{k=0}^{n-1}a^{n}_k(X+Y)^{n-k}(X^{k}+Y^{k}) \ ,
  \end{equation}
  with 
  \begin{equation}
  a^{n}_k \equiv\frac{(-1)^{n+k+1}}{2}\binom{n}{k} \ .
  \label{def_coeff_ank}
  \end{equation}
  \end{Lemme}

\begin{proof}
Let us consider an arbitrary natural number $n$, and two unknown variables $X$ and $Y$. Using the binomial formula, we get :
\begin{eqnarray*}
X^n+Y^n &=& [(X+Y)-Y]^n+[(X+Y)-X]^n \\
&=& \sum_{k=0}^n \binom{n}{k} (X+Y)^{n-k}(-1)^k(X^k+Y^k) \ .
\end{eqnarray*}
Moving the $k=n$ term to the right hand side of the equality leads to
\begin{equation}
(X^n+Y^n)(1-(-1)^n) = \sum_{k=0}^{n-1} \binom{n}{k} (X+Y)^{n-k}(-1)^k(X^k+Y^k) \ .
\end{equation}
Hence, for $n$ odd, we get
\begin{equation}
X^{n} + Y^{n} = \sum_{k=0}^{n-1} \frac{(-1)^k}{2}\binom{n}{k} (X+Y)^{n-k}(X^k+Y^k) \ .
\end{equation}
Similarly, we have
\begin{eqnarray*}
X^n-Y^n &=& [(X+Y)-Y]^n-[(X+Y)-X]^n \\
&=& \sum_{k=0}^n \binom{n}{k} (X+Y)^{n-k}(-1)^{k+1}(X^k-Y^k) \ ,
\end{eqnarray*}
\ie
\begin{equation}
(X^n-Y^n)(1-(-1)^{n+1}) = \sum_{k=0}^{n-1} (-1)^{k+1}\binom{n}{k} (X+Y)^{n-k}(X^k-Y^k) \ .
\end{equation}
Thus, when $n$ is even
\begin{equation}
(X^n-Y^n) = \sum_{k=0}^{n-1} \frac{(-1)^{k+1}}{2}\binom{n}{k} (X+Y)^{n-k}(X^k-Y^k) \ .
\end{equation}
Introducing the coefficient $a^n_k$, defined by eq.~\eqref{def_coeff_ank}, allows to recover the two formulas we had to prove.
\end{proof}

%
\begin{Lemme} 
\label{lemma3}
For any $j\in\mbN$, we have the factorisations:
\begin{equation}
\hspace{0,3cm} X^{2j}-Y^{2j}=\sum_{i=0}^{j-1} b^{2j}_{2i+1}\big(X+Y\Big)^{2j-2i-1}\Big(X^{2i+1}-Y^{2i+1}\Big) \ ,
\label{lemma3_1}
\end{equation}
\begin{equation}
\text{and} \hspace{0,3cm} X^{2j+1} + Y^{2j+1}=\sum_{i=0}^jc^{2j+1}_{2i}\Big(X+Y\Big)^{2j+1-2i}\Big( X^{2i} + Y^{2i}\Big) \ ,
\label{lemma3_2}
\end{equation}
where the coefficients $b^{2j}_{2i+1}$ and $c^{2j+1}_{2i}$ are recursively defined through:
\begin{equation}
\forall j\in \mathbb{N}^*, \forall i<j, \hspace{0,3cm}  b^{2j}_{2i+1}\equiv a^{2j}_{2i+1} + \sum_{k=i+1}^{j-1} a^{2j}_{2k}b^{2k}_{2i+1} \ ,
\label{expr_coeff_b}
\end{equation}
\begin{equation}
\text{and}\hspace{0,3cm}\forall j\in \mathbb{N}, \forall i\leq j, \hspace{0,3cm}  c^{2j+1}_{2i} \equiv a^{2j+1}_{2i} + \sum_{k=i}^{j-1} a^{2j+1}_{2k+1}c^{2k+1}_{2i} \ ,
\label{expr_coeff_c}
\end{equation}
 \begin{equation}
  \text{with}\hspace{0,3cm}a^{n}_k =\frac{(-1)^{n+k+1}}{2}\binom{n}{k} \ .
  \label{expr_coeff_a}
  \end{equation}
  \end{Lemme}
\begin{proof}
Again, we only prove the odd case, the even case being proved in the same way.
We make the proof through strong induction. The first step $j=0$ gives $0=0$, which is true. We then choose $j\in\mbN$. We assume the lemma to be true $\forall i<j$. We prove it for $j$, using lemma~\ref{lemma2} and then the recurrence hypothesis:

\begin{eqnarray*}
 X^{2j}-Y^{2j} &=& \sum_{i=0}^{2j-1}a^{2j}_i(X+Y)^{2j-i}(X^{i}-Y^{i})\\
&=&  \sum_{k=0}^{j-1}a^{2j}_{2k}(X+Y)^{2j-2k}(X^{2k}-Y^{2k}) +  \sum_{k=0}^{j-1}a^{2j}_{2k+1}(X+Y)^{2j-(2k+1)}(X^{2k+1}-Y^{2k+1})\\
&=&  \sum_{k=0}^{j-1}a^{2j}_{2k}(X+Y)^{2j-2k}\sum_{i=0}^{k-1} b^{2k}_{2i+1}(X+Y)^{2k-2i-1}(X^{2i+1}-Y^{2i+1})\\ & & +  \sum_{i=0}^{j-1}a^{2j}_{2i+1}(X+Y)^{2j-(2i+1)}(X^{2i+1}-Y^{2i+1})\\
&=& \sum_{i=0}^{j-1}  \Big( \sum_{k=i+1}^{j-1} a^{2j}_{2k} b^{2k}_{2i+1} \Big) (X+Y)^{2j-2i-1} (X^{2i+1} -Y^{2i+1})\\ & & +  \sum_{i=0}^{j-1}a^{2j}_{2i+1}(X+Y)^{2j-2i-1}(X^{2i+1}-Y^{2i+1})\\
&=& \sum_{i=0}^{j-1}  \Big( a^{2j}_{2i+1} + \sum_{k=i+1}^{j-1} a^{2j}_{2k} b^{2k}_{2i+1} \Big) (X+Y)^{2j-2i-1} (X^{2i+1} -Y^{2i+1})
\end{eqnarray*}
\end{proof}

\subsection{Algebra of operators and algebra isomorphism}
\label{app:algebra_isomorphism}
An important problem that we face from the beginning is that if we treat the operators
\begin{center}
 $\begin{array}{ccccc}
\mcD_k & : & \mbH\times \mbH & \longrightarrow & \mbH \\
 & & (\psi_1, \psi_2) & \longmapsto & \psi_2 (x)\partial^{k}\psi_1 (x)\\
\end{array}$
\end{center} 
directly, we won't be able to compose them, which is not ideal to solve a factorisation problem. Hence, we are going to consider first the lifted operators:
 \begin{center}
 $\begin{array}{ccccc}
\tilde{\mcD}_k & : & \mbH\times \mbH & \longrightarrow & \mbH\otimes \mbH \\
 & & (\psi_1, \psi_2) & \longmapsto & \partial^{k}\psi_1 \otimes \psi_2\\
\end{array}$
\end{center}
where $\psi_2 \otimes \partial^{k}\psi_1:(x,y)\mapsto\psi_2(x) \partial^{k}\psi_1 (y)$ is the tensor product
of functions $\psi_2$ and $\partial^{k}\psi_1$. We recover $\mcD_k$ from $\tilde{\mcD}_k$ using the linear projection:
 \begin{center}
 $\begin{array}{ccccc}
\pi_\mbH & : & \mbH\otimes \mbH & \longrightarrow & \mbH \\
 & & f(x,y) & \longmapsto & \pi_\mbH(f)(x)\equiv f(x,x) \ . \\
\end{array}$
\end{center}
Indeed 
\begin{equation}
\mcD_k = \pi_\mbH\circ\tilde{\mcD}_k \ ,
\end{equation}
where $\circ$ denotes the composition.
We note that, since the operators $\tilde{\mcD}_k$ are bilinear, they factorize through $\otimes$ in a unique way. In other words, the fundamental property of the tensor product ensures the existence, for all integer $k$, of a unique operator $\bar{\mcD}_k$
 \begin{center}
 $\begin{array}{ccccc}
\bar{\mcD}_k & : & \mbH\otimes \mbH & \longrightarrow & \mbH\otimes \mbH \\
 & & f(x,y) & \longmapsto &\partial^{k}_xf(x,y)\\
\end{array}$
\end{center}
We can now work with the operators $\bar{\mcD}_k$ that are endomorphisms of $\mbH\otimes\mbH$, \ie $\bar{\mcD}_k\in \mathrm{End}(\mbH \otimes \mbH )$. Moreover, they can be written as
\begin{equation}
\bar{\mcD}_k= \partial^k \otimes id_\mbH \ ,
\end{equation}
where $id_\mbH$ is the identity of $\mathrm{End}(\mbH)$. The transitions between the different spaces is summarized in the following commutative diagram.
\begin{large} 
\begin{displaymath}
\xymatrix @ !=2cm{
    H\times H \ar[r]^{\otimes} \ar[d]_{ \mcD_k} \ar[rd]^{\tilde{\mcD}_k} & H\otimes H \ar[d]^{\bar{\mcD}_k} \\
     H & H\otimes H   \ar[l]^{\pi_\mbH}
  }
  \end{displaymath}
\end{large}
The other crucial point of this change of space is that the derivation operator $\partial \equiv \frac{\rmd}{\rmd x}$ of $\mbH$ corresponds in $\mbH\otimes\mbH$ to the operator $\partial\otimes id_\mbH + id_\mbH \otimes \partial$ in the sense that
\begin{equation}
\frac{\rmd}{\rmd x} = \pi_\mbH \circ (\partial\otimes id_\mbH + id_\mbH \otimes \partial) \ .
\end{equation}
Indeed
\begin{eqnarray*}
\left[\pi_\mbH \circ (\partial\otimes id_\mbH + id_\mbH \otimes \partial)\right] (\psi_1 \otimes \psi_2)(x) &=& \int \left[ \partial_x\psi_1(x)\psi_2(y) + \psi_1(x) \partial_y \psi_2(y) \right] \delta(x-y) \rmd y \qquad \\
&=&  \psi_2(x)\partial_x\psi_1(x) + \psi_1(x) \partial_x \psi_2(x)  \\
&=& \frac{\rmd}{\rmd x} \left( \psi_1(x)\psi_2(x)\right) \ .
\end{eqnarray*}
More generally, it can be shown that
\begin{equation}
\frac{\rmd^k}{\rmd x^k} = \pi_\mbH \circ (\partial\otimes id_\mbH + id_\mbH \otimes \partial)^k \ .
\end{equation}

We can now prove the factorisation announced  at the beginning of this appendix by using the factorisation of polynomials shown in section~\ref{app:factorisation_polynomials}.
To do so, we first defined the operators:
\begin{center}
$\begin{array}{cccccccccc}
& \bar{\mcD}^A_k  : & \mbH\otimes \mbH & \longrightarrow & \mbH\otimes \mbH  \\\
& & \psi_1\otimes \psi_2 & \longmapsto & \partial^k\psi_1\otimes\psi_2 - \psi_1\otimes \partial^k\psi_2  \\
\end{array}$
\end{center}
and
\begin{center}
$\begin{array}{cccccccccc}
& \bar{\mcD}^S_k : & \mbH\otimes \mbH & \longrightarrow & \mbH\otimes \mbH \\\
&  & \psi_1\otimes\psi_2 & \longmapsto & \partial^k\psi_1\otimes\psi_2 + \psi_1\otimes\partial^k\psi_2 \\
\end{array}$
\end{center}
or in other words
\begin{equation}
\bar{\mcD}^A_k = \partial^k\otimes id_\mbH - id_\mbH\otimes \partial^k \quad \text{and} \quad \bar{\mcD}^S_k = \partial^k\otimes id_\mbH + id_\mbH\otimes \partial^k \ .
\end{equation}
Let us denote by $\mbA$ the subalgebra of $\text{End}(\mbH\otimes\mbH)$ generated by the families $(\mcD^A_k,\mcD_k^S)_{k\in\mbN}$.
The algebra $(\mbA, +, \cdot , \circ)$,  where $+$, $\cdot$ and $\circ$ are respectively the sum of operators, the product by a real number, and the composition of operators, is a commutative algebra made out of elements of the form
\begin{equation}
\sum_{k,n} \lambda_{k,n} \partial^k\otimes\partial^n \ ,
\end{equation}
where the sum is over a finite set of natural integers and $\lambda_{k,n}\in\mbR$.
We then define the map 
\begin{equation}
\Lambda : \mbA \to \mbR[X,Y]
\end{equation}
such that for any pair of natural numbers $(k,n)$
\begin{equation}
\Lambda((\partial^k\otimes id_\mbH)\circ(id_\mbH\otimes\partial^n))=\Lambda(\partial^k\otimes\partial^n) \equiv X^k Y^n
\end{equation}
that we extend by linearity to the full algebra $\mbA$.
Note that the first equality in the definition of $\Lambda$ has been explicitly written down to insist on the fact that it is the composition in $\mbA$ that corresponds to the product in $\mbR[X,Y]$ and not the tensor product. 

The important point here is that $\Lambda$ is an algebra isomorphism since it is a morphism that sends a basis of $\mbA$ (the family ($\partial^k\otimes\partial^n)_{k,n\in\mbN}$) to the canonical basis of $\mbR[X,Y]$. 
Thus applying the map $\pi_\mbH\circ\Lambda^{-1}$ to the factorisation formula of lemma~\ref{lemma3} allows us to complete the proof.

Indeed, by first applying $\Lambda^{-1}$ to eqs.~\eqref{lemma3_1} and~\eqref{lemma3_2} of lemma~\ref{lemma3}, we get that for all natural numbers $j$: 
\begin{equation}
\partial^{2j}\otimes id_\mbH - id_\mbH\otimes\partial^{2j}=\sum_{i=0}^{j-1} b^{2j}_{2i+1}\left[ (\partial\otimes id_\mbH) + (id_\mbH\otimes\partial) \right]^{2j-2i-1} \circ \left[ \partial^{2i+1}\otimes id_\mbH - id_\mbH\otimes\partial^{2i+1}\right]
\label{eq:factorisation-operateur-avant-projection-antisym}
\end{equation}
and 
\begin{equation}
\partial^{2j+1}\otimes id_\mbH + id_\mbH\otimes\partial^{2j+1}=\sum_{i=0}^jc^{2j+1}_{2i}\left[ (\partial\otimes id_\mbH) + (id_\mbH\otimes\partial) \right]^{2j+1-2i} \circ \left[ \partial^{2i}\otimes id_\mbH + id_\mbH\otimes\partial^{2i}\right] \ .
\label{eq:factorisation-operateur-avant-projection-sym}
\end{equation}
Besides, for any $k$ and $n$:
\begin{eqnarray*}
& &\left[ \pi_\mbH \circ \left[ (\partial\otimes id_\mbH) + (id_\mbH\otimes\partial) \right]^k \circ \left[ \partial^n\otimes id_\mbH \pm id_\mbH\otimes\partial^n\right]\right] (\psi_1,\psi_2)(x) \\
&=& \left[ \pi_\mbH \circ \left[ (\partial\otimes id_\mbH) + (id_\mbH\otimes\partial) \right]^k  \right] \left(\partial^n\psi_1\otimes \psi_2 \pm \psi_1\otimes\partial^n\psi_2\right)(x) \\
&=& \frac{\rmd^k}{\rmd x^k} \left( \psi_2(x) \partial^n\psi_1(x) \pm \psi_1(x) \partial^n\psi_2(x)\right) \\
&=& \frac{\rmd^k}{\rmd x^k} \mcD_n^{S,A}(\psi_1,\psi_2)(x) \ .
\end{eqnarray*}
Thus, applying $\pi_\mbH$ to formula~\eqref{eq:factorisation-operateur-avant-projection-antisym} and~\eqref{eq:factorisation-operateur-avant-projection-sym} finally yields:
\begin{equation}
\mcD^A_{2j}=\sum_{i=0}^{j-1} b^{2j}_{2i+1}\frac{\rmd^{(2j-2i-1)}}{\rmd x^{(2j-2i-1)}}\mcD^A_{2i+1}
\label{lala_decomp_antisym}
\end{equation}
and
\begin{equation}
\mcD^S_{2j+1}=\sum_{i=0}^jc^{2j+1}_{2i}\frac{\rmd^{(2j+1-2i)}}{\rmd x^{(2j+1-2i)}}\mcD^S_{2i}
\label{lala_decomp_sym}
\end{equation}
where the coefficients $b_k^n$ and $c_k^n$ are those of lemma~\ref{lemma3}.

\newpage
\section{A basis of 2-forms in $d_1=1$ spatial dimension}
\label{app:proof_Basis_form}

The objective of this appendix is to show that the family~\eqref{1dBasis2Forms} is a basis of the space $\Omega^2_{\rm loc}(\mbF)$ in $d_1=1$ spatial dimension. To lighten notations, we will rather denote a fluctuation $\delta\bphi$ by $\bpsi$.

\subsection{A generative family of local two-forms}
\label{app:generative_family}

In this subsection, we show that the family~\eqref{1dBasis2Forms} generates the space $\Omega^2_{\rm loc}(\mbF)$.
To this end, let us consider a local two-form $\bgamma\in\Omega^2_{\rm loc}(\mbF)$. It generically reads
\begin{equation}
\bgamma=\sum_{k=0}^q \gamma^{(k)}_{ijx}\delta^{ix}\wedge\partial^k\delta^{jx} \ ,
\label{app:generic_2-form}
\end{equation}
for some natural number $q$, where we implicitly sum over $i,j\in\{1,\dots,d_2\}$ and integrate over $x\in\mbR$. We introduce the symmetric and antisymmetric parts of the matrices $\bgamma^{(k)}(x,[\bphi])$, which are given by ${}^{S/A}\gamma^{(k)}_{ij}\equiv (\gamma_{ij}^{(k)}\pm\gamma^{(k)}_{ji})/2$. Distinguishing symmetric from antisymmetric parts of each $\bgamma^{(k)}$ and separating odd and even degree derivatives in~\eqref{app:generic_2-form} gives:
\begin{eqnarray}
\bgamma = \sum_{\ell=0}^{\left \lfloor q/2 \right \rfloor} \left[{}^S\gamma^{(2\ell)}_{ijx}+{}^A\gamma^{(2\ell)}_{ijx} \right]\delta^{ix}\wedge\partial^{2\ell}\delta^{jx} + \sum_{\ell=0}^{\left \lfloor (q-1)/2 \right \rfloor} \left[{}^S\gamma^{(2\ell+1)}_{ijx}+{}^A\gamma^{(2\ell+1)}_{ijx} \right]\delta^{ix}\wedge\partial^{2\ell+1}\delta^{jx} \ .
\end{eqnarray}
For each implicit sum over $i,j=\{1,\dots,d_2\}$ above, we add a copy of it where the dummy indices are exchanged $i\leftrightarrow j$, and divide everything by two so that the overall quantity is unchanged.
Using the (anti)symmetry of ${}^{A/S}\bgamma^{(k)}$, we get
\begin{eqnarray}
\bgamma &=&\frac{1}{2}\sum_{\ell=0}^{\left \lfloor q/2 \right \rfloor} \left\{{}^S\gamma^{(2\ell)}_{ijx}\left[\delta^{ix}\wedge\partial^{2\ell}\delta^{jx}+\delta^{jx}\wedge\partial^{2\ell}\delta^{ix} \right]
+{}^A\gamma^{(2\ell)}_{ijx} \left[\delta^{ix}\wedge\partial^{2\ell}\delta^{jx}-\delta^{jx}\wedge\partial^{2\ell}\delta^{ix}\right] \right\} \notag\\
&+& \frac{1}{2}\sum_{\ell=0}^{\left \lfloor (q-1)/2 \right \rfloor} \Big\{{}^S\gamma^{(2\ell+1)}_{ijx}\left[\delta^{ix}\wedge\partial^{2\ell+1}\delta^{jx}+\delta^{jx}\wedge\partial^{2\ell+1}\delta^{ix} \right] \label{app:decompo_basis_01}  \\
 &+& {}^A\gamma^{(2\ell+1)}_{ijx} \left[\delta^{ix}\wedge\partial^{2\ell+1}\delta^{jx}-\delta^{jx}\wedge\partial^{2\ell+1}\delta^{ix}\right] \Big\} \notag \ .
\end{eqnarray}
We see that the terms proportional to ${}^A\gamma^{(2\ell)}_{ijx}$ and ${}^S\gamma^{(2\ell+1)}_{ijx}$ are already decomposed in the basis~\eqref{1dBasis2Forms}.
In what follows, we show that the other terms can be decomposed as well in that basis.
Let us start by relating the elementary two-forms in~\eqref{app:decompo_basis_01} that do not belong to the basis~\eqref{1dBasis2Forms} to the differential operator $\mcD^{S/A}_{k}$ defined in appendix~\ref{app:factorisation_operators}.
For any pair $\bphi_1,\bpsi_2$ of perturbations and any natural number $k$, we have
\begin{eqnarray*}
\left[\delta^i\wedge\partial^k\delta^j \pm \delta^j\wedge\partial^k\delta^i  \right](\bpsi_1,\bpsi_2) &=& (-1)^k \left\{ [\psi_1^i\partial^k\psi_2^j-\psi_2^i\partial^k\psi_1^j] \pm [\psi_1^j\partial^k\psi_2^i-\psi_2^j\partial^k\psi_1^i] \right\} \\
&=& (-1)^k \left\{ [\psi_1^i\partial^k\psi_2^j\mp\psi_2^j\partial^k\psi_1^i] \pm [\psi_1^j\partial^k\psi_2^i \mp \psi_2^i\partial^k\psi_1^j] \right\} \ ,
\end{eqnarray*}
where we simply rearranged the terms.
Consequently, using the definitions of $\mcD^{S/A}_{k}$ from appendix~\ref{app:factorisation_operators}, we have the relation
\begin{equation}
\left[\delta^i\wedge\partial^k\delta^j \pm \delta^j\wedge\partial^k\delta^i  \right](\bpsi_1,\bpsi_2) = (-1)^k \left\{ \mcD^{A/S}_k(\psi_2^j,\psi_1^i) \pm  \mcD^{A/S}_k(\psi_2^i,\psi_1^j)\right\} \ .
\label{app:relation_2forms_interm_op}
\end{equation}
Using this relation together with decomposition~\eqref{lala_decomp_antisym} on the one hand, we have:
\begin{eqnarray*}
\left[\delta^i\wedge\partial^{2\ell}\delta^j + \delta^j\wedge\partial^{2\ell}\delta^i  \right](\bpsi_1,\bpsi_2) &=& \mcD^{A}_{2\ell}(\psi_2^j,\psi_1^i) +  \mcD^{A}_{2\ell}(\psi_2^i,\psi_1^j) \\
&=& \sum_{k=0}^{\ell-1}b_{2k+1}^{2\ell}\partial^{2\ell-2k-1}_x\left[\mcD^{A}_{2k+1}(\psi_2^j,\psi_1^i) +  \mcD^{A}_{2k+1}(\psi_2^i,\psi_1^j) \right] \\
&=& - \sum_{k=0}^{\ell-1}b_{2k+1}^{2\ell}\partial^{2\ell-2k-1}_x \left\{ \left[ \delta^i\wedge\partial^{2k+1}\delta^j + \delta^j\wedge\partial^{2k+1}\delta^i \right] (\bpsi_1,\bpsi_2)\right\} \ ,
\end{eqnarray*}
\ie 
\begin{equation}
\delta^i\wedge\partial^{2\ell}\delta^j + \delta^j\wedge\partial^{2\ell}\delta^i  = - \sum_{k=0}^{\ell-1}b_{2k+1}^{2\ell}\partial^{2\ell-2k-1}_x\left[ \delta^i\wedge\partial^{2k+1}\delta^j + \delta^j\wedge\partial^{2k+1}\delta^i \right]  \ .
\label{decompo_elementary_2forms_plus}
\end{equation}
On the other hand, using~\eqref{app:relation_2forms_interm_op} and~\eqref{lala_decomp_sym} gives:
\begin{eqnarray*}
\left[\delta^i\wedge\partial^{2\ell+1}\delta^j - \delta^j\wedge\partial^{2\ell+1}\delta^i  \right](\bpsi_1,\bpsi_2) &=& -\mcD^{S}_{2\ell+1}(\psi_2^j,\psi_1^i) +  \mcD^{S}_{2\ell+1}(\psi_2^i,\psi_1^j) \\
&=& \sum_{k=0}^{\ell}c_{2k}^{2\ell+1}\partial^{2\ell-2k+1}_x\left[-\mcD^{S}_{2k}(\psi_2^j,\psi_1^i) +  \mcD^{S}_{2k}(\psi_2^i,\psi_1^j) \right] \\
&=& - \sum_{k=0}^{\ell}c_{2k}^{2\ell+1}\partial^{2\ell-2k+1}_x \left\{\left[ \delta^i\wedge\partial^{2k}\delta^j - \delta^j\wedge\partial^{2k}\delta^i \right] (\bpsi_1,\bpsi_2) \right\} \ ,
\end{eqnarray*}
\ie 
\begin{equation}
\delta^i\wedge\partial^{2\ell+1}\delta^j - \delta^j\wedge\partial^{2\ell+1}\delta^i  = - \sum_{k=0}^{\ell}c_{2k}^{2\ell+1}\partial^{2\ell-2k+1}_x \left[ \delta^i\wedge\partial^{2k}\delta^j - \delta^j\wedge\partial^{2k}\delta^i \right] \ .
\label{decompo_elementary_2forms_minus}
\end{equation}
Eqs.~\eqref{decompo_elementary_2forms_plus} and~\eqref{decompo_elementary_2forms_minus} are the generalization of the ($d_1=1$ dimensional version of) eq.~\eqref{relation_a_generaliser} to higher dimensions $d_2$ and orders in derivation.

We can now decompose on the basis~\eqref{1dBasis2Forms} the sum over the $ {}^S\gamma^{(2\ell)}_{ijx}$ in eq.~\eqref{app:decompo_basis_01} and rearrange the terms as follows:
\begin{eqnarray*}
& & \sum_{\ell=0}^{\left \lfloor q/2 \right \rfloor} {}^S\gamma^{(2\ell)}_{ijx}\left[\delta^{ix}\wedge\partial^{2\ell}\delta^{jx}+\delta^{jx}\wedge\partial^{2\ell}\delta^{ix} \right] \\
&=& \sum_{\ell=0}^{\left \lfloor q/2 \right \rfloor} \sum_{k=0}^{\ell-1}b_{2k+1}^{2\ell}\left[\partial^{2\ell-2k-1}_x \left({}^S\gamma^{(2\ell)}_{ijx}\right) \right]\left[ \delta^i\wedge\partial^{2k+1}\delta^j + \delta^j\wedge\partial^{2k+1}\delta^i \right] \\
&=& \sum_{k=0}^{\left \lfloor q/2 \right \rfloor -1} \sum_{\ell=k+1}^{\left \lfloor q/2 \right \rfloor}b_{2k+1}^{2\ell}\left[\partial^{2\ell-2k-1}_x \left({}^S\gamma^{(2\ell)}_{ijx}\right) \right]\left[ \delta^i\wedge\partial^{2k+1}\delta^j + \delta^j\wedge\partial^{2k+1}\delta^i \right] \\
&=& \sum_{\ell=0}^{\left \lfloor (q-1)/2 \right \rfloor } \sum_{k=\ell+1}^{\left \lfloor q/2 \right \rfloor}b_{2\ell+1}^{2k}\left[\partial^{2k-2\ell-1}_x \left({}^S\gamma^{(2k)}_{ijx} \right) \right]\left[ \delta^i\wedge\partial^{2\ell+1}\delta^j + \delta^j\wedge\partial^{2\ell+1}\delta^i \right] \ .
\end{eqnarray*}
To get the first equality above, we used eq.~\eqref{decompo_elementary_2forms_plus} and integrated by parts several times. The second equality is then obtained by permuting the two sums, while the third one comes from exchanging the dummy indices $k\leftrightarrow\ell$ and replacing the upper bound of the first sum by $\left \lfloor (q-1)/2 \right \rfloor\geq \left \lfloor q/2 \right \rfloor -1$, which is harmless since, when this inequality is strict, the second sum is empty.
Using eq.~\eqref{decompo_elementary_2forms_minus}, we can perform the same calculation with the sum over the $ {}^A\gamma^{(2\ell+1)}_{ijx}$ in eq.~\eqref{app:decompo_basis_01}:
\begin{eqnarray*}
& & \sum_{\ell=0}^{\left \lfloor (q-1)/2 \right \rfloor} {}^A\gamma^{(2\ell+1)}_{ijx} \left[\delta^{ix}\wedge\partial^{2\ell+1}\delta^{jx}-\delta^{jx}\wedge\partial^{2\ell+1}\delta^{ix}\right] \\
&=& \sum_{\ell=0}^{\left \lfloor (q-1)/2 \right \rfloor} \sum_{k=0}^{\ell}c_{2k}^{2\ell+1} \left[\partial^{2\ell-2k+1}_x \left( {}^A\gamma^{(2\ell+1)}_{ijx}\right)\right] \left[ \delta^i\wedge\partial^{2k}\delta^j - \delta^j\wedge\partial^{2k}\delta^i \right] \\
&=& \sum_{k=0}^{\left \lfloor (q-1)/2 \right \rfloor} \sum_{k=\ell}^{\left \lfloor (q-1)/2 \right \rfloor}c_{2k}^{2\ell+1} \left[\partial^{2\ell-2k+1}_x \left( {}^A\gamma^{(2\ell+1)}_{ijx}\right)\right] \left[ \delta^i\wedge\partial^{2k}\delta^j - \delta^j\wedge\partial^{2k}\delta^i \right] \\
&=& \sum_{\ell=0}^{\left \lfloor q/2 \right \rfloor} \sum_{\ell=k}^{\left \lfloor (q-1)/2 \right \rfloor}c_{2\ell}^{2k+1} \left[\partial^{2k-2\ell+1}_x \left( {}^A\gamma^{(2k+1)}_{ijx}\right)\right] \left[ \delta^i\wedge\partial^{2\ell}\delta^j - \delta^j\wedge\partial^{2\ell}\delta^i \right]
\end{eqnarray*}
Injecting these expressions back into eq.~\eqref{app:decompo_basis_01} and rearranging the (implicit) sums over $i,j=1,\dots,d_2$ using the symmetry of the corresponding terms, we get formulas~\eqref{eq:generalLoc2Form01}-\eqref{eq:expressionBeta_two-form}, from which we deduce that the family~\eqref{1dBasis2Forms} indeed generates the space $\Omega^2_{\rm loc}(\mbF)$.

\subsection{A free family of local 2-forms}
\label{subsec:free_family}

The objective of this subsection is to show that the family~\eqref{1dBasis2Forms} of functional 2--forms 
is free in the following sense:
if there exist smooth functions $\alpha_{ij}^{(k)}(x)$ and $\beta_{ij}^{(k)}(x)$ on $\mbR$ such that
\begin{eqnarray}
& &\sum_{1\leq i < j\leq d_2} \sum_k \int \rmd x \; \alpha_{ij}^{(k)}(x)  \left[\delta^{ix}\wedge\partial^{2k}\delta^{jx} - \delta^{jx}\wedge\partial^{2k}\delta^{ix} \right] \notag \\ 
& & + \sum_{1\leq i \leq j \leq d_2} \sum_k \int \rmd x \;  \beta_{ij}^{(k)}(x)\left[\delta^{ix}\wedge\partial^{2k+1}\delta^{jx} + \delta^{jx}\wedge\partial^{2k+1}\delta^{ix} \right] = 0 \ ,
\label{cancel_free_two_forms_01}
\end{eqnarray}
where the sums over $k$ are finite, then the functions $\alpha_{ij}^{(k)}$ and $\beta_{ij}^{(k)}$ are all identically zero.
Note that our notions of ``generative'' and ``free'' for a family of vectors (here elements of $\Omega^2(\mbF)$) slightly vary from the usual ones in which linear combinations are always finite, whereas here the sums over discrete indices are finite, but we also integrate (\ie ``sum over infinitely many terms'') over the continuous variable $x$.
\\

Let us start by denoting by $\bomega$ the 2--form on the left--hand side of eq.~\eqref{cancel_free_two_forms_01}.
The functions $\alpha_{ij}^{(k)}$ (respectively $\beta{ij}^{(k)}$) are only defined for $1\leq i < j \leq d_2$ ( $1\leq i \leq j \leq d_2$ respectively). For $1\leq j \leq i \leq d_2$, we define $\alpha_{ij}^{(k)}\equiv-\alpha_{ji}^{(k)}$ and $\beta_{ij}^{(k)}\equiv \beta_{ji}^{(k)}$.
Equation~\eqref{cancel_free_two_forms_01} then means that for any pair of functions $(\bpsi_1,\bpsi_2)$ from $\mbR$ to $\mbR^{d_2}$:\begin{eqnarray}
\bomega(\bpsi_1,\bpsi_2) = 0 \ ,
\label{cancel_free_two_forms_02}
\end{eqnarray}
with
\begin{eqnarray}
\bomega(\bpsi_1,\bpsi_2) &=& \sum_{i,j=1}^{d_2} \sum_k \int \rmd x \; \alpha_{ij}^{(k)}(x)  \left[ \psi_1^{ix}\partial^{2k}\psi_2^{jx} - \psi_2^{ix}\partial^{2k}\psi_1^{jx} \right]   \\ 
&-&\sum_{i,j=1}^{d_2} \sum_k \int \rmd x \;  \beta_{ij}^{(k)}(x) \left[ \psi_1^{ix}\partial^{2k+1}\psi_2^{jx} - \psi_2^{ix}\partial^{2k+1}\psi_1^{jx} \right]  \ .\notag
\end{eqnarray}
Let us denote by $\hat{\bpsi}_\ell$ (respectively by $\halpha_{ij}^{(k)}$ and $\hbeta_{ij}^{(k)}$) the Fourier transform of $\bpsi_\ell$ (respectively of $\alpha_{ij}^{(k)}$ and $\beta_{ij}^{(k)}$) with the following convention:
\begin{equation}
\psi^i_\ell(x) = \int \hat{\psi}_\ell^i(p)e^{2i\pi px} \rmd p \ ,
\end{equation}
and similarly for the $\alpha_{ij}^{(k)}$ and $\beta_{ij}^{(k)}$.
We can then rewrite $\bomega(\bpsi_1,\bpsi_2)$ as follows
\begin{eqnarray}
\bomega(\bpsi_1,\bpsi_2) &=& \sum_{i,j=1}^{d_2} \sum_k \int \rmd x \rmd p \rmd q  \; \alpha_{ij}^{(k)}(x) e^{2i\pi(p+q)x}  \left[ \hpsi_1^{ip}\hpsi_2^{jq} (2i\pi q)^{2k} - \hpsi_2^{iq}\hpsi_1^{jp}(2i\pi p)^{2k} \right]  \\ 
&-&\sum_{i,j=1}^{d_2} \sum_k \int \rmd x \rmd p \rmd q   \;  \beta_{ij}^{(k)}(x)  e^{2i\pi(p+q)x} \left[ \hpsi_1^{ip}\hpsi_2^{jq}(2i\pi q)^{2k+1} - \hpsi_2^{iq}\hpsi_1^{jp}(2i\pi p)^{2k+1} \right]  \ . \notag
\end{eqnarray}
We now choose $\bpsi_1,\bpsi_2$ such that $\hpsi_1^i(p)=\delta^{i,m}\delta(p-p_0)$ and $\hpsi_2^i(q)=\delta^{i,n}\delta(q-q_0)$. We thus have
\begin{eqnarray}
\bomega(\bpsi_1,\bpsi_2) &=&  \sum_k (2i\pi)^{2k}   \;   \left[  p_0^{2k} + q_0^{2k} \right] \halpha_{ij}^{(k)}(-p_0-q_0)   \notag\\ 
&+& \sum_k (2i\pi)^{2k+1} \;   \left[p_0^{2k+1} - q_0^{2k+1} \right] \hbeta_{ij}^{(k)}(-p_0-q_0)    \ .
\label{cancel_free_two_forms_03}
\end{eqnarray}
Upon defining
\begin{equation}
r \equiv p_0+q_0 \ ,  \quad s\equiv p_0-q_0 \ ,
\end{equation}
together with the polynomials
\begin{equation}
P_{k,r}(s)\equiv  (r+s)^{2k} + (r-s)^{2k} \quad \text{and} \quad Q_{k,r}(s)\equiv  (r+s)^{2k+1} - (r-s)^{2k+1} \ ,
\end{equation}
eq.~\eqref{cancel_free_two_forms_03} reads
\begin{eqnarray}
\bomega(\bpsi_1,\bpsi_2) &=&  \sum_k (i\pi)^{2k}   \;   \halpha_{ij}^{(k)}(-r) P_{k,r}(s) + \sum_k (i\pi)^{2k+1} \;   \hbeta_{ij}^{(k)}(-r) Q_{k,r}(s)   \ .
\label{freedom_polynomials}
\end{eqnarray}
The assumption~\eqref{cancel_free_two_forms_02} implies that the right--hand side of eq.~\eqref{freedom_polynomials} vanishes for all $r,s\in \mbR$. But the family of polynomials (in the $s$ variable) $(P_{k,r}(s),Q_{k,r}(s))_{k\in\mbN}$ is clearly free (since it is staggered in degree). Thus, we have 
\begin{equation}
\forall r\in \mbR \ , \ \ \halpha_{ij}^{(k)}(-r)= \hbeta{ij}^{(k)}(-r) = 0 \ ,
\end{equation}
\ie the Fourier transforms of the $\alpha_{ij}^{(k)}$'s and $\beta{ij}^{(k)}$'s are identically zero, hence these functions are themselves identically vanishing.

Note that this result can be straightforwardly generalized to the family~\eqref{1dBasis2Forms_higherd1} in dimension $d_1>1$. Indeed, repeating the same steps as above, we get to an equation similar to~\eqref{freedom_polynomials} where the polynomials $P_{k,r}(s)$ and $Q_{k,r}(s)$ are respectively replaced by polynomials $P_{k,r_n}(s_n)$ and $Q_{k,r_n}(s_n)$, with a sum over $n=1,\dots,d_1$. We can eventually conclude by the same argument since the variables $s_n$, $n=1,\dots,d_1$ are independent.

\section{Decomposition of antisymmetric operators}
\label{appendix:another_decomp_of_operators}

In this appendix we prove formulas~\eqref{type_2_decompos} \&~\eqref{type_3_decompos} of section~\ref{subsec:skewsym_operators} in the main text.
To this end, let us associate to any $x$-dependent square matrix $\bM(x)=(M_{ij})_{i,j=1,\dots,d_2}$ the differential operators $\mcP_k^{A}[\bM]$ and $\mcP_k^S[\bM]$, for any $k\in\mbN$, that each map an $\mbR^{d_2}$-valued function $\bphi(x)$ over $\mbR$ to another one, respectively given by
\begin{eqnarray}
[\mcP_k^A[\bM] \bphi]_{i} \equiv \sum_{j=1}^{d_2} \left(M_{ij}\partial^k_x \phi_j - (-1)^k \partial_x^k M_{ij}\phi_j \right) \ ,
\end{eqnarray}
and
\begin{eqnarray}
[\mcP_k^S[\bM] \bphi]_{i} \equiv \sum_{j=1}^{d_2} \left(M_{ij}\partial^k_x \phi_j + (-1)^k \partial_x^k M_{ij}\phi_j \right) \ .
\end{eqnarray}
Taking the $L^2$-scalar product of $\mcP_k^{S/A}[\bM]\bphi$ with another $\mbR^{d_2}$-valued function $\bpsi(x)$ and performing $k$ integrations by parts gives the relations
\begin{equation}
\int \bpsi \cdot \mcP^{S/A}_k[\bM]\bphi \equiv \sum_i\int \psi_i(x) [\mcP_k^{S/A}[\bM] \bphi]_{i}(x)  \rmd x  = \sum_{i,j}\int M_{ij}(x) \mcD_k^{S/A}(\phi_j,\psi_i)(x) \rmd x \ ,
\label{une_de_plus}
\end{equation}
where $\mcD^S_k$ is defined in appendix~\ref{app:factorisation_operators}.
Hence, using relation~\eqref{une_de_plus} together with decompositions~\eqref{lala_decomp_antisym} \&~\eqref{lala_decomp_sym} and performing integrations by parts, we get on the one hand
\begin{eqnarray*}
\int \bpsi \cdot\mcP_{2k}^A[\bM] \bphi &=& \sum_{i,j}\int M_{ij}\sum_{\ell=0}^{k-1} b_{2\ell+1}^{2k} \partial_x^{2k-2\ell-1} \mcD^A_{2\ell+1}(\phi_j,\psi_i) \\
&=&\sum_{i,j} \sum_{\ell=0}^{k-1}  b_{2\ell+1}^{2k}(-1)^{2k-2\ell-1} \int \left(\partial_x^{2k-2\ell-1}M_{ij}\right) \mcD^A_{2\ell+1}(\phi_j,\psi_i) \\
&=& -\int \bpsi \cdot \left\{ \sum_{\ell=0}^{k-1} b_{2\ell+1}^{2k}  \mcP_{2\ell+1}^A[\partial_x^{2k-2\ell-1}\bM] \right\} \ \bphi
\end{eqnarray*}
and on the other hand
\begin{eqnarray*}
\int \bpsi \cdot\mcP_{2k+1}^S[\bM] \bphi &=& \sum_{i,j}\int M_{ij}\sum_{\ell=0}^k c_{2\ell}^{2k+1} \partial_x^{2k+1-2\ell} \mcD^S_{2\ell}(\phi_j,\psi_i) \\
&=&\sum_{i,j} \sum_{\ell=0}^{k} c_{2\ell}^{2k+1}(-1)^{2k+1-2\ell} \int \left(\partial_x^{2k+1-2\ell}M_{ij}\right) \mcD^S_{2\ell}(\phi_j,\psi_i) \\
&=& -\int \bpsi \cdot \left\{ \sum_{\ell=0}^{k} c_{2\ell}^{2k+1}  \mcP_{2\ell}^S[\partial_x^{2k+1-2\ell}\bM] \right\} \bphi \ .
\end{eqnarray*}
Since these relations are valid for all functions $\bphi,\bpsi$, we conclude that
\begin{equation}
 \mcP_{2k}^A[\bM]= -\sum_{\ell=0}^{k-1} b_{2\ell+1}^{2k}  \mcP_{2\ell+1}^A[\partial_x^{2k-2\ell-1}\bM]
\end{equation}
and 
\begin{equation}
\mcP_{2k+1}^S[\bM] = -\sum_{\ell=0}^{k} c_{2\ell}^{2k+1}  \mcP_{2\ell}^S[\partial_x^{2k+1-2\ell}\bM] \ ,
\end{equation}
from which we deduce, in turn, formulas~\eqref{type_2_decompos} \&~\eqref{type_3_decompos} that we wanted to prove.

\newpage
\section{Phenomenology of the basis of vorticities for $d_1=1$}
 In this appendix we successively study the phenomenology associated with elements of the antisymmetric and symmetric subfamilies by using the vorticity dynamics~\eqref{Vortex_dynamics}, restricting to the case where the diffusion operator $\bD$ is the identity.

\subsection{The antisymmetric family}
\label{app:alphaFamily}

We first define
\begin{equation}
U \equiv \begin{bmatrix} \rho_1 \cos(\theta_1) \\  \rho_2 \cos(\theta_2) \end{bmatrix} \quad \text{and} \quad V \equiv \begin{bmatrix} -\rho_2 \sin(\theta_2) \\  \rho_1 \sin(\theta_1) \end{bmatrix} \ .
\end{equation}
In what follows, for any vector $Z=(Z_1,Z_2)^\top \in\mbR^2$, we denote by $Z^\perp$ the vector $Z$ rotated by $\pi/2$ clockwise, \ie $Z^\perp=(-Z_2,Z_1)^\top$. Then, using trigonometric identities, we decompose $\delta\bphi_k(t=0)$ as a superposition of circularly polarised harmonics:
\begin{eqnarray}
\delta\bphi_k(0) &\equiv &  \begin{bmatrix} \rho_1 \cos(kx+\theta_1) \\  \rho_2 \cos(kx+\theta_2) \end{bmatrix} \\
&=& \begin{bmatrix} \rho_1(\cos kx \cos \theta_1 - \sin kx\sin \theta_1) \\ \rho_2(\cos kx \cos \theta_2 - \sin kx\sin \theta_2)  \end{bmatrix} \\
&=& U\cos kx + V^\perp \sin kx \\
&=& \frac{U+V}{2}\cos kx + \left(\frac{U+V}{2}\right)^\perp\sin kx \\
& & + \frac{U-V}{2}\cos kx - \left(\frac{U-V}{2}\right)^\perp\sin kx \ .
\end{eqnarray}
But for $\bW_k$ given by eq.~\eqref{Omega_alpha_family}, for any vector $Z\in\mbR^2$, we have:
\begin{eqnarray}
e^{t\bW_k}(Z\cos kx \pm Z^\perp \sin kx) &=&  (\cos W_k t \cos kx \mp \sin W_k t \sin kx) Z \\
& & + (\sin W_kt \cos kx \pm \cos W_kt \sin kx) Z^\perp \\ 
&=& \cos(kx \pm W_kt) W \pm \sin(kx\pm W_kt)Z^\perp \ ,
\end{eqnarray}
\ie the operator $t\to e^{t\bW_k}$ translates in opposite directions circular waves that are polarised in opposite senses.

Hence
\begin{eqnarray}
\delta\bphi_k(t) &=& e^{t\bW_k} \delta\bphi_k(0) \\
&=& \frac{U+V}{2}\cos (kx+W_kt) + \left(\frac{U+V}{2}\right)^\perp\sin (kx+W_kt) \\
& & + \frac{U-V}{2}\cos (kx-W_kt) - \left(\frac{U-V}{2}\right)^\perp\sin (kx-W_kt) \ .
\end{eqnarray}
Finally, upon defining $\rho_\pm$ and $\theta_\pm$ the norm and angle (with respect to $\be_1=(1,0)^\top$) of $(U\pm V)/2$ respectively, the solution finally reads in the canonical basis:
\begin{eqnarray}
\delta\bphi_k(t) = \rho_+  \begin{bmatrix} \cos(kx+W_k t+\theta_+) \\  \sin(kx+W_k t+\theta_+) \end{bmatrix} + \rho_-  \begin{bmatrix} \cos(kx-W_k t-\theta_-) \\  -\sin(kx-W_k t-\theta_-) \end{bmatrix} \ ,
\end{eqnarray}
where, in particular, $\rho_\pm$ are given by
\begin{eqnarray}
\rho_\pm = \frac{1}{2}\left[ \rho_1^2 +\rho_2^2 \pm 2\rho_1\rho_2\sin(\theta_1-\theta_2)\right]^{1/2} \ .
\end{eqnarray}

\subsection{The symmetric family}
\label{app:betaFamily}
We first set 
\begin{equation}
\tbphi_k(t)\equiv \bc(t) e^{ikx} \ ,
\end{equation}
where $\bc(t) = (c_1(t),c_2(t))^\top$ and $(\be_1,\be_2)$ denotes the canonical basis of $\mbR^2$. Upon denoting by $(\lambda_1,\lambda_2)$ the eigenvalues of $\bbeta^{(\ell)}$ and $(\bb_1,\bb_2)$ the associated orthonormal eigenvectors, the solution of 
dynamics $\partial_t\tbphi_k = \bW \tbphi_k$, with $\bW=\bbeta^{(\ell)}\partial^{2\ell+1}$, reads
\begin{equation}
\tbphi_k(x,t) = e^{i(kx+W_k^jt)}(\bc(0)\cdot\bb_j)\bb_j \ , \quad \text{with} \quad W_k^j \equiv (-1)^\ell k^{2\ell+1}\lambda_j \ . 
\end{equation}
Now if we denote by $R$ the transition matrix between basis $\bb$ and $\be$, \ie $\bb_i = R_{ji}\be_j$, and assume that $c_j(0)=\rho_je^{i\theta_j}$, then 
\begin{equation}
\tbphi_k(x,t) = \rho_m e^{i(kx+W_k^jt+\theta_m)}R_{mj}\bb_j \ .
\label{beta_complex_solution}
\end{equation}
Taking the real part finally leads to the solution
\begin{equation}
\delta\bphi_k(x,t) = \rho_m \cos(kx+W_k^jt+\theta_m)R_{mj}\bb_j \ ,
\end{equation}
with initial condition $\delta\bphi_k(x,0)= [ \rho_1\cos(kx+\theta_1) ,\rho_2 \cos (kx+\theta_2) ]^\top$.

We now assume that $\bbeta^{(\ell)}$ reads
\begin{equation}
\bbeta^{(\ell)} = \beta_{12}^{(\ell)}\begin{bmatrix} 0 & 1 \\ 1  & 0 \end{bmatrix} \ .
\end{equation}
Its eigenvalues are $\lambda_1=\beta_{12}^{(\ell)},\lambda_2=-\beta_{12}^{(\ell)}$ and respective eigenvectors are $\bb_1=(1,1)^\top/\sqrt{2}$ and $\bb_2=(1,-1)^\top/\sqrt{2}$.
In this case, solution~\eqref{beta_complex_solution} is thus
\begin{eqnarray}
\tbphi_k(x,t) &=& \frac{1}{2}\left\{ \rho_1 e^{i(kx+W_k^1t+\theta_1)} + \rho_2 e^{i(kx+W_k^1t+\theta_2)} \right\} \begin{bmatrix} 1 \\ 1\end{bmatrix} \\
& & + \frac{1}{2}\left\{ \rho_1 e^{i(kx+W_k^2t+\theta_1)} - \rho_2 e^{i(kx+W_k^2t+\theta_2)} \right\} \begin{bmatrix} 1 \\ -1\end{bmatrix} \ ,
\end{eqnarray}
\ie 
\begin{eqnarray}
\tbphi_k(x,t) =  \rho^+ e^{i(kx+W_k^1t+\theta^+)}  \begin{bmatrix} 1 \\ 1\end{bmatrix} +  \rho^- e^{i(kx+W_k^2t+\theta^-)} \begin{bmatrix} 1 \\ -1\end{bmatrix} \ ,
\end{eqnarray}
with moduli $\rho^\pm$ and phases $\theta^\pm$ such that
\begin{equation}
\rho^\pm e^{i\theta^\pm} \equiv \frac{\rho_1e^{i\theta_1}\pm\rho_2e^{i\theta_2}}{2} \ .
\end{equation}
Noting that $W_k^2=-W_k^1$ and taking the real part then gives the solution
\begin{equation}
\delta\bphi_k(x,t) =  \rho^+ \cos(kx+W_k^1t+\theta^+)  \begin{bmatrix} 1 \\ 1\end{bmatrix} +  \rho^- \cos(kx-W_k^1t+\theta^-) \begin{bmatrix} 1 \\ -1\end{bmatrix} \ ,
\end{equation}
with initial condition $\delta\bphi_k(x,0)= [ \rho_1\cos(kx+\theta_1) ,\rho_2 \cos (kx+\theta_2) ]^\top$, and
where, in particular, $\rho^\pm$ are given by
\begin{equation}
\rho^\pm=\frac{1}{2} \left[ \rho_1^2 + \rho_2^2 \pm 2\rho_1\rho_2 \cos(\theta_1-\theta_2) \right]^{1/2} \ .
\end{equation}

\newpage
\section{Non-reciprocal flocking}
\label{app:NR_flock}

\paragraph{Active polar mixture.}
In~\cite{fruchart2021non}, the authors study a mixture of two species $(A,B)$ of active Brownian particles undergoing (anti-)aligning interactions in two dimensions. Starting from a microscopic dynamics, they then perform an explicit coarse-graining to obtain a hydrodynamic description for the density and polarization fields of each species $\rho_{A,B}, \bP_{A,B}$. They show that the system can be in several phases: disordered, aligned, antialigned, chiral, swap, or a mix between the latter two that they call the chiral$+$swap phase.
In order to analyse the transitions between these phases, they focus on the case where all the order-parameter fields are homogeneous in space. Their hydrodynamics then reduces to two equations over the polarity fields:
\begin{eqnarray}
\partial_t\bP_A &=& \left[ j_{AA}\rho_A -\eta - \frac{1}{2\eta}\|j_{AA}\bP_A + j_{AB}\bP_B\|^2 \right] \bP_A + j_{AB}\rho_A\bP_B  \\
\partial_t\bP_B &=& \left[ j_{BB}\rho_B -\eta - \frac{1}{2\eta}\|j_{BB}\bP_B + j_{BA}\bP_A\|^2 \right] \bP_B + j_{BA}\rho_B\bP_A \ , 
\end{eqnarray}
where $\rho_A$ and $\rho_B$ are the respective homogeneous densities of each species, $\eta$ the squared amplitude of the microscopic angular noise, and $(j_{\alpha\beta})$ the matrix of the (renormalized) alignment constants.
Still neglecting the gradient terms in the hydrodynamics originally derived in~\cite{fruchart2021non}, we add random fields to the above equations to account for fluctuations in a simple way:
\begin{eqnarray}
\partial_t\bP_A &=& \left[ j_{AA}\rho_A -\eta - \frac{1}{2\eta}\|j_{AA}\bP_A + j_{AB}\bP_B\|^2 \right] \bP_A + j_{AB}\rho_A\bP_B + \bLambda_A \label{chasingSpinsMeanFieldDyn_A} \\
\partial_t\bP_B &=& \left[ j_{BB}\rho_B -\eta - \frac{1}{2\eta}\|j_{BB}\bP_B + j_{BA}\bP_A\|^2 \right] \bP_B + j_{BA}\rho_B\bP_A + \bLambda_B \ , 
\label{chasingSpinsMeanFieldDyn_B}
\end{eqnarray}
where $\bLambda_A(\br,t)$ and $\bLambda_B(\br,t)$ are independent Gaussian fields, each having the same statistics as $\boldeta$ in eq.~\eqref{EDPS02}. Note that another possibility to implement fluctuations, (which is perhaps a little more coherent with the fact that we neglected the gradient terms of the original hydrodynamics) would have been to consider uniform noise fields, in which case the problem would have been reduced to a finite-dimensional one (see our paper I~\cite{o2024geometric} for the corresponding tools). Interestingly, the resulting finite-dimensional vorticity $\bar{\bomega}$ is related to the one studied below for non-uniform random field, $\bomega$, as follows: 
\begin{equation}
\bomega(\delta\bP_A,\delta\bP_B)=\int \bar{\bomega}(\delta\bP_A(\br),\delta\bP_B(\br))\rmd \br \ . 
\end{equation}
Since we here aim at illustrating the use of our field-theoretic methods, below we stick to the case of the non-uniform random fields of eqs.~\eqref{chasingSpinsMeanFieldDyn_A}-\eqref{chasingSpinsMeanFieldDyn_B}.


As shown in~\cite{fruchart2021non}, when $j_{AB}\neq j_{BA}$ and $\eta$ is sufficiently small, the polarizations of the two species undergo an ``angular run and chase'' dynamics, that can either take the form of a chiral phase, a swap phase, or a chiral$+$swap phase. Here we want to show that these phenomenologies can be heuristically inferred from the vorticity dynamics~\eqref{Vortex_dynamics} associated to eqs.~\eqref{chasingSpinsMeanFieldDyn_A}-\eqref{chasingSpinsMeanFieldDyn_B}.

As the operator that is denoted by $\bb$ in the generic dynamics~\eqref{EDPS02} is here the identity, the diffusion operator $\bD=\bb\bb^\dagger$ is also the identity operator.
Let us now compute the vorticity two-form $\bomega \equiv\mbd \bD^{-1}\ba$ of the dynamics, where $\ba(\br,[\bP_A,\bP_B])$ is a
 $\mbR^{4}$-valued field over $\mbR^{2}$ obtained by stacking together the deterministic parts on the right-hand side of
  eqs.~\eqref{chasingSpinsMeanFieldDyn_A} and~\eqref{chasingSpinsMeanFieldDyn_B}.

\begin{eqnarray*}
\bomega &=& \sum_{\alpha,\gamma=A,B} \sum_{i,k=1}^2 \int \rmd \br \rmd \br' \Bigg\{ \Big[ j_{\alpha\alpha}\rho_\alpha-\eta-\frac{1}{2\eta}\Big| \sum_{\beta=A,B} j_{\alpha\beta} \bP_\beta^\br\Big|^2 \Big] \frac{\delta P_\alpha^{i\br}}{\delta P_\gamma^{k\br'}}\delta^{\gamma k \br'} \wedge \delta^{\alpha i \br} \\
& & -\frac{1}{\eta} \sum_{l=1}^2\left[ \sum_{\beta=A,B} j_{\alpha\beta} P_\beta^{l\br} \right] \left[ \sum_{\mu=A,B} j_{\alpha\mu} \frac{\delta P_\mu^{l\br}}{\delta P_\gamma^{k\br'}}\right] P^{i\br}_\alpha \delta^{\gamma k\br'}\wedge \delta^{\alpha i \br} \Bigg\} \\
& & + \sum_{\gamma=A,B}\sum_{i,k=1}^2 \int \rmd \br \rmd \br'\left[ j_{AB}\rho_A \frac{\delta P_B^{i\br}}{\delta P_\gamma^{k\br'}}\delta^{\gamma k \br'} \wedge \delta^{A i\br} + j_{BA}\rho_B \frac{\delta P_A^{i\br}}{\delta P_\gamma^{k\br'}}\delta^{\gamma k\br'}\wedge \delta^{B i\br}\right] \\
&=& -\frac{1}{\eta} \sum_{\alpha,\gamma} \sum_{i,l} \int \rmd \br \left[ \sum_\beta j_{\alpha\beta} P_\beta^{l\br} \right]  j_{\alpha\gamma}  P^{i\br}_\alpha \delta^{\gamma l\br}\wedge \delta^{\alpha i \br}  \\
& & + \sum_{i} \int \rmd \br \left[ j_{AB}\rho_A \delta^{B i \br} \wedge \delta^{A i\br} + j_{BA}\rho_B \delta^{A i\br}\wedge \delta^{B i\br}\right]
\end{eqnarray*}
where we used $\frac{\delta P_\alpha^{i\br}}{\delta P_\gamma^{k\br'}}=\delta_{\alpha,\gamma}\delta^{i,k}\delta(\br-\br')$ together with 
$\delta^{\alpha i \br} \wedge \delta^{\alpha i \br}=0$, the latter stemming from the antisymmetry of the wedge product.
As we now have integrals over $\br$ only, we can safely drop the variable $\br$ everywhere to lighten notations.
Using again the antisymmetry of the wedge product, we get
\begin{eqnarray*}
\sum_{\alpha,\gamma}\sum_{i,l} j_{\alpha\beta}j_{\alpha\gamma} P_\beta^l P_\alpha^i \delta^{\gamma l} \wedge \delta^{\alpha i} = \frac{1}{2}\sum_{\alpha,\gamma}\sum_{i,l} \left[j_{\alpha\beta}j_{\alpha\gamma} P_\beta^l P_\alpha^i- j_{\gamma\beta}j_{\gamma\alpha} P_\beta^i P_\gamma^l \right] \delta^{\gamma l} \wedge \delta^{\alpha i} \ ,
\end{eqnarray*} 
so that the vorticity two-form reads
\begin{eqnarray*}
\bomega &=& -\frac{1}{2\eta} \sum_{\alpha,\gamma,\beta} \sum_{i,l} \int \left[j_{\alpha\beta}j_{\alpha\gamma} P_\beta^l P_\alpha^i- j_{\gamma\beta}j_{\gamma\alpha} P_\beta^i P_\gamma^l \right] \delta^{\gamma l}\wedge \delta^{\alpha i }  \\
& & + \sum_{i} \int  \left[ j_{AB}\rho_A \delta^{B i \br} \wedge \delta^{A i\br} + j_{BA}\rho_B \delta^{A i}\wedge \delta^{B i}\right] \ .
\end{eqnarray*}
Expanding the sums over $\alpha$ and $\gamma$, grouping the terms where $\alpha=\gamma$ and those where $\alpha\neq \gamma$, and using the antisymmetry of the wedge product, we get, after some rearrangement:
\begin{eqnarray*}
\bomega &=&\frac{1}{\eta}\int \Det(\bP_A,\bP_B) \left[ j_{AB}j_{AA}\delta^{A1}\wedge\delta^{A2} -j_{BA}j_{BB}\delta^{B1}\wedge\delta^{B2} \right] \\
& & +\frac{1}{\eta} \sum_{i, l}\int  \left[j_{AA}j_{AB}P_A^i P_A^l  +(j_{AB}^2 - j_{BA}^2) P_A^i P_B^l  - j_{BB} j_{BA} P_B^i P_B^l \right] \delta^{Ai}\wedge \delta^{B l} \\
& & - \sum_i \int \left[ j_{AB}\rho_A - j_{BA}\rho_B\right] \delta^{Ai}\wedge\delta^{Bi}
\end{eqnarray*}
We first note that the vorticity two-form of dynamics~\eqref{chasingSpinsMeanFieldDyn_A}-\eqref{chasingSpinsMeanFieldDyn_B} is entirely generated by the antisymmetric subfamily of~\eqref{1dBasis2Forms} of order $\ell=0$.
To simplify the analysis of the phenomenology associated with each component, we now assume that both self-interaction terms are equal, $j_{AA}=j_{BB}=j_0$, that the cross interaction is purely antisymmetric, $j_{AB}=-j_{BA}=j_-$, and that the species densities are equal, $\rho_A=\rho_B=\rho_0$. The vorticity then reduces to :
\begin{eqnarray}
\bomega &=&\frac{1}{\eta}\int \Det(\bP_A,\bP_B) j_0 j_- \left[ \delta^{A1}\wedge\delta^{A2} +\delta^{B1}\wedge\delta^{B2} \right] \notag \\
& & +\frac{1}{\eta} \sum_{i, l}\int   j_0 j_-\left[P_A^i P_A^l +  P_B^i P_B^l \right] \delta^{Ai}\wedge \delta^{B l} \notag \\
& & - \sum_i \int 2\rho_0 j_- \delta^{Ai}\wedge\delta^{Bi} \ .
\label{NRflock_vorticityForm}
\end{eqnarray}
Let us decompose the vorticity two-form as $\bomega=\bomega_1+\bomega_2+\bomega_3$, where $\bomega_i$ stands for the $i^{th}$ line on the right-hand side of eq.~\eqref{NRflock_vorticityForm}.
Using the definition~\eqref{def_affinity_operator} of $\whomega$ together with relation~\eqref{Omega_vs_whomega}, $\bW=-\bD\whomega/2$, the vorticity operator, which acts on $\delta\bphi\equiv (\delta P_A^1,\delta P_A^2,\delta P_B^1,\delta P_B^2)^\top$ by multiplication, is readily shown to read $\bW=\bW_1+\bW_2+\bW_3$, with:

\begin{eqnarray}
\bW_1 = \frac{j_0 j_-}{2\eta}\Det(\bP_A,\bP_B)
\begin{bmatrix}
0 & -1 & 0 & 0 \\
1 & 0 & 0 & 0 \\
0 & 0 & 0 & -1 \\
0 & 0 & 1 & 0 \\
\end{bmatrix} \ ,
\end{eqnarray}
\begin{eqnarray}
\bW_2 = \frac{j_0 j_-}{2\eta}
\begin{bmatrix}
0 & -(\bP_A\bP_A^\top + \bP_B\bP_B^\top ) \\
\bP_A\bP_A^\top + \bP_B\bP_B^\top & 0  \\
\end{bmatrix} \ ,
\end{eqnarray}
\begin{eqnarray}
\bW_3 = \rho_0j_-
\begin{bmatrix}
0 & 0 & 1 & 0 \\
0 & 0 & 0 & 1 \\
-1 & 0 & 0 & 0 \\
0 & -1 & 0 & 0 \\
\end{bmatrix} \ ,
\end{eqnarray}
the operator $\bW_2$ being defined by two-by-two blocks.

The contribution of $\bW_1$ in the vortex dynamics~\eqref{Vortex_dynamics} is clearly responsible for the \textit{chiral phase} reported in~\cite{fruchart2021non}: for instance if $j_->0$ then the polarization of species A ``runs'' after that of species B which itself flees away, hence spontaneously breaking spatial parity in $\mbR^2$ with an orientation that depends on the initial angle between the two polarities. This collective rotation of $\bP_A$ and $\bP_B$ would occur at the angular speed $j_0  j_-\Det(\bP_A,\bP_B)/2\eta$, if $\bW_3$ was the only term in the vorticity.

To interpret $\bW_1$ and $\bW_2$, let us consider a vector $\bu$ in $\mbR^2$, and build from it the two orthogonal vectors of $\mbR^4$ defined as $\bU \equiv (\bu,\bu)^\top$ and $\bV \equiv (-\bu,\bu)^\top$.	
The operator that generates rotations in the plane $(\bU,\bV)$, at angular speed $|\bU||\bV|=2|\bu|^2$, and in the direction $\bU\to\bV$, reads
\begin{equation}
\bR\equiv \bV\bU^\top - \bU\bV^\top = 2
\begin{bmatrix}
0 & -\bu \bu^\top \\
\bu \bu^\top & 0 \\
\end{bmatrix} \ .
\end{equation}
The dynamics $\partial_t\delta\bphi = \bR \delta\bphi$, where $\bphi=(\delta\bP_A^\top,\delta\bP_B^\top)$, follows a circular orbit in $\mbR^4$ which start at $\delta\bphi(0)$ and passes, at a time equal to a quarter of its period, by the point $\delta\bphi(T/4)$. In the latter, the projections of $\delta\bP_A(T/4)$ and $\delta\bP_B(T/4)$ along $u^\perp$ are the same as initially (they are invariant under the dynamics) while the their projections on $\bu$ are given by $-\delta\bP_B(0)$ and $\delta\bP_A(0)$, respectively. This phenomenology closely resembles that of the \textit{swap phase} reported in~\cite{fruchart2021non}, and the operator $\bW_2$ (respectively $\bW_3$) is a superposition of two such rotation generators respectively along $\bP_A$ and $\bP_B$ (along $(1,0)^\top$ and $(0,1)^\top$, respectively).

\section{Active Ising model}
\label{app:active_Ising}

In this appendix, we compute the vorticity two-form and operator of the active Ising model (AIM).

We introduce slightly different notation here: the functional one-forms $\partial^k \delta^{mx}$ are now denoted by $(-1)^k\partial^k\delta m^x$, and similarly for the field $\rho$. This allows getting rid of the powers of $-1$ that appear when using derivatives of Dirac deltas. Indeed, if $\delta\bphi\equiv(\delta \rho,\delta m)$ is a perturbation, we simply have  
\begin{equation}
\partial^k \delta m^x (\delta\bphi) = \partial^k \delta m^x \ ,
\end{equation}
rather than
\begin{equation}
\partial^k\delta^{m x} (\delta\bphi) = (-1)^k \partial^k \delta m^x \ .
\end{equation}
The main risk is now to make the confusion between functional forms and perturbations. For instance $\delta\rho^x\wedge\delta m^x$ is not some sort of cross product between the two components of a perturbation $\delta\bphi=(\delta\rho,\delta m)$, but a bilinear map that associates to a pair of fluctuations $\delta\bphi_1=(\delta\rho_1,\delta m_1)$, $\delta\bphi_2=(\delta\rho_2,\delta m_2)$ the number 
$\delta\rho_1^x \delta m_2^x - \delta\rho_2^x \delta m_1^x$.
Finally, we adopt the generalized \textit{Einstein convention} that consist in integrating over repeated continuous indices, \textit{with any subscript of a partial symbol $\partial$ being excluded from this convention}, as we only keep them in the calculation to remember with respect to which variable we take derivatives.

Let us now turn to the AIM~\eqref{AIM_dynamics_rho}-\eqref{AIM_dynamics_m}. The noise being additive, the spurious drift vanishes, so that the drift $\ba$ of the generic dynamics~\eqref{EDPS02} reads in this case
\begin{equation}
\ba = 
\begin{bmatrix}
D\partial_x^2 \rho - v\partial_x m \\
D \partial_x^2 m -v\partial_x \rho + \gamma_1 m + \gamma_2 m^3
\end{bmatrix} \ ,
\end{equation}
while the diffusion operator is
\begin{equation}
\bD =
\begin{bmatrix}
-D\partial_x^2 & 0 \\
0 & D \\
\end{bmatrix} \ .
\end{equation}
Denoting by $G_{xy}\equiv G(x-y)$ the Green function of the 1d Laplacian, \ie $\partial_x^2 G(x-y)=\delta(x-y)$, the one form $\bD^{-1}\ba$ reads
\begin{eqnarray}
\bD^{-1}\ba = \left(-\rho_x + \frac{v}{D}G_{xz}\partial_z m_z\right)\delta\rho^x  + \left( \partial^2_x m_x -\frac{v}{D}\partial_x\rho_x + \frac{\gamma_1}{D}m_x +\frac{\gamma_2}{D}m_x^3 \right) \delta m^x \ .
\end{eqnarray}
We now apply formula~\eqref{extDer_convenient} to compute the vorticity two-form $\bomega\equiv \bD^{-1}\ba$ of the AIM:
\begin{eqnarray*}
\bomega &=& \mbd\left(-\rho_x + \frac{v}{D}G_{xz}\partial_z m_z\right)\wedge \delta\rho^x  + \mbd\left( \partial^2_x m_x -\frac{v}{D}\partial_x\rho_x + \frac{\gamma_1}{D}m_x +\frac{\gamma_2}{D}m_x^3\right)\wedge \delta m^x \\
&=& \frac{\delta}{\delta\rho^y}\left(-\rho_x + \frac{v}{D}G_{xz}\partial_z m_z\right)\delta\rho^y \wedge \delta\rho^x  + \frac{\delta}{\delta m^y}\left(-\rho_x + \frac{v}{D}G_{xz}\partial_z m_z\right)\delta m^y \wedge\delta\rho^x \\
& & + \frac{\delta}{\delta\rho_y}\left( \partial^2_x m_x -\frac{v}{D}\partial_x\rho_x + \frac{\gamma_1}{D}m_x +\frac{\gamma_2}{D} m_x^3 \right)\delta\rho^y\wedge \delta m^x \\
& & + \frac{\delta}{\delta m^y}\left( \partial^2_x m_x -\frac{v}{D}\partial_x\rho_x + \frac{\gamma_1}{D}m_x +\frac{\gamma_2}{D}m_x^3 \right)\delta m^y\wedge \delta m^x \ .
\end{eqnarray*}
Computing the functional derivatives and using the fact that $\delta\rho^x\wedge \delta\rho^x =\delta m^x \wedge \delta m^x =0$ because the wedge product is antisymmetric, we get
\begin{eqnarray}
\bomega &=& \frac{v}{D}G_{xz}\partial_z\delta(y-z) \delta m^y \wedge \delta \rho^x + [\partial_x^2\delta(x-y)]\delta m^y\wedge \delta m^x \\
& & -\frac{v}{D}[\partial_x\delta(x-y)] \delta\rho^y \wedge \delta m^x +\frac{\gamma_2'}{D} m_x^3\delta(x-y) \delta\rho^y\wedge \delta m^x \ ,
\end{eqnarray}
where $\gamma_2'\equiv \frac{d}{d\rho}\gamma_2$.
Proceeding to several integrations by parts and integrating out the Dirac deltas, we finally get:
\begin{eqnarray}
\bomega = \frac{v}{D}\delta\rho^x \wedge \partial\delta m^x - \frac{v}{D} G_{xy} \delta\rho^x \wedge \partial\delta m^y +\frac{\gamma_2'}{D}m_x^3 \delta\rho^x \wedge \delta m^x \ .
\label{appdx_AIM_vort2form}
\end{eqnarray} 
In order to get the associated vorticity operator $\bW$, we first determine the cycle affinity operator $\whomega$, defined as~\eqref{def_affinity_operator}.
To this purpose, we apply $\bomega$ to a pair $\delta\bphi_1,\delta\bphi_2$ of perturbations, integrate by parts every derivative that applies to $\delta\bphi_1$, and change integration variables. We can then factorise $\delta\phi_1$ and identify $\whomega$ as the resulting operator that acts on $\delta\bphi_2$ only:
\begin{eqnarray*}
\bomega(\delta\bphi_1,\delta\bphi_2) &=& \frac{v}{D}\left(\delta\rho_1^x \partial\delta m_2^x - \delta\rho_2^x \partial\delta m_1^x \right) - \frac{v}{D} G_{xy} \left(\delta\rho_1^x \partial\delta m_2^y - \delta\rho_2^x \partial\delta m_1^y \right)  +\frac{\gamma_2'}{D}m_x^3 \left[\delta\rho^x_1 \delta m^x_2 - \delta\rho^x_2  \delta m^x_1 \right] \\
&=& \delta\rho^x_1 \left[ \frac{v}{D}(\partial\delta m_2^x -G_{xy}\partial m_2^y) +\frac{\gamma_2'}{D}m_x^3 \delta m^x_2 \right] \\
& & + \delta m_1^x \left[\frac{v}{D}\partial\delta\rho_2^x - \frac{v}{D}\partial_x G_{yx}\delta\rho_2^y -\frac{\gamma_2'}{D}m_x^3 \delta\rho_2^x \right] \ .
\end{eqnarray*}
Hence, the cycle-affinity operator, that applies to a perturbation $\delta\bphi(y)$, reads
\begin{equation}
\whomega = 
\begin{bmatrix}
0 & \frac{v}{D}\left(\delta(x-y) - G_{xy}\right)\partial_y +\frac{\gamma_2'}{D}m_x^3\delta(x-y) \\
\frac{v}{D}\partial_x\delta(x-y) - \frac{v}{D}\partial_x G_{yx} -\frac{\gamma_2'}{D}m_x^3\delta(x-y)  & 0 \\
\end{bmatrix} 
\end{equation}
and is antisymmetric as expected.
Consequently, the vorticity operator, which satisfies $\bW=-\bD\whomega/2$, is given by
\begin{equation}
\bW = 
\frac{1}{2}\begin{bmatrix}
0 & v\delta(x-y)\left(\partial_y^3 -\partial_y\right) + \partial_x^2\gamma_2'm_x^3\delta(x-y) \\
-v\partial_x\delta(x-y) + v\partial_x G_{yx} +\gamma_2'm_x^3\delta(x-y)  & 0 \\
\end{bmatrix} \ ,
\end{equation}
where we used the fact that $\partial_x^2 G_{xy} = \delta(x-y)$.

In the original coarse-grained Active Ising Model, $\gamma_2$ depends on $\rho(x)$~\cite{solon2015flocking}. We see that this dependence generates a term in $\bomega$ along the antisymmetric family (the last one on the right-hand side of~\eqref{appdx_AIM_vort2form}), which does not come from the self-propulsion but from the aligning mechanism. As we are currently focusing on the symmetric family and self-propulsion effects, we assume from now on that $\gamma_2$ is instead a constant.
The vorticity operator thus reduces to
\begin{equation}
\bW = 
\frac{1}{2}\begin{bmatrix}
0 & v\delta(x-y)\left(\partial_y^3 -\partial_y\right)  \\
-v\partial_x\delta(x-y) + v\partial_x G_{yx}   & 0 \\
\end{bmatrix} \ ,
\end{equation}
%
In Fourier space, $\bW$ acts by multiplication by the matrix $\bW_k$ given by:
\begin{equation}
\bW_k = \frac{iv}{2}\begin{bmatrix}
0 &  k^3 + k \\
k + 1/k & 0 \\
\end{bmatrix}
\end{equation}
where we used the facts that the Fourier transform\footnote{We choose the convention $FT[f](k)=\hat{f}(k)=\int^\infty_{-\infty} f(x)e^{ikx}d x$ and $FT^{-1}[\hat{f}](x) = \frac{1}{2\pi}\int^\infty_{-\infty}\hat{f}(k)e^{-ikx}dk$, where $FT$ and $FT^{-1}$ respectively stand for the Fourier transform and its inverse.} turns $G(x)$ into $G_k = -1/k^2$ and a convolution into an usual product.
The vectors $\bv_{\pm}\equiv (\pm k/\sqrt{1+k^2},1\sqrt{1+k^2})^\top$ are eigenvectors of $\bW_k$ for the eigenvalues $\lambda_{\pm}=\pm iv(1+k^2)/2$. It follows that the solution of the vorticity dynamics~\eqref{Vortex_dynamics} in Fourier space is given by
\begin{eqnarray}
\delta\bphi_k(t) = \langle \delta\bphi_k(0)|\bv_+\rangle e^{\lambda_+ t} \bv_+ + \langle \delta\bphi_k(0)|\bv_-\rangle e^{\lambda_- t} \bv_- \ ,
\end{eqnarray}
where $\langle\cdot |\cdot \rangle$ denotes the canonical scalar product of $\mathbb{C}^2$. Injecting the expressions of $\bv_{\pm}$ and $\lambda_\pm$ and rearranging the terms, we get
\begin{equation}
\delta \bphi_k(t) = 
\begin{bmatrix}
\frac{2 k^2}{1+k^2}\cos\left(\frac{vt(1+k^2)}{2}\right)\delta\rho_k(0) + \frac{2 i k}{1+k^2}\sin\left(\frac{vt(1+k^2)}{2}\right)\delta m_k(0) \\
\frac{2 ik}{1+k^2}\sin\left(\frac{vt(1+k^2)}{2}\right)\delta\rho_k(0) + \frac{2}{1+k^2}\cos\left(\frac{vt(1+k^2)}{2}\right)\delta m_k(0) 
\end{bmatrix} \ .
\end{equation}
We now take the inverse Fourier transform to get the real space solution of dynamics~\eqref{Vortex_dynamics}:
\begin{equation}
\delta\bphi(t) = 
\begin{bmatrix}
-K'' \ast f_{CF} \ast \delta\rho(0) -K' \ast f_{SF}(t) \ast \delta m(0) \\
-K' \ast f_{SF} \ast \delta\rho(0) +K \ast f_{CF}(t) \ast \delta m(0)
\end{bmatrix} \ ,
\label{solution_vortex_AIM}
\end{equation}
where $\ast$ stands for the convolution and the functions $f_{CF}(x,t)$ and $f_{SF}(x,t)$, which are defined as
\begin{equation}
f_{CF}(x,t) \equiv \frac{1}{\sqrt{2\pi v t}}\sin \left(\frac{x^2}{2vt}-\frac{vt}{2}+\frac{\pi}{4}\right)
\end{equation}
and 
\begin{equation}
f_{SF}(x,t) \equiv \frac{1}{\sqrt{2\pi v t}}\cos \left(\frac{x^2}{2vt}-\frac{vt}{2}+\frac{\pi}{4}\right) \ ,
\end{equation}
are the inverse Fourier transform of $\cos(vt[1+k^2]/2)$ and $\sin(vt[1+k^2]/2)$, respectively.
In equation~\eqref{solution_vortex_AIM}, $K(x)\equiv e^{-|x|}$ is the inverse Fourier transform of $\frac{2}{1+k^2}$ (\ie twice the Green function of the Screened Poisson equation, with unit screening constant), and $K'$ and $K''$ are its derivatives.


As our main concern is whether the phenomenology of this dynamics, despite having a non-local vorticity operator, is akin to that of the symmetric family for the local case as defined in section~\ref{subsubsec:oddSubspace}, we consider an initial perturbation 
 whose components are harmonic waves of the same spatial frequency $k$ which are either in phase or in phase opposition :
\begin{equation}
\delta\bphi(t=0)=\begin{bmatrix}
\delta \rho (0) \\
\delta m(0)
\end{bmatrix}=
A \cos (kx +\theta)
\begin{bmatrix}
1 \\
\varepsilon
\end{bmatrix} \ ,
\end{equation}
where $\varepsilon=1$ ($\varepsilon=-1$) when the components are in phase (in phase opposition).
In this case, the solution~\eqref{solution_vortex_AIM} of the vorticity dynamics can be shown to read
\begin{equation}
\delta\bphi(x,t) = \delta\bphi_{\rm prop}(x,t) + \delta\bphi_{\rm stat}(x,t) \ ,
\end{equation} 
with
\begin{eqnarray}
\delta\bphi_{\rm prop}(x,t) = \frac{2}{1+k^2}\cos\left(kx + \varepsilon \frac{v[1+k^2]}{2}t + \theta \right)
\begin{bmatrix}
\varepsilon k \\
1
\end{bmatrix} \ ,
\end{eqnarray} 
and
\begin{eqnarray}
\delta\bphi_{\rm stat}(x,t) = \frac{2(k-1)\varepsilon}{1+k^2}
\begin{bmatrix}
k \cos\left(\frac{v[1+k^2]}{2}t\right)\cos\left(kx+\theta\right) \\
\sin\left(\frac{v[1+k^2]}{2}t\right)\sin\left(kx+\theta\right)
\end{bmatrix} \ .
\end{eqnarray}
The component $\delta\bphi_{\rm stat}$ of $\delta\bphi(x,t)$ is always a standing wave, while $\delta\bphi_{\rm prop}$ propagates either leftward or rightward depending on $\varepsilon=\pm 1$, \ie on whether the initial density and polarity perturbations are in phase or in antiphase. If we were to play a video of the time evolution of the graph of $\delta\bphi(x,t)$, we would see a traveling wave, whose direction of propagation is the same as that of $\delta\bphi_{\rm prop}$, and whose amplitude and speed seems to be modulated by its stationary component $\delta\bphi_{\rm stat}$. Thus, although non-local, the vorticity operator generates a flow which is phenomenologically very close to that of the symmetric family.

%
%
%

\newpage
\section{Loop--wise entropy production in the phase--separated AMB}
\label{App:sec:loopwiseEnt}

In this appendix, we show that the entropy produced along the oriented loop $\mcC=\partial\mcS\subset\mbF$, where $\mcS$ is parametrized by $R:[0,1]\times[0,2\pi/|w|]\to \mbF$ whose expression is given in eq.~\eqref{AMBsurfaceParam}, takes the form~\eqref{loopwiseAMB}.
Let us denote $s=[0,1]\times[0,2\pi/|w|]$. We then use eq.~\eqref{loopWiseEntProd02} to compute the entropy production of dynamics~\eqref{AMB} along $\partial\mcS$ \via the integral of $\bomega$ over $\mcS=R(s)$:
\begin{eqnarray}
\whSigma[\partial\mcS]=\int_\mcS \bomega = \int_{s} R^*\bomega = \int_{s}  \bomega_{R(\tau,t)} (\partial_\tau R,\partial_t R) \rmd \tau \rmd t 
\end{eqnarray}
where $R^*$ stands for the pullback by $R$ (\ie change of variables) and the last integral is a usual two-dimensional integral on $s$. Using the expression~\eqref{cycleAffAMB} of the cycle affinity of AMB, we get
\begin{eqnarray*}
\whSigma[\partial\mcS] &=& \int_{s\times\mbR_{<0}} \rmd x \rmd \tau \rmd t \left[(2\lambda+\kappa')\partial_x\rho \right]\,\partial_x\delta^x\wedge\delta^x(\partial_\tau \rho,\partial_t \rho) \\
&=&\int_{s\times\mbR_{<0}} \rmd x \rmd \tau \rmd t \,(2\lambda+\kappa')\partial_x\rho \left[-(\partial_x\partial_\tau\rho)(\partial_t \rho)  + (\partial_\tau \rho)(\partial_x\partial_t \rho) \right]  \ .
\end{eqnarray*}
But the derivatives of $\rho(\tau,t,x)$ with respect to $\tau$ and $t$ respectively read
\begin{equation}
\partial_\tau\rho(\tau,t,x)= A(x)\cos(kx-wt) \ ,
\end{equation}
and 
\begin{equation}
\partial_t\rho(\tau,t,x)=A(x)\tau w\sin(kx-wt) \ .
\end{equation}
In turn, the entropy production along $\partial\mcS$ is
\begin{eqnarray*}
\whSigma[\partial\mcS] &=& \int_{s\times\mbR_{<0}} \rmd x \rmd \tau \rmd t \,(2\lambda+\kappa')\Big[ \partial_x\rhoss - A\tau k\sin(kx-wt) + A'\tau\cos(kx-wt) \Big] \\
& & \times \Big[ -A'A\tau w \cos(kx-wt)\sin(kx-wt) + A^2\tau w k \sin^2(kx-wt) \\
& & + A'A\tau w \cos(kx-wt)\sin(kx-wt)+A^2\tau w k \cos^2(kx-wt) \Big]  \\
&=& \int_{s\times\mbR_{<0}} \rmd x \rmd \tau \rmd t \,(2\lambda+\kappa')\Big[ \partial_x\rhoss - A\tau k\sin(kx-wt) + A'\tau\cos(kx-wt) \Big] A^2 \tau w k   \\
&=& wk\frac{2\pi}{|w|}\frac{1}{2}\int_{\mbR_{<0}} \rmd x (2\lambda+\kappa')A^2\partial_x\rhoss  - \int_{[0,1]\times \mbR_{<0}}A^3\tau^2 w k^2 (2\lambda+\kappa')\partial_x\rhoss \int_0^{2\pi/|w|} \rmd t \sin(kx-wt) \\
& & + \int_{[0,1]\times\mbR_{<0}} A'A^2\tau^2 w k \int_0^{2\pi/|w|}\rmd t \cos(kx-wt) \\
&=& \mathrm{sgn}(w)k\pi\int_{\mbR_{<0}} \rmd x (2\lambda+\kappa')A^2\partial_x\rhoss \ ,
\end{eqnarray*}
where we have just expanded the various products and integrated the sinusoidal functions over one period $2\pi/|w|$ of the variable $t$.
If we now assume that $(2\lambda+\kappa')$ is constant, we finally get
\begin{equation}
\whSigma[\partial\mcS]=(2\lambda+\kappa')\mathrm{sgn}(w)k\pi\int_{\mbR_{<0}} \rmd x A^2\partial_x\rhoss \ .
\end{equation}
This is the result used in section~\eqref{subsubsec:loopWiseEP_AMB}.

\newpage
\section{Decomposition of the current, hidden current, and hidden IEPR}
\label{app:hidden_current}

\subsection{Joint IEPR of the field and its current}
\label{appsubsec:jointIEPR}
For a given solution $(\bphi_t,\bj_t)_t$ of dynamics~\eqref{EDPS_current01}-\eqref{EDPS_current02}, we introduce the $\mbR^{d_4}$-valued field:
\begin{equation}
\bq_t\equiv \bq_0 + \int_0^t \bj_s \rmd s \ ,
\end{equation}
with $\bq_0$ such that $\bB\bq_0=\bphi_0$. Time integrating eq.~\eqref{EDPS_current01} gives 
\begin{equation}
\bphi_t =\bB \bq_t \ ,
\label{phi_vs_Bq}
\end{equation}
which allows to see $\bj, \bJ, \mks$, and $\mkh$ directly as functionals of $\bq$. In turn, this allows the rewriting of eq.~\eqref{EDPS_current01}-\eqref{EDPS_current02} as a closed random dynamics over the field $\bq$:
\begin{equation}
\partial_t \bq = \bJ + \mks + \mkh + \bbb \boldeta \ .
\label{EDPS_q}
\end{equation}
The crucial point is now that $\mks$ and $\mkh_{\blambda}$ respectively coincide with the Stratonovitch-spurious drift and the $\blambda$-gauge associated with the dynamics of~\eqref{EDPS_q}. Indeed, adapting the general formula~\eqref{functionalSpuriousDrift} to dynamics~\eqref{EDPS_q} gives the Stratonovitch spurious drift:
\begin{eqnarray*}
\bbb^{i_1\br_1}_{i_3\br_3} \frac{\delta}{\delta q^{i_2\br_2}} \bbb^{i_2\br_2}_{i_3\br_3} &=& \bbb^{i_1\br_1}_{i_3\br_3} \frac{\delta \phi^{i_4\br_4}}{\delta q^{i_2\br_2}}\frac{\delta}{\delta \phi^{i_4\br_4}} \bbb^{i_2\br_2}_{i_3\br_3} \\
&=& \bbb^{i_1\br_1}_{i_3\br_3} B^{i_4\br_4}_{i_2\br_2}\frac{\delta}{\delta \phi^{i_4\br_4}} \bbb^{i_2\br_2}_{i_3\br_3} \\
&=& \bbb^{i_1\br_1}_{i_3\br_3} \frac{\delta}{\delta \phi^{i_4\br_4}} b^{i_4\br_4}_{i_3\br_3}  \\
&=& \mks^{i_1\br_1} \ ,
\end{eqnarray*}
where we used the chain rule to get the first equality, relation~\eqref{phi_vs_Bq} to go from the first to the second line, relation~\eqref{b_decomp} together with the fact that $\bB$ is independent of $\bphi$ to get to the third, and the definition of $\mks$ to conclude. A very similar calculation also shows the coincidence between $\mkh_{\blambda}$ and the $\blambda$-gauge drift of dynamics~\eqref{EDPS_q}.

We can now deduce that the joint (Stratonovitch) path-probability $\mcP[(\bphi_t,\bq_t)_{t\in\mbT}]$ reads
\begin{eqnarray*}
\mcP[(\bphi_t,\bq_t)_t] &=& \mcP[(\bphi_t)_t | (\bq_t)_t] \mcP[(\bq_t)_t] \\
&\propto & \delta[(\bphi_t-\bB \bq_t)_t]\exp\left\{\frac{-1}{4}\int_0^\mcT (\partial_t\bq_t -\bJ_t)\left[\bbb \bbb^\dagger\right]^{-1}  (\partial_t\bq_t -\bJ_t) \rmd t \right\} \ ,
\end{eqnarray*}
up to a multiplicative factor that will be irrelevant to compute the IEPR. Note that, in the last equation, the dagger stands for the adjoint with respect to the canonical $L^2$ scalar product. Making the change of variable $(\bq_t)_{t\in\mbT}\to(\bj_t=\partial_t\bq_t)_{t\in\mbT}$ in the path integral gives the joint path-probability\footnote{The rigorous calculation requires a temporal discretization, and the change of variables is then between the variables $(\bj_{ndt})_n$ and $(\bq_{(n+1/2)dt})_n$, with $j_{ndt}\equiv (q_{(n+1/2)dt}-q_{(n-1/2)dt})/dt$. This change of variables make a Jacobian (that is a power of $dt$) and boundary terms appear, all of which vanish when considering the ratio with the time-reversed path and the long-time limit in the calculation of $\sigma_{\bphi,\bj}$.}
\begin{equation}
\mcP[(\bphi_t,\bj_t)_{t\in\mbT}] \propto  \delta[(\partial_t\bphi_t-\bB \bj_t)_t]\exp\left\{\frac{-1}{4}\int_0^\mcT (\bj_t -\bJ_t)\left[\bbb \bbb^\dagger\right]^{-1}  (\bj_t -\bJ_t) \rmd t \right\} \ .
\end{equation}

We can finally deduce the IEPR associated to the pair $(\bphi,\bj)$ of variables:
\begin{equation}
\sigma_{\bphi,\bj} \equiv \lim_{\mcT\to\infty} \frac{1}{\mcT}\ln \frac{\mcP[(\bphi_t,\bj_t)_{t\in\mbT}]}{\mcP[(\bphi_{\mcT-t},-\bj_{\mcT-t})_{t\in\mbT}]} = \llangle \int  \bj \left[\bbb \bbb^\dagger\right]^{-1} \bJ \rmd \br \rrangle \ .
\end{equation}


\subsection{Hidden current}
\label{appsubsec:hiddenCurrent}

In this section we prove the decomposition~\eqref{Ex_decomp_current} of the space of currents in the case of section~\ref{subsubsec:ex_conserved_scalar}:
\begin{equation}
\mbJ_{\rho} = \IM ( -\bM\nabla ) \oplus \ker(\nabla\cdot) \ ,
\label{target_decomposition}
\end{equation}
before briefly discussing the general decomposition~\eqref{target_split} of section~\ref{subsubsec:hidden_general_case}.

Let us start by recalling that, according to Helmholtz-Hodge theory~\cite{warner1983foundations}, if $g_{ij}$ is a Riemannian metric on $\mbR^{d_1}$, then the space of (square-integrable) vector fields over this space (that we here assimilate to $\mbJ_\rho$, for a fixed $\rho$) can be decomposed as the superposition of the image of the Riemannian gradient and the kernel of the corresponding divergence. This means that, for a given vector field $\bu(\br)$ over $\mbR^{d_1}$, there exists a scalar potential $\mu(\br)$ and a vector field $\bgamma(\br)$ whose Riemannian divergence vanishes, \ie $\nu^{-1}\partial_i ( \nu \gamma^i)=0$, with $\nu=\sqrt{\mathrm{det}(g_{ij})}$ the corresponding Riemannian volume element, which are such that
\begin{equation}
u^i = -g^{ij}\partial_j\mu + \gamma^i \ ,
\label{intermed_Hodge_decomp}
\end{equation}
where $g^{ij}$ is the dual metric, satisfying $g^{ij}g_{jk}=\delta^i_k$.

Let us now suppose that $d_1\neq 2$ and define the metric tensor 
\begin{equation}
g_{ij}= (\Det \bM)^{1/(d_1-2)} [M^{-1}]_{ij} \ .
\label{Riemann_trick}
\end{equation}
Then, for a given current $\bJ(\br)$ over $\mbR^{d_1}$, we set 
\begin{equation}
\bu\equiv \bJ/\nu \ ,
\label{intermed_vector_field}
\end{equation}
where $\nu$ is the Riemaniann volume element associated to the metric~\eqref{Riemann_trick}.
Hence, there exist a scalar potential $\mu$ and a divergence-free vector field $\bgamma$ such that decomposition~\eqref{intermed_Hodge_decomp} holds for $\bu$ given by~\eqref{intermed_vector_field}.
Then, note that 
\begin{eqnarray*}
\nu &\equiv & \sqrt{\Det g_{ij}}  \\
&=& \left\{ \Det\left[ (\Det \bM)^{1/(d_1-2)} [M^{-1}]_{ij} \right] \right\}^{1/2} \\
&=& \left\{ (\Det \bM)^{d_1/(d_1-2)} (\Det \bM)^{-1} \right\}^{1/2} \\
&=& (\Det \bM)^{1/(d_1-2)} \ ,
\end{eqnarray*}
while the dual metric reads
\begin{equation}
g^{ij}= (\Det \bM)^{-1/(d_1-2)} M^{ij} \ .
\end{equation}
Thus $\bJ = \nu \bu$ splits as
\begin{equation}
J^i  = -M^{ij}\partial_j \mu + \tilde{\gamma}^i \ ,
\end{equation}
where $\tilde{\bgamma}\equiv \nu \bgamma$ satisfies $\partial_i \tilde{\gamma}^i = \partial_i(\nu\gamma^i) = 0$, which is what we aimed at proving.

Let us now briefly discuss the decomposition
\begin{equation}
\mbJ_{\bphi} = \IM ( \bbb_{\bphi}\bbb^\dagger_{\bphi} \bB^\dagger ) \oplus \ker(\bB) \ ,
\label{target_decomposition}
\end{equation}
that we assumed to hold in the general setting of section~\ref{subsubsec:hidden_general_case}.
We first denote by $(\cdot | \cdot)^0_{T_{\bphi}\mbF}$ and $(\cdot | \cdot)^0_{\mbJ_{\bphi}}$ the canonical $L^2$-scalar product on $T_{\bphi}\mbF$ and $\mbJ_{\bphi}$, respectively. Then we define another scalar product on the latter space:
\begin{equation}
\left(\bj_1|\bj_2\right)^1_{\mbJ_{\bphi}} \equiv \Big(\bj_1 \left| \left[ \bbb_{\bphi}\bbb_{\bphi}^\dagger\right]^{-1} \right. \bj_2\Big)^0_{\mbJ_{\bphi}} = \int  \bj_1\left[ \bbb_{\bphi}\bbb_{\bphi}^\dagger\right]^{-1} \bj_2 \rmd \br \ .
\end{equation}
Then, it turns out $\bbb_{\bphi}\bbb_{\bphi}^\dagger\bB^\dagger$ coincides with the adjoint of $\bB$ with respect to the scalar products $(\cdot | \cdot)^0_{T_{\bphi}\mbF}$ in $T_{\bphi}\mbF$ and $(\cdot | \cdot)^1_{\mbJ_{\bphi}}$ in $\mbJ_{\bphi}$, which we denote by $\bB^\ast$. Indeed, for any $\bj\in\mbJ_{\bphi}$ and $\delta\bphi\in T_{\bphi}\mbF$, using the definition of an adjoint operator, we have
\begin{equation}
(\bj | \bB^\ast \delta\bphi )_{\mbJ_{\bphi}}^1 = (\bB\bj | \delta\bphi )^0_{T_{\bphi}\mbF}
= (\bj | \bB^\dagger \delta\bphi )_{\mbJ_{\bphi}}^0 
= (\bj | \bbb_{\bphi}\bbb_{\bphi}^\dagger \bB^\dagger \delta\bphi )_{\mbJ_{\bphi}}^1 \ .
\end{equation}
Thus, the splitting~\eqref{target_decomposition} is simply a decomposition of a space into the superposition of the kernel of an operator and the image of its adjoint, and can consequently be expected to occur rather generically, although a formal exploration of the conditions for it lies beyond our scope here.

\bibliographystyle{ieeetr}
\bibliography{../../../Biblio}


\end{document}